\documentclass[11pt,reqno]{amsart}

\usepackage[T1]{fontenc}
\usepackage[utf8]{inputenc}
\usepackage[english]{babel}

\usepackage{lmodern}

\usepackage[a4paper,margin=1.15in]{geometry}
\usepackage{microtype}
\usepackage{mathtools}
\usepackage{amssymb}
\usepackage{bm}
\usepackage{braket}
\usepackage{extpfeil}
\allowdisplaybreaks
\numberwithin{equation}{section}

\usepackage{array}
\usepackage{tabularx}

\usepackage[shortlabels]{enumitem}
\usepackage[title]{appendix}
\usepackage[all,poly]{xy}
\usepackage{tikz}
\usetikzlibrary{arrows.meta,positioning}

\usepackage{csquotes}

\usepackage{graphicx} 
\usetikzlibrary{arrows.meta,calc}

\definecolor{figInk}{HTML}{253449}
\definecolor{figBlue}{HTML}{3A7096}
\definecolor{figPale}{HTML}{EAF3F9}
\definecolor{figWarm}{HTML}{F4DFBD}
\definecolor{figRed}{HTML}{A24C40}
\definecolor{figGuide}{HTML}{ACB5BE}

\usepackage{xcolor}
\definecolor{linkblue}{RGB}{0,72,128}
\usepackage[
colorlinks=true,
linkcolor=linkblue,
citecolor=linkblue,
urlcolor=linkblue
]{hyperref}
\usepackage[nameinlink,noabbrev]{cleveref}

\theoremstyle{plain}
\newtheorem{thm}{Theorem}[section]
\newtheorem{theorem}[thm]{Theorem}
\newtheorem{lemma}[thm]{Lemma}

\newtheorem{corollary*}[thm]{Corollary*}
\newtheorem{proposition}[thm]{Proposition}
\newtheorem{proposition*}[thm]{Proposition*}

\theoremstyle{definition}

\newtheorem{defn-thm}[thm]{Definition-Theorem}

\theoremstyle{remark}
\newtheorem{remark}[thm]{Remark}

\newcommand{\Rmnum}[1]{\uppercase\expandafter{\romannumeral #1}}

\tikzset{
	solid black dot/.style={
		circle,
		fill=black,
		inner sep=1.5pt
	}
}

\def\d{\mathrm{d}}

\usepackage{color}

\definecolor{A}{rgb}{.75,1,.75}

\definecolor{yellow}{rgb}{1,1,0}
\definecolor{orange}{rgb}{1,.7,0}
\definecolor{red}{rgb}{1,0,0}
\definecolor{white}{rgb}{1,1,1}
\definecolor{green}{rgb}{0,1,0}

\begin{document}
	
	\title[Spectral curves for Frobenius manifolds]
        {Spectral Curves for Semisimple Conformal Frobenius Manifolds}
	\author{Dingkang Xiong}
	\address{Department of Mathematical Sciences, Tsinghua University, Beijing 100084, China}
	\email{xdk22@mails.tsinghua.edu.cn}
	\author{Jian Zhou}
	\address{Department of Mathematical Sciences, Tsinghua University, Beijing 100084, China}
	\email{jianzhou@mail.tsinghua.edu.cn}
	
	\begin{abstract}
        We construct connected spectral curves for semisimple conformal Frobenius manifolds with flat unit, using the periods of the second structure equation. 
        Eynard--Orantin topological recursion on these curves reproduces the stable ancestor correlators of Givental--Teleman reconstruction.
        The spectral curves are defined on simply connected domains in the complex plane. 
        We apply the construction to the quantum cohomology of $\mathbb P^n$ for $n\ge2$ at a point, obtaining explicit spectral curve data in terms of Hecke triangle groups and hypergeometric functions.
        Finally, for $\mathbb P^2$ and $\mathbb P^3$, we express the complete spectral curve data, including the bidifferentials, in classical modular and quasimodular forms.
	\end{abstract}
	\maketitle
	
	\setcounter{tocdepth}{1}
	\tableofcontents

\section{Introduction}
\label{sec:introduction}

Understanding how genus-zero geometry controls higher-genus
invariants is a central problem in Gromov--Witten theory.
Frobenius manifolds organize the genus-zero structures of the
theory into a family of Frobenius algebras compatible
with a flat metric, bringing together quantum cohomology, singularity theory and integrable systems
\cite{dubrovinGeometry2DTopological1996,
kontsevichGromovWittenClassesQuantum1994,dubrovinNormalFormsHierarchies2001}.
Semisimplicity of the Frobenius algebras makes a particularly strong reconstruction principle possible. 
Givental's quantization formalism reconstructs higher-genus
ancestor potentials from copies of the theory of a point, using the $R$-matrix at a chosen semisimple point of the Frobenius manifold
\cite{giventalSemisimpleFrobeniusStructures2001,
giventalGromovWittenInvariantsQuantization2001}.
Teleman's classification identifies this reconstruction with
semisimple cohomological field theories and, in the homogeneous
case with flat unit, determines the higher-genus theory from
genus-zero data \cite{telemanStructure2DSemisimple2012}.
Besides providing a structural description, this reduces the
calculation of higher-genus invariants to understanding the
genus-zero differential equations and their solutions.

Eynard--Orantin topological recursion approaches reconstruction
through the geometry of a curve \cite{eynardInvariantsAlgebraicCurves2007}. Starting from a Riemann surface
$\Sigma$, two functions $x,y$, and a symmetric bidifferential $B$,
it produces (multi)differentials $\omega_{g,n}$ by residues at the simple
critical points of $x$.
Originally developed in the study of matrix models \cite{chekhovHermiteanMatrixModel2006,chekhovFreeEnergyTopological2006}, this residue
formalism provides a common language for a wide range of
enumerative problems. Its relation to semisimple reconstruction
was established by Dunin-Barkowski, Orantin, Shadrin and Spitz
(DOSS), who identified the two procedures for compatible local
spectral curve data
\cite{dunin-barkowskiIdentificationGiventalFormula2014};
see also \cite{milanovEynardOrantinRecursion2014,
xiongEYNARDORANTINTOPOLOGICALRECURSION,chekhovTopologicalRecursionHard2019}.
The DOSS correspondence explains how the $R$-matrix and the unit
are encoded in local expansions of $B$ and $y$.
It therefore suggests a geometric form of reconstruction:
realize the prescribed local data on a connected surface, so that one residue recursion recovers the entire
stable ancestor theory. The existence of such a surface 
requires additional work beyond the local correspondence.

The projective line $\mathbb P^1$ illustrates how these different structures meet in a concrete enumerative theory.
The Gromov--Witten theory of $\mathbb P^1$ has close connections
with Toda integrable hierarchies. Its non-equivariant descendant
theory is governed by the extended Toda hierarchy of
Carlet--Dubrovin--Zhang \cite{carletExtendedTodaHierarchy2004}; see also \cite{getzlerTODACONJECTURE2001,zhangCP1TopologicalSigma2002,dubrovinVirasoroSymmetriesExtended2004}.
In the equivariant setting, Okounkov and Pandharipande identified the descendant partition function with a tau-function of the $2$-Toda hierarchy \cite{okounkovEquivariantGromovWitten2006}.
Okounkov and Pandharipande also established a correspondence between stationary Gromov--Witten theory of curves and Hurwitz theory with completed cycles \cite{okounkovGromovWittenTheory2006,okounkovGromovWittenTheoryHurwitz2009}.
Norbury and Scott proposed the spectral curve
\begin{equation}
    \Sigma=\mathbb C_z,\quad
    x(z)=z+\frac1z,
    \qquad y(z)=\log z,
    \qquad B(z_1,z_2)=\frac{dz_1\,dz_2}{(z_1-z_2)^2}
    \label{eq:intro-P1-curve}
\end{equation}
for the stationary theory of $\mathbb P^1$
\cite{norburyGromovWittenInvariants2014}.
The all-genus identification was proved by DOSS
\cite{dunin-barkowskiIdentificationGiventalFormula2014};
Borot and Norbury subsequently recovered non-stationary invariants
from the same recursion differentials \cite{borotLoopEquationsGromovWitten2019}.
Thus a simple one-dimensional spectral curve encodes Gromov--Witten invariants in all genera. 
Extending this picture to higher-dimensional targets raises the problem of finding spectral curve data reproducing the relevant CohFT structure.

In this article, we construct spectral curves $(\Sigma,x,y,B)$ for semisimple conformal Frobenius manifolds with flat unit. Each curve is defined on a simply connected domain in the complex plane.

Our construction is relative to Dubrovin's description of
one-dimensional Landau--Ginzburg superpotentials through periods
of the second structure equation \cite{dubrovinGeometry2DTopological1996,
dubrovinPainleveTranscendentsTwoDimensional1999,dubrovinAlmostDualityFrobenius2004}.
Dunin-Barkowski, Norbury, Orantin, Popolitov and Shadrin (DNOPS)
established a correspondence between these superpotentials and
global spectral curves under additional geometric and
compatibility assumptions
\cite{dunin-barkowskiDubrovinsSuperpotentialGlobal2019}.
Their approach has two difficulties in applications.
First, their superpotential surface is required to be compact
so that bidifferential $B$ is   the Bergman kernel
normalized to have vanishing periods along a chosen basis of
$\mathcal A$-cycles. 
Sometimes it is not easy to see the superpotential curve can be compactified.
Secondly, they impose some conditions for the compatibility of the Bergman kernel
with the superpotential.
These conditions are not automatically satisfied for a
general semisimple conformal Frobenius manifold.
In fact, the results of this article suggest that, in some cases, the superpotential curve may not admit a compactification.

To bypass these difficulties, we are inspired by Milanov's
period description of topological recursion for the total ancestor potential
\cite{milanovEynardOrantinRecursion2014}
and its global extension for simple singularities
\cite{milanovEynardOrantinRecursion2015}. 
We work on connected open Riemann surfaces, where $B$ is constructed to realize the prescribed $R$-matrix data and to have only its normalized diagonal double
pole, without imposing an $\mathcal A$-period normalization as the compact case.
This leads to a general connected spectral curve
realization for semisimple conformal Frobenius manifolds with
flat unit. When a simultaneously normalized common period of the second structure equation is given, the construction admits an explicit expression in
terms of that period.

For projective spaces $\mathbb{P}^n$ with $n\geq 2$, this yields explicit spectral curve data, and for $\mathbb P^2$ and $\mathbb P^3$ it
gives classical modular expressions for the spectral functions $(x,y)$, and the bidifferential $B$.

We note that, for the projective plane, Milanov has obtained a modular description of the quantum cohomology period map \cite{milanovPeriodMapQuantum2019}.
Related one-dimensional Landau--Ginzburg descriptions of its small and big quantum cohomology were studied by Almeida
\cite{almeida1DLandauGinzburgSuperpotential2025}.

\medskip
\noindent\textit{Main results of this article.}
The first result of this article is a spectral curve realization for a general semisimple conformal Frobenius manifold.
Let $N$ be its dimension, $v$ a semisimple point, $u_i$ the
canonical values, and $f_i$ the normalized canonical vectors.
The homogeneous $R$-matrix determines the ancestor correlators
denoted by $\langle\cdots\rangle^R_{g,r;v}$.
The following summarizes Theorem~\ref{thm:connected-spectral-realization}
and its extension in Appendix~\ref{sec:non-tame}.

\begin{theorem} \label{thm:intro-connected-realization}
At every semisimple point of a conformal Frobenius manifold with flat unit, there exists a spectral curve
$(\Sigma,x,y,B)$ with $\Sigma$ simply connected domain.
The function $x$ has exactly $N$ distinct simple critical points $a_i$, with $x(a_i)=u_i$ and $dy(a_i)\ne0$.
The bidifferential $B$ has only its normalized diagonal double pole.
For $r\ge1$ and $2g-2+r>0$, the Eynard--Orantin topological recursion gives
\begin{equation}
\begin{aligned}
    \omega_{g,r}(p_1,\ldots,p_r)
    ={}&\sum_{i_1,\ldots,i_r=1}^{N}
    \sum_{k_1,\ldots,k_r\ge0}
    \left\langle
       \prod_{\nu=1}^{r}\bar\tau_{k_\nu}(f_{i_\nu})
    \right\rangle^R_{g,r;v}
    \prod_{\nu=1}^{r}W_{k_\nu}^{i_\nu}(p_\nu),
\end{aligned}
\label{eq:intro-ancestor-realization}
\end{equation}
where the external one-forms $W_k^i$ are determined by $B$ and $R$.
\end{theorem}

Here $\bar\tau_k$ denotes an ancestor insertion. The forms in
\eqref{eq:intro-ancestor-realization} are defined by $W=R\Xi$,
where $\Xi$ is obtained by local coefficient extraction from
the bidifferential $B$, as in
\eqref{eq:DOSS-coefficient-forms}--\eqref{eq:DOSS-ancestor-forms}.
We construct the curve from individually normalized
local periods of the second structure equation and realize the required globalization by Cauchy
integrals on a connected open surface.
There is some freedom in the construction of the spectral curve.

A more explicit construction is available when one period has
the required singular part at every canonical value, as in the
superpotential constructions of Dubrovin and DNOPS
\cite{dubrovinPainleveTranscendentsTwoDimensional1999,
dunin-barkowskiDubrovinsSuperpotentialGlobal2019}.
Write $\eta$ for the Frobenius flat metric, $e$ for the unit, $E$ for
the Euler field, and $d$ for the conformal charge. In our conventions,
$A=E\circ-$ and
$\theta=\nabla E-(1-d/2)\operatorname{Id}$.
The second structure equation and the simultaneous normalized period solution $J$ are
\begin{equation}
\begin{gathered}
    (x\operatorname{Id}-A)\partial_xJ
    =\left(\theta-\frac12\operatorname{Id}\right)J,\\
    J(s_i)=\frac{f_i}{s_i}+O(1).
\end{gathered}
\label{eq:intro-common-period}
\end{equation}
Here $v$ is a fixed semisimple point, $J=J(x)$ is a $T_vM$-valued function of $x$, and $\{f_i\}$ is the normalized canonical frame. Near the $i$-th critical point, we write $x=u_i+s_i^2/2$.
The flat-pencil pairing
$\langle J,K\rangle_x=\eta(J,(A-x\operatorname{Id})K)$
is constant on solutions of the second structure equation.
For $x_\nu=x(q_\nu)$, the corresponding period bidifferential is defined by
\begin{equation}
\begin{aligned}
    B^J(q_1,q_2)
    ={}&\frac{dx_1\,dx_2}{(x_1-x_2)^2}
    \Bigg\{1+\langle J,J\rangle_x\\
    &\quad+\eta\!\left(J(q_1),
    \left[(x_1-x_2)\theta
       +\frac{x_1+x_2}{2}\operatorname{Id}-A\right]J(q_2)\right)
    \Bigg\}.
\end{aligned}
\label{eq:intro-period-kernel}
\end{equation}
Theorem~\ref{thm:common-period-spectral-realization} gives the
following refinement.

\begin{theorem}\label{thm:intro-common-period}
At a tame semisimple point, suppose that $J$ satisfies
\eqref{eq:intro-common-period}.
There is a connected, simply connected domain $\Sigma$ on which
$dy^J=\eta(e,J)\,dx$ has a holomorphic primitive and
$(\Sigma,x,y^J,B^J)$ is a spectral curve.
Its Eynard--Orantin topological recursion differentials satisfy
\eqref{eq:intro-ancestor-realization}, with
\begin{equation}
    W_k^i(q)=(-1)^{k+1}d_q\,
       \eta\bigl(f_i,\partial_x^kJ(q)\bigr),
    \qquad k\ge0.
    \label{eq:intro-period-ancestor-forms}
\end{equation}
\end{theorem}

We also discuss the existence and ambiguity of the common period $J$ through the Gram matrix of the
flat-pencil pairing. 

We then apply this result to the quantum cohomology of
$\mathbb P^n$, $n\ge2$, at the small quantum point $Q=1$, $t=0$.
Here the construction can be made explicit at each step, from
the second structure equation to the full bidifferential and a
domain supporting the recursion.
The second structure equation reduces to the scalar period equation
\begin{equation}
    \left[
       \left(x\partial_x+\frac{n-1}{2}\right)^{n+1}
       -(n+1)^{n+1}\partial_x^{n+1}
    \right]y=0.
    \label{eq:intro-Pn-scalar-equation}
\end{equation}
The solution corresponding to the period $J$ in \eqref{eq:intro-common-period} selected here is
\begin{equation}
    \begin{aligned}
    y(x)&=-(2\pi)^{\frac{n-1}{2}}
       \Gamma\!\left(\frac{n-1}{2}\right)
       x^{-\frac{n-1}{2}}
       \sum_{d\ge0}
       \frac{\left(\frac{n-1}{2}\right)_{(n+1)d}}
            {(d!)^{n+1}}x^{-(n+1)d}\\
        &=-(2\pi)^{\frac{n-1}{2}}\Gamma(\frac{n-1}{2})x^{-\frac{n-1}{2}}
        \cdot {}_{n+1}F_{n}\left(
        \begin{matrix}
            \dfrac{n-1}{2n+2},\dfrac{n+1}{2n+2},\ldots,\dfrac{3n-1}{2n+2}\\
            1,\ldots,1
        \end{matrix}
        ;\frac{(n+1)^{n+1}}{x^{n+1}}\right).
    \end{aligned}
    \label{eq:intro-Pn-spectral-functions}
\end{equation}
The subscripted parentheses denote rising factorials.
This is a generalized hypergeometric period, initially
convergent for $|x|>n+1$, with the branch chosen to be real
and negative for $x>n+1$.
Its normalization is fixed by the local data at the
canonical values. 
More precisely, for
\begin{equation}
    u_j=(n+1)e^{\frac{2\pi\mathrm i j}{n+1}},
    \qquad x=u_j+\frac{s_j^2}{2},
    \qquad 0\le j\le n,
    \label{eq:intro-Pn-canonical-values}
\end{equation}
we prove, on compatible continuation branches, that
\begin{equation}
    \left.\frac{dy}{ds_j}\right|_{s_j=0}
       =\frac1{\sqrt{n+1}}
          e^{-\frac{\pi\mathrm i n j}{n+1}}=h_j,
    \label{eq:intro-Pn-canonical-slopes}
\end{equation}
where $h_j=\eta(e,f_j)$ is the corresponding component of the unit $e$ in the normalized canonical frame. 

The scalar solution also determines the entire vector period $J$.
Writing $H$ for the hyperplane class and
$J=\sum_{\alpha=0}^{n}j_\alpha H^\alpha$, the first two
components obtained from $y$ are
\begin{equation}
    j_n=\partial_xy,
    \qquad
    j_{n-1}=\frac{x\partial_xy+\frac{n-1}{2}y}{n+1}.
    \label{eq:intro-Pn-top-components}
\end{equation}
Proposition~\ref{prop:Pn-differential-reconstruction} gives
a finite differential reconstruction: every component of the selected period $J$ is a linear combination of derivatives of $y$ of order at most $n$, with polynomial coefficients in $x$.
Consequently, both the bidifferential
\eqref{eq:intro-period-kernel} and the ancestor forms
\eqref{eq:intro-period-ancestor-forms} are explicitly determined
by one scalar hypergeometric solution $y$.
The reconstructed period satisfies $\langle J,J\rangle_x=0$.


To describe the surface $\Sigma$ explicitly, we use the
Hauptmodul $\mathcal J_{n+1}$ of the Hecke triangle group of
signature $(2,n+1,\infty)$, normalized to take the values
$1$, $0$ and $\infty$ at the order-two elliptic point, the
order-$(n+1)$ elliptic point and the cusp, respectively; see \cite{doranAutomorphicFormsTriangle2013}.
The $x$-function is determined by
\begin{equation}
    x(\tau)^{n+1}=(n+1)^{n+1}\mathcal J_{n+1}(\tau),
    \qquad x(\mathrm i)=n+1,
    \quad \tau\in\mathbb{H}.
    \label{eq:intro-Pn-Hecke-x}
\end{equation}
Here $\mathrm{i}$ is the imaginary unit. The resulting spectral curve statement is Theorem~\ref{thm:Pn-explicit-spectral-curve}.

\begin{theorem}\label{thm:intro-Pn-spectral-curve}
Let $n\ge2$ and choose
\begin{equation}
    \max\left\{\sin\frac{\pi}{n+1},
                   \cos\frac{\pi}{n+1}\right\}<b<1.
    \label{eq:intro-Pn-domain-height}
\end{equation}
On the simply connected domain
\begin{equation}
    \Sigma=\left\{\tau\in\mathbb H:
       -(2n+1)\cos\frac{\pi}{n+1}
       <\Re\tau<\cos\frac{\pi}{n+1},\quad
       \Im\tau>b\right\},
    \label{eq:intro-Pn-domain}
\end{equation}
the functions in \eqref{eq:intro-Pn-spectral-functions} and
\eqref{eq:intro-Pn-Hecke-x}, continued from the infinity, together
with the bidifferential $B^J$, define a spectral
curve $(\Sigma,x,y,B^J)$. Its critical points are exactly
\begin{equation}
    a_j=-2j\cos\frac{\pi}{n+1}+\mathrm i,
    \qquad x(a_j)=u_j,\qquad 0\le j\le n,
    \label{eq:intro-Pn-critical-points}
\end{equation}
and they are simple. 
The Eynard--Orantin topological recursion
satisfies \eqref{eq:intro-ancestor-realization}, 
with the Gromov--Witten ancestor correlators of $\mathbb P^n$ at
$Q=1$, $t=0$ and the forms
\eqref{eq:intro-period-ancestor-forms}.
\end{theorem}


We also describe a natural algebraic curve, which is a double cover of the
$x$-plane; see
Subsection~\ref{subsec:Pn-quotient-geometry}.
It has the equation
\begin{equation}
    w^2=x^{n+1}-(n+1)^{n+1},
    \label{eq:intro-Pn-algebraic-cover}
\end{equation}
whose smooth compactification has genus $g=\lfloor n/2\rfloor$.
This describes a compact surface for
$x$-function. 
However, the algebraic coordinate $w$
is distinct from the $y$-function. The function $y$ and the bidifferential $B^J$ need not descend to this compact surface.

Finally, the uniformization $\tau$ gives automorphic expressions for
the $x$-function in every dimension; see Subsection~\ref{subsec:Pn-modular-spectral-data}. For $n=2,3,5$,
the corresponding Hecke groups of $\mathcal{J}_{n+1}$ are arithmetic, and the formulas can be
written in classical modular forms. In particular, we obtain
the full modular spectral curve data for $\mathbb P^2$ and
$\mathbb P^3$, including every component of $J$ and $B^J$.
Write $\eta_{\mathrm D}$ for the Dedekind eta function and
$E_2,E_4$ for the normalized Eisenstein series; $E_2$ is
quasimodular. Propositions~\ref{prop:P2-modular-spectral-data}
and~\ref{prop:P3-modular-spectral-data} give
\begin{equation}
\begin{aligned}
    \mathbb P^2:\qquad
    x(\tau)&=\frac{E_4(\tau)}{4\eta_{\mathrm D}(\tau)^8},
    &y(\tau)&=-2\sqrt2\,\pi\eta_{\mathrm D}(\tau)^4,\\
    \mathbb P^3:\qquad
    x(\sqrt2t)&=\frac{2E_2(2t)-E_2(t)}
       {\eta_{\mathrm D}(t)^2\eta_{\mathrm D}(2t)^2},
    &y(\sqrt2t)&=-2\pi\eta_{\mathrm D}(t)^2
       \eta_{\mathrm D}(2t)^2,
\end{aligned}
\label{eq:intro-modular-spectral-functions}
\end{equation}
where $t=\tau/\sqrt2$ in the second line.
The complete bidifferentials $B^J$ are written in
\eqref{eq:P2-full-modular-bidifferential} and
\eqref{eq:P3-full-modular-bidifferential}.
These formulas also include the full vector periods $J$ in classical modular and quasimodular forms.
We also give a classical modular expression for the $x$-function of $\mathbb P^5$:
\begin{equation}
    x(\sqrt3t)
    =\frac6{108^{1/6}}
    \frac{\vartheta_{A_2}(t)}
    {\eta_{\mathrm D}(t)\eta_{\mathrm D}(3t)},
    \quad 
    \vartheta_{A_2}(t)
    =\sum_{r,s\in\mathbb Z}e^{2\pi\mathrm i t(r^2+rs+s^2)},
\end{equation} 
where $t=\tau/\sqrt{3}$ and $\vartheta_{A_2}$ is the theta function for the $A_2$ lattice.

The remainder of this paper is organized as follows.
Section~\ref{sec:backgrounds} fixes our conventions for Frobenius manifolds, Givental--Teleman reconstruction and Eynard--Orantin topological recursion.
Section~\ref{sec:spectral-construction} proves the general
construction of the spectral curve and its refinement using a vector period $J$.
Section~\ref{sec:projective-space} constructs the periods and
spectral curves of projective spaces and describes their
upper half-plane uniformization.
Section~\ref{sec:Pn-modular-spectral-data} derives the automorphic
expressions for the function $x$ and the classical modular formulas for $\mathbb P^2$, $\mathbb P^3$ and $\mathbb P^5$.
Appendix~\ref{sec:non-tame} treats coincident canonical values
at semisimple points.

\vspace{.1in}
{\bf Declaration of AI Assistance}
The first-named author had a beneficial discussion with AI during the development of the mathematical constructions presented in this
paper. AI tools were also used to polish the language of the manuscript. 
The authors take responsibility for all results, proofs, and calculations of the article.

\section{Background: Frobenius manifolds, CohFTs and topological recursion}\label{sec:backgrounds}
This section recalls the background and fixes the conventions used in our spectral curve construction. We review semisimple
conformal Frobenius manifolds and their homogeneous $R$-matrix calibrations,
state the Givental--Teleman ancestor reconstruction formula,
and describe Eynard--Orantin topological recursion and its
relation to reconstruction through the DOSS correspondence.

\subsection{Semisimple conformal Frobenius manifolds}
\label{subsec:frobenius}
In this subsection, we recall the basic definitions and conventions for Frobenius
manifolds, following
\cite{dubrovinGeometry2DTopological1996,giventalGromovWittenInvariantsQuantization2001,milanovPeriodMapQuantum2019}.

Let $M$ be a complex analytic manifold of dimension $N$.
A \emph{Frobenius structure} on $M$ consists of a holomorphic,
commutative and associative multiplication $\circ$ on $TM$,
a unit vector field $e$, and a nondegenerate symmetric holomorphic
metric $\eta$. The metric is flat, its Levi--Civita connection
$\nabla$ satisfies $\nabla e=0$, and
\begin{equation}
    \eta(X\circ Y,Z)=\eta(X,Y\circ Z)
\end{equation}
for all vector fields $X,Y,Z$. Moreover, the covariant derivative $\nabla c$
of the symmetric tensor
\begin{equation*}
    c(X,Y,Z)=\eta(X\circ Y,Z)
\end{equation*}
is totally symmetric.

Choose local flat coordinates $t^1,\ldots,t^N$ such that
$e=\partial_1$, denoting $\partial_\alpha=\partial/\partial t^\alpha$ and 
$ \eta_{\alpha\beta}
    =\eta(\partial_\alpha,\partial_\beta)$.
The matrix $(\eta_{\alpha\beta})$ is constant, and its inverse
is denoted by $(\eta^{\alpha\beta})$.
Throughout, we use the Einstein summation convention. 

The Frobenius condition implies the local existence of a
holomorphic function $F$, called the Frobenius potential, determined up to a polynomial of degree at most two, such that
\begin{equation}
    \eta(\partial_\alpha\circ\partial_\beta,\partial_\gamma)
    =F_{\alpha\beta\gamma},
    \qquad
    F_{\alpha\beta\gamma}
    =\partial_\alpha\partial_\beta\partial_\gamma F.
\end{equation}
Thus
\begin{equation}
    \partial_\alpha\circ\partial_\beta
    =
    \eta^{\gamma\delta}F_{\alpha\beta\delta}\,
    \partial_\gamma.
\end{equation}

The Frobenius manifold is called \emph{conformal of charge $d$}
if it admits an Euler vector field $E$ satisfying
\begin{equation}
    \nabla(\nabla E)=0,
    \qquad
    \mathcal L_E(\circ)=\circ,
    \qquad
    \mathcal L_E\eta=(2-d)\eta.
    \label{eq:frobenius-conformality}
\end{equation}
The homogeneity condition on the multiplication means
\begin{equation}
    [E,X\circ Y]-[E,X]\circ Y-X\circ[E,Y]
    =X\circ Y.
\end{equation}
In particular, $[E,e]=-e$.
The first condition in \eqref{eq:frobenius-conformality}
says that $E$ is affine in flat coordinates, while the metric
condition gives
\begin{equation}
    \eta(\nabla_XE,Y)+\eta(X,\nabla_YE)
    =(2-d)\eta(X,Y).
\end{equation}

We use the endomorphisms
\begin{equation}
    A(X):=E\circ X,
    \qquad
    \theta
    :=
    \nabla E-\left(1-\frac d2\right)\operatorname{Id},
    \label{eq:frobenius-A-theta}
\end{equation}
where $(\nabla E)(X)=\nabla_XE$.
Metric invariance and conformality imply
\begin{equation}
    \eta(AX,Y)=\eta(X,AY),
    \qquad
    \eta(\theta X,Y)=-\eta(X,\theta Y),
    \qquad
    \nabla\theta=0.
    \label{eq:frobenius-adjointness}
\end{equation}
Since $\nabla e=0$ and $[E,e]=-e$, we also have
\begin{equation}
    \nabla_eE=e,
    \qquad
    \theta(e)=\frac d2e.
    \label{eq:frobenius-grading-unit}
\end{equation}
Our convention for $\theta$ is the negative of the grading
operator often denoted by $\mu$.
We impose no diagonalizability assumption on $\theta$.

A point $t\in M$ is \emph{semisimple} if the algebra
$(T_tM,\circ,e)$ is isomorphic to $\mathbb C^N$.
We work in a sufficiently small, simply connected neighbourhood
of a semisimple point.
After fixing labels, the primitive idempotents of the algebra $T_tM$ extend to
holomorphic vector fields $\pi_1,\ldots,\pi_N$ satisfying
\begin{equation}
    \pi_i\circ\pi_j=\delta_{ij}\pi_i,
    \qquad
    e=\sum_{i=1}^{N}\pi_i.
    \label{eq:frobenius-idempotents}
\end{equation}
They commute as vector fields, and the Euler field determines canonical
coordinates $u_1,\ldots,u_N$ such that
\begin{equation}
    \pi_i=\frac{\partial}{\partial u_i},
    \qquad
    E=\sum_{i=1}^{N}u_i\pi_i,
    \qquad
    A\pi_i=u_i\pi_i.
    \label{eq:frobenius-canonical-coordinates}
\end{equation}
For the construction of canonical coordinates and the normalized
idempotent frame, see
\cite{dubrovinGeometry2DTopological1996}.

Semisimplicity does not require the canonical values $u_i(t)$
to be pairwise distinct.
A semisimple point is called \emph{tame} if they are distinct.
We assume tameness, but our construction of the spectral curve can be extended to the non-tame case. The details are given in Appendix~\ref{sec:non-tame}.

Metric invariance implies
\begin{equation*}
    \eta(\pi_i,\pi_j)=0
    \qquad (i\ne j).
\end{equation*}
Choose holomorphic square roots and define the normalized
canonical frame by
\begin{equation}
    h_i^2=\eta(\pi_i,\pi_i),
    \qquad
    f_i=\frac{\pi_i}{h_i}.
    \label{eq:frobenius-normalized-frame}
\end{equation}
Nondegeneracy implies $h_i\ne0$.
We then have
\begin{equation}
    \eta(f_i,f_j)=\delta_{ij},
    \qquad
    h_i=\eta(e,f_i),
    \qquad
    e=\sum_{i=1}^{N}h_i f_i,
    \label{eq:frobenius-unit-components}
\end{equation}
and
\begin{equation}
    f_i\circ f_j=\delta_{ij}\frac{f_i}{h_i},
    \qquad
    \eta=\sum_{i=1}^{N}h_i^2\,du_i\otimes du_i.
    \label{eq:frobenius-canonical-metric}
\end{equation}
The labels and square-root choices will remain fixed throughout
each local construction.

Let $\Phi$ be the transition matrix defined by
\begin{equation}
    f_i=\Phi^\alpha{}_i\,\partial_\alpha,
\end{equation}
and set
\begin{equation}
    U=\operatorname{diag}(u_1,\ldots,u_N),
    \qquad
    \Theta=\Phi^{-1}\theta\Phi,
    \qquad
    h=(h_1,\ldots,h_N)^T.
    \label{eq:frobenius-canonical-matrices}
\end{equation}
Then
\begin{equation}
    \Phi^T\eta\Phi=\operatorname{Id},
    \qquad
    A\Phi=\Phi U,
    \qquad
    \Theta^T=-\Theta.
    \label{eq:frobenius-frame-identities}
\end{equation}
Here $\eta$ denotes $(\eta_{\alpha\beta})$,
and $T$ denotes ordinary matrix transpose.
The column $h$ represents the unit in the normalized canonical frame.

The Frobenius structure determines a flat meromorphic connection
on the pullback of $TM$ to $M\times\mathbb C^*$, called the \emph{deformed flat connection} \cite{dubrovinGeometry2DTopological1996,giventalGromovWittenInvariantsQuantization2001}:
\begin{equation}
    \widehat\nabla_X
    =
    \nabla_X-z^{-1}(X\circ-),
    \qquad
    \widehat\nabla_{\partial_z}
    =
    \partial_z+z^{-2}A-z^{-1}\theta.
    \label{eq:frobenius-deformed-connection}
\end{equation}
Its horizontal sections satisfy
\begin{equation}
    \nabla_XY=z^{-1}X\circ Y,
    \qquad
    \partial_zY
    =
    \left(-z^{-2}A+z^{-1}\theta\right)Y.
    \label{eq:frobenius-horizontal-equations}
\end{equation}

On the tame semisimple point, a formal fundamental solution
near $z=0$ takes the form \cite{dubrovinGeometry2DTopological1996,dubrovinPainleveTranscendentsTwoDimensional1999,giventalGromovWittenInvariantsQuantization2001}
\begin{equation}
    Y(t,z)=\Phi(t)R(t,z)e^{U(t)/z},
    \qquad
    R(t,z)=\operatorname{Id}+\sum_{k\ge1}R_k(t)z^k.
    \label{eq:frobenius-formal-solution}
\end{equation}
We use the homogeneous, symplectic normalization
\begin{equation}
    R(t,z)R(t,-z)^T=\operatorname{Id},
    \qquad
    \left(E+z\frac{\partial}{\partial z}\right)R(t,z)=0,
    \label{eq:frobenius-R-normalization}
\end{equation}
where $E$ acts entrywise on $R$.
The series $R$ is formal; convergence in $z$ is not assumed.
Substituting \eqref{eq:frobenius-formal-solution} into the
connection equations yields
\begin{equation}
    [U,R_{k+1}]
    =
    \left(\Theta-k\operatorname{Id}\right)R_k,
    \qquad
    k\ge0,
    \qquad
    R_0=\operatorname{Id},
    \label{eq:frobenius-R-recursion}
\end{equation}
together with
\begin{equation}
    z\,\mathrm dR
    =
    [\mathrm dU,R]
    -
    z\,\Phi^{-1}(\nabla\Phi)R.
    \label{eq:frobenius-R-parameter-equation}
\end{equation}
Here $\mathrm d$ denotes exterior differentiation on $M$.
These equations fix our conventions for the $R$ matrix.

Finally, the periods used to construct the spectral curve
satisfy the order-zero second structure equations for the Frobenius manifold \cite{milanovPeriodMapQuantum2019,dunin-barkowskiDubrovinsSuperpotentialGlobal2019}:
\begin{equation}
    \begin{aligned}
        (x\operatorname{Id}-A)\partial_xJ
        &=
        \left(\theta-\frac12\operatorname{Id}\right)J,\\
        \nabla_XJ
        &=-X\circ\partial_xJ.
    \end{aligned}
    \label{eq:frobenius-second structure}
\end{equation}
Here $J(t,x)\in T_{t}M$ is a tangent vector-valued function, defined
locally away from
\begin{equation*}
    \det(x\operatorname{Id}-A)=0.
\end{equation*}
The compatibility of these equations follows from the
Frobenius axioms, and their possible finite singularities
in the $x$-plane occur at the canonical values $u_i$.
Given such a period, a local scalar primitive $y$ is specified by
\begin{equation}
    \partial_xy=\eta(e,J),
    \label{eq:frobenius-scalar-period}
\end{equation}
up to an $x$-independent term.
The choice, normalization and analytic continuation of these
periods will provide the functions entering the
spectral-curve construction.

\subsection{Givental--Teleman reconstruction}
\label{subsec:givental-teleman}

A cohomological field theory (CohFT) with flat unit assigns cohomology
classes on $\overline{\mathcal M}_{g,n}$ to state-space insertions,
compatibly with permutations, gluing and forgetting a unit insertion
\cite{kontsevichGromovWittenClassesQuantum1994}.
Its genus-zero primary potential determines a Frobenius structure.
For a homogeneous semisimple CohFT, the higher-genus theory is recovered
by Givental's formula and Teleman's classification
\cite{giventalSemisimpleFrobeniusStructures2001,
giventalGromovWittenInvariantsQuantization2001,
telemanStructure2DSemisimple2012}.
We recall only the ancestor formulation needed for the spectral-curve
construction, following the conventions of
\cite{xiongEYNARDORANTINTOPOLOGICALRECURSION}.

Fix a tame semisimple point $v\in M$ and retain $R,\Phi,h_i$ from
Section~\ref{subsec:frobenius}, suppressing their dependence on $v$.
For insertion variables $\mathbf t(z)=\sum_{k\ge0}t_kz^k$, the dilaton
shift is
\begin{equation}
    \mathbf q(z)=\mathbf t(z)-ze,
    \qquad
    q_k^\alpha=t_k^\alpha-\delta_{k1}\delta^\alpha_1,
    \qquad
    q_k^i=t_k^i-\delta_{k1}h_i.
    \label{eq:ancestor-dilaton-shifts}
\end{equation}
The flat and normalized canonical parameters are related by
$t_k^\alpha=\Phi^\alpha{}_i t_k^i$ and
$q_k^\alpha=\Phi^\alpha{}_i q_k^i$.

Write $\psi_r=c_1(L_r)$ for the cotangent-line class at the
$r$th marked point of $\overline{\mathcal M}_{g,n}$. 
The Witten--Kontsevich partition function is
\cite{wittenTwodimensionalGravityIntersection1990,
kontsevichIntersectionTheoryModuli1992}
\begin{equation}
    Z_{\mathrm{WK}}(\mathbf t)
    =\exp\left(
    \sum_{\substack{g,n\ge0\\2g-2+n>0}}\frac1{n!}
    \sum_{k_1,\ldots,k_n\ge0}
    \left(\int_{\overline{\mathcal M}_{g,n}}
    \prod_{r=1}^{n}\psi_r^{k_r}\right)
    \prod_{r=1}^{n}t_{k_r}
    \right).
\end{equation}
In dilaton variables, we write
$Z_{\mathrm{WK}}(\mathbf q)=Z_{\mathrm{WK}}(q_0,q_1+1,q_2,\ldots)$.

For $2g-2+n>0$, let
$\Omega_{g,n;v}(v_1,\ldots,v_n)$ denote the CohFT class at $v$,
with insertions $v_1,\ldots,v_n\in T_vM$.
The connected ancestor correlators are
\begin{equation}
    \left\langle
        \bar\tau_{k_1}(v_1)\cdots\bar\tau_{k_n}(v_n)
    \right\rangle_{g,n;v}
    :=
    \int_{\overline{\mathcal M}_{g,n}}
    \Omega_{g,n;v}(v_1,\ldots,v_n)
    \prod_{r=1}^{n}\psi_r^{k_r},
    \qquad k_r\geq0.
    \label{eq:ancestor-correlator-definition}
\end{equation}
Here $\bar\tau_{k_r}(v_r)$ denotes the insertion $v_r$
together with the power $\psi_r^{k_r}$.

The brackets are multilinear in the insertion vectors and
symmetric under simultaneous permutations of the pairs $(k_r,v_r)$.
The point $v$ specifies the background Frobenius structure.
Denote the formal insertion variables by
\begin{equation*}
    t_k=t_k^\alpha\partial_\alpha=t_k^i f_i.
\end{equation*}

The ancestor partition function packages these correlators as follows:
\begin{equation}
    Z_{\mathrm A}(\mathbf t;v)
    =\exp\left(
    \sum_{\substack{g,n\ge0\\2g-2+n>0}}\frac1{n!}
    \sum_{k_1,\ldots,k_n\ge0}
    \left\langle
    \bar\tau_{k_1}(t_{k_1})\cdots\bar\tau_{k_n}(t_{k_n})
    \right\rangle_{g,n;v}
    \right).
    \label{eq:ancestor-partition-function}
\end{equation}
In particular, multilinearity gives
\begin{equation}
    \begin{aligned}
        \left\langle
            \bar\tau_{k_1}(t_{k_1})\cdots
            \bar\tau_{k_n}(t_{k_n})
        \right\rangle_{g,n;v}
        =
        \sum_{i_1,\ldots,i_n=1}^{N}
        \left\langle
            \bar\tau_{k_1}(f_{i_1})\cdots
            \bar\tau_{k_n}(f_{i_n})
        \right\rangle_{g,n;v}
        \prod_{r=1}^{n}t_{k_r}^{i_r}.
    \end{aligned}
    \label{eq:ancestor-correlator-components}
\end{equation}

For our purposes, the quantization $\widehat R$ is specified by its
explicit action
\cite[Proposition~7.3]{giventalGromovWittenInvariantsQuantization2001}:
\begin{equation}
    (\widehat R G)(\mathbf q)
    =\left[
    \exp\left(\frac12\sum_{k,l\ge0}\sum_{i,j=1}^{N}
    (V_{kl})_{ij}
    \frac{\partial^2}{\partial q_k^i\partial q_l^j}\right)G
    \right](R^{-1}\mathbf q),
    \label{eq:givental-R-action}
\end{equation}
where
\begin{equation}
    \sum_{k,l\ge0}V_{kl}z^kw^l
    =\frac{\operatorname{Id}-R(-z)^TR(-w)}{z+w}.
    \label{eq:givental-quadratic-kernel}
\end{equation}
The right-hand side is a formal power series by symplecticity, and
$V_{kl}^T=V_{lk}$. The differential operator acts before substitution.
The change of frame acts by
$(\widehat\Phi G)(\mathbf q)=G(\Phi^{-1}\mathbf q)$.

\begin{theorem}[Givental--Teleman reconstruction
{\cite{giventalSemisimpleFrobeniusStructures2001,
giventalGromovWittenInvariantsQuantization2001,
telemanStructure2DSemisimple2012}}]
\label{thm:givental-teleman}
For a homogeneous semisimple CohFT with flat unit, the ancestor
partition function at a tame semisimple point $v$ is
\begin{equation}
    Z_{\mathrm A}(\mathbf t;v)
    =\left(\prod_{i=1}^{N}h_i\right)^{1/24}
    \widehat\Phi\,\widehat R\,Z_{\mathrm{WK}}^{\otimes N},
    \qquad
    Z_{\mathrm{WK}}^{\otimes N}
    =\prod_{i=1}^{N}Z_{\mathrm{WK}}(\mathbf q^i).
    \label{eq:ancestor-reconstruction}
\end{equation}
Here $R$ is the homogeneous calibration of
Section~\ref{subsec:frobenius}, and the operators act in dilaton
variables. 
\end{theorem}


For a Frobenius manifold given independently of a CohFT,
\eqref{eq:ancestor-reconstruction} defines the formal ancestor theory
used below; its correlators are denoted by
$\langle\prod_r\bar\tau_{k_r}(v_r)\rangle^R_{g,n;v}$.
Thus the spectral curve construction requires only the calibration $R$-matrix, the unit components $h_i$, and the change of frame $\Phi$.

\subsection{Topological recursion and the DOSS correspondence}
\label{subsec:EO-DOSS}

We recall Eynard--Orantin topological recursion
\cite{eynardInvariantsAlgebraicCurves2007}, allowing the underlying
Riemann surface to be noncompact. 
A spectral curve consists of a Riemann surface $\Sigma$, two
meromorphic functions $x,y$, and a symmetric meromorphic bidifferential
$B$ on $\Sigma\times\Sigma$.
Assume that $dx$ has finitely many simple zeros $a_1,\ldots,a_N$,
that $x,y$ are holomorphic near these points, and that $dy(a_i)\ne0$.
The bidifferential $B$ is holomorphic away from the diagonal and has
the local form
\begin{equation}
    B(p,q)=\left(
    \frac1{(\zeta(p)-\zeta(q))^2}+\text{holomorphic}
    \right)d\zeta(p)\,d\zeta(q)
    \label{eq:EO-diagonal-normalization}
\end{equation}
near the diagonal, for any local coordinate $\zeta$.
On a compact surface, one often chooses the normalized fundamental
bidifferential associated with a Torelli marking, called the Bergman kernel.
No such normalization is imposed here.

Write $u_i=x(a_i)$ (we use the same notation because these critical values will coincide with the canonical values in our construction), and choose local coordinates $s_i$ satisfying
\begin{equation}
    x=u_i+\frac{s_i^2}{2},
    \qquad
    \sigma_i(s_i)=-s_i,
    \label{eq:EO-local-coordinate}
\end{equation}
where $\sigma_i$ is the local involution of $x$.
Throughout the main text, the critical values are pairwise distinct,
in accordance with the tameness assumption.

The recursion starts from
\begin{equation}
    \omega_{0,1}=y\,dx,
    \qquad
    \omega_{0,2}=B.
\end{equation}
For $2g-2+n>0$, it defines symmetric meromorphic differentials by
\begin{equation}
\begin{aligned}
    \omega_{g,n+1}(p_0,p_1,\ldots,p_n)
    ={}&\sum_{i=1}^{N}\operatorname*{Res}_{q=a_i}K_i(p_0,q)
    \Bigg[
    \omega_{g-1,n+2}(q,\sigma_i(q),p_1,\ldots,p_n)
    \\
    &\quad+
    \sum_{\substack{g_1+g_2=g\\I\sqcup J=\{1,\ldots,n\}}}^{\prime}
    \omega_{g_1,|I|+1}(q,p_I)
    \omega_{g_2,|J|+1}(\sigma_i(q),p_J)
    \Bigg],
\end{aligned}
\label{eq:EO-recursion}
\end{equation}
where
\begin{equation}
    K_i(p,q)
    =-\frac{\displaystyle\int_{\sigma_i(q)}^qB(p,\cdot)}
    {2\bigl(y(q)-y(\sigma_i(q))\bigr)\,dx(q)}.
    \label{eq:EO-recursion-kernel}
\end{equation}
The integral is local near $a_i$.
The prime $\sum^{\prime}$ excludes factors of type $(0,1)$.

The recursion depends on $y$ only through its odd part
\begin{equation}
    \frac{y(s_i)-y(-s_i)}2
    =\sum_{k\ge0}y_{2k+1}^{\,i}s_i^{2k+1},
    \qquad y_1^{\,i}\ne0,
\end{equation}
where $y^i_k$ are the local coefficients of $y$ at $a_i$.
In particular, adding a function of $x$ to $y$ does not change the differentials $\omega_{g,n}$.

We now describe the local correspondence with Givental reconstruction
\cite{dunin-barkowskiIdentificationGiventalFormula2014}; see also
\cite{eynardInvariantsSpectralCurves2014,milanovEynardOrantinRecursion2014,xiongEYNARDORANTINTOPOLOGICALRECURSION}.
Fix a tame semisimple point $v$, and identify
$x(a_i)=u_i(v)$. 
We suppress the dependence of $R, h$ on $v$.

Near a pair of critical points $(a_i,a_j)$, expand
\begin{equation}
    B(s_i,s_j')
    =\delta_{ij}\frac{ds_i\,ds_j'}{(s_i-s_j')^2}
    +\sum_{k,l\ge0}B_{kl}^{ij}s_i^k(s_j')^l\,ds_i\,ds_j'.
    \label{eq:EO-B-expansion}
\end{equation}
The singular term is present only for $i=j$.

From topological recursion to the $R$ and $h$:
 The $R$-matrix is determined by $B$ as follows:
\begin{equation}
    (R_{k+1})^i{}_j
    =(-1)^k(2k-1)!!\,B_{0,2k}^{ij},
    \qquad k\ge0.
    \label{eq:B-extract-R}
\end{equation}
and 
\begin{equation}
h_i = y^i_1.
\end{equation} 
To get   $y$ and $B$ from $R$ and $h$ is much more involved.
It is easy to get $y$ from $R$ and $h$:
\begin{equation}
    \frac{y(s_i)-y(-s_i)}2
    =\sum_{k\ge0}[z^k]\bigl(R(z)^{-1}h\bigr)_i
    \frac{s_i^{2k+1}}{(2k+1)!!},
    \label{eq:DOSS-y-R-compatibility}
\end{equation}
is equivalent to 
\begin{equation}
    \sum_{k\ge0}(2k+1)!!\,y_{2k+1}^{\,i}z^k
    =\bigl(R(z)^{-1}h\bigr)_i.
    \label{eq:DOSS-unit-series}
\end{equation}
So only $y_{2k+1}^{\,i}$ are determined.
But this is not a problem because by \eqref{eq:EO-recursion-kernel}, 
the result of the topological recursion depends only on $\frac{y(s_i)-y(-s_i)}2$.

The coefficients $B^{ij}_{2k,2l}$ are determined from $R$:
\begin{equation}
    \sum_{k,l\ge0}(2k-1)!!(2l-1)!!B_{2k,2l}^{ij}z^kw^l
    =\left[
    \frac{\operatorname{Id}-R(z)^{-1}R(w)^{-T}}{z+w}
    \right]_{ij}.
    \label{eq:DOSS-B-R-compatibility}
\end{equation}
Since $R(z)^{-1}=R(-z)^T$, the right-hand side is precisely
\eqref{eq:givental-quadratic-kernel}.
Even though only $B_{2k,2l}^{ij}$ are  determined  by $R$ as in \eqref{eq:DOSS-B-R-compatibility}, 
under suitable conditions this is sufficient for 
the correspondence with ancestor correlators.
This is because the differential $\omega_{g,n}$ obtained from the recursion is expanded 
in some ancestor one-form $W^i_k$ defined recursively as follows:
first define the one-forms:
\begin{equation}
    \Xi_k^i(p)
    =(2k-1)!!\,[s_i^{2k}]\frac{B(s_i,p)}{ds_i},
    \qquad k\ge0,
    \label{eq:DOSS-coefficient-forms}
\end{equation}
next define
\begin{equation}
    W_0^i=\Xi_0^i,
    \qquad
    W_{k+1}^i=-\mathrm d\left(\frac{W_k^i}{dx}\right).
    \label{eq:B-ancestor-differential-recursion}
\end{equation}
 
Note by \eqref{eq:DOSS-B-R-compatibility},
the following identity is satisfied by the coefficients $B^{ij}_{2k,2l}$:
\begin{equation}
    (2l+1)B_{2r,2l+2}^{ij}+(2r+1)B_{2r+2,2l}^{ij}
    +\sum_{a=1}^{N}B_{0,2r}^{ai}B_{0,2l}^{aj}=0,
    \qquad r,l\ge0.
    \label{eq:B-differential-compatibility1}
\end{equation}
As observed in \cite{xiongEYNARDORANTINTOPOLOGICALRECURSION},
the following stronger compatibility condition is required
to establish the DOSS correspondence:
\begin{equation}
    (l+1)B_{2r,l+2}^{ij}+(2r+1)B_{2r+2,l}^{ij}
    +\sum_{a=1}^{N}B_{0,2r}^{ai}B_{0,l}^{aj}=0,
    \qquad r,l\ge0.
    \label{eq:B-differential-compatibility}
\end{equation} 
It holds for the $\mathcal{A}$-cycle normalized Bergman kernel in the case of compactifiable spectral curve
\cite{eynardInvariantsSpectralCurves2014,xiongEYNARDORANTINTOPOLOGICALRECURSION}.
As we observe in \cite{xiongEYNARDORANTINTOPOLOGICALRECURSION},
the  condition \eqref{eq:B-differential-compatibility} is equivalent to the following identity for the one-forms:
\begin{equation}
    -\mathrm d\left(\frac{\Xi_r^i}{dx}\right)
    =\Xi_{r+1}^i
    +(2r-1)!!\sum_{a=1}^{N}B_{0,2r}^{ai}\,\Xi_0^a.
    \label{eq:B-external-differential-recursion}
\end{equation}
With the definition of $W^i_k$ above, 
this can be reexpressed using the $R$-transform:
\begin{equation}
    \Xi(z,p)=\sum_{k\ge0}\Xi_k(p)z^k,
    \qquad
    W(z,p):=R(z)\Xi(z,p)=\sum_{k\ge0}W_k(p)z^k.
    \label{eq:DOSS-ancestor-forms}
\end{equation}
That is
\begin{equation}
    W_k^i(p)
    =\sum_{\ell=0}^{k}(R_\ell)^i{}_j\,\Xi_{k-\ell}^j(p).
\end{equation}
Thus $\Xi_k^i$ is extracted directly from $B$, and $W_k^i$ incorporates the change of variables in Givental's formula \eqref{eq:givental-R-action}.
   
As we observed in \cite{xiongEYNARDORANTINTOPOLOGICALRECURSION},
it is not necessary to have \eqref{eq:B-ancestor-differential-recursion}.
One only needs to have $W=R\Xi$. 
Hence a slight generalization of the DOSS correspondence can be established.

\begin{theorem}[DOSS correspondence
    {\cite{dunin-barkowskiIdentificationGiventalFormula2014}, see also \cite{milanovEynardOrantinRecursion2014,xiongEYNARDORANTINTOPOLOGICALRECURSION}}]
\label{thm:DOSS}
Suppose that \eqref{eq:DOSS-B-R-compatibility} and
\eqref{eq:DOSS-y-R-compatibility} hold for the homogeneous $R$-matrix
and unit vector $h$ at a tame semisimple point $v$.
Then, for $2g-2+n>0$ and $n\ge1$,
\begin{equation}
\begin{aligned}
    \omega_{g,n}(p_1,\ldots,p_n)
    ={}&\sum_{i_1,\ldots,i_n=1}^{N}
    \sum_{k_1,\ldots,k_n\ge0}
    \left\langle
    \bar\tau_{k_1}(f_{i_1})\cdots\bar\tau_{k_n}(f_{i_n})
    \right\rangle^R_{g,n;v}
    \prod_{r=1}^{n}W_{k_r}^{i_r}(p_r).
\end{aligned}
\label{eq:DOSS-ancestor-correspondence}
\end{equation}
The correlators are defined by Givental's ancestor reconstruction and the ancestor one-form $W$ is defined by \eqref{eq:DOSS-ancestor-forms}.
\end{theorem}

Note the correspondence is established using only the local data
$y^i_{2k+1}$, $B^{ij}_{kl}$ and $W^i_k$,
without knowing or requiring that they actually come from a global spectral curve. 
In next section we will propose a method to construct a global spectral curve.

\section{Construction of the spectral curve}
\label{sec:spectral-construction}

Fix a tame semisimple point $v$ of a conformal
Frobenius manifold with flat unit. We evaluate the Frobenius data $R, \Phi, h$ at $v$ and
suppress this dependence from the notation. 

Throughout this section, tangent vectors are represented by their
column coordinates in the flat frame $(\partial_\alpha)$.
Thus
\begin{equation}
    f_i=(\Phi^1{}_i,\ldots,\Phi^N{}_i)^T
\end{equation}
denotes the column vector representing the normalized canonical vector $f_i$. 
The matrices $A,\theta$ act in the flat frame, whereas $R$ is expressed in the normalized canonical frame.

We first construct a connected and simply connected surface over the $x$-plane with
simple critical points of $x$-function above the prescribed values $u_i$.
The second structure equation \eqref{eq:frobenius-second structure} for the Frobenius manifold then supplies normalized periods
near these points. 
Their scalar primitives and associated local
bidifferentials have the coefficients required by the DOSS correspondence. We use Cauchy integrals to realize these coefficients
as global functions and bidifferentials on a suitable open subdomain of the surface. 

We then turn to an explicit construction using a common period.
Subsection~\ref{subsec:Dubrovin-superpotential}
discusses its simultaneous normalization and uniqueness,
and Subsection~\ref{subsec:common-period-spectral-curve}
constructs the resulting spectral data and identifies
their topological recursion.

\subsection{The covering surface and the function \texorpdfstring{$x$}{x}}
\label{subsec:covering-surface}

Regard $u_1,\ldots,u_N$ as marked points in the complex $x$-plane.
Choose pairwise disjoint disks around them and, in each disk,
a radial slit starting at $u_i$. Join the slit neighbourhoods by narrow planar corridors arranged along an embedded tree.

The slit neighbourhood at $u_i$ admits a square-root coordinate
\begin{equation}
    s_i=\sqrt{2(x-u_i)}.
\end{equation}
Adjoin the other half of a sufficiently small $s_i$-disk,
so that the surface contains a full neighbourhood of $s_i=0$.
The coordinate $x$ on the planar part extends across this disc by
\begin{equation}
    x=u_i+\frac{s_i^2}{2}.
    \label{eq:construction-local-x}
\end{equation}
Performing this attachment at every $u_i$ produces a connected
open Riemann surface over the $x$-plane. We use ``over the
$x$-plane'' to mean the resulting holomorphic map $x$ to the plane;
no properness or finite-sheeted covering of the entire plane is asserted.

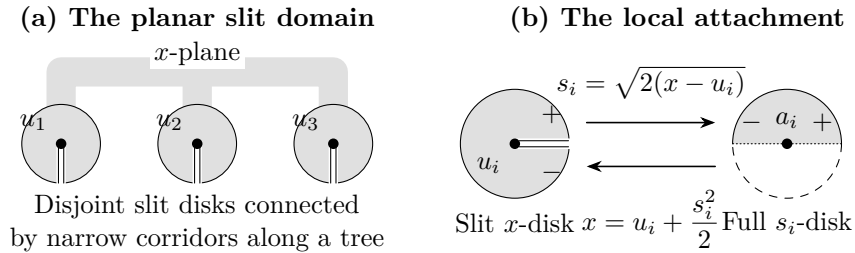
\begin{figure}[htbp]
\centering
\begin{tikzpicture}[
    x=1cm,
    y=1cm,
    font=\small,
    >=Stealth,
    line cap=round,
    line join=round,
    branch point/.style={
        circle,
        fill=black,
        inner sep=1.4pt
    },
    map arrow/.style={
        ->,
        semithick
    }
]

\begin{scope}

    \node[font=\small\bfseries] at (1.8,1.65)
        {(a) The planar slit domain};

    \draw[
        draw=black!12,
        line width=11pt
    ]
        (0,0.35) -- (0,0.95)
        -- (3.6,0.95) -- (3.6,0.35);

    \draw[
        draw=black!12,
        line width=11pt
    ]
        (1.8,0.95) -- (1.8,0.35);

    \foreach \xx/\ii in {0/1,1.8/2,3.6/3}{
        \fill[black!12] (\xx,0) circle[radius=0.52];
        \draw[thin] (\xx,0) circle[radius=0.52];

        \draw[
            white,
            line width=3.5pt
        ] (\xx,0) -- (\xx,-0.55);

        \draw[thin] (\xx-0.035,-0.04)
            -- (\xx-0.035,-0.52);
        \draw[thin] (\xx+0.035,-0.04)
            -- (\xx+0.035,-0.52);

        \node[branch point] at (\xx,0) {};
        \node[above left=2pt] at (\xx,0)
            {$u_{\ii}$};
    }

    \node[align=center] at (1.8,-1.05)
        {Disjoint slit disks connected\\
         by narrow corridors along a tree};

    \node[above,fill=white,inner sep=1.5pt]
        at (1.8,0.98) {$x$-plane};

\end{scope}

\begin{scope}[shift={(6.0,0)}]

    \node[font=\small\bfseries] at (2.15,1.65)
        {(b) The local attachment};

    \fill[black!12] (0,0) circle[radius=0.72];
    \draw[thin] (0,0) circle[radius=0.72];

    \draw[white,line width=4pt]
        (0,0) -- (0.75,0);

    \draw[thin] (0.06,0.04) -- (0.71,0.04);
    \draw[thin] (0.06,-0.04) -- (0.71,-0.04);

    \node[branch point] at (0,0) {};
    \node[below left=2pt] at (0,0) {$u_i$};

    \node[above=2pt] at (0.49,0.06) {$+$};
    \node[below=2pt] at (0.49,-0.06) {$-$};

    \node at (0,-1.05) {Slit $x$-disk};

    \draw[map arrow]
        (0.95,0.28) -- (2.65,0.28)
        node[midway,above=5pt,align=center]
        {$s_i=\sqrt{2(x-u_i)}$};

    \begin{scope}[shift={(3.6,0)}]

        \fill[black!12]
            (-0.72,0)
            arc[start angle=180,end angle=0,radius=0.72]
            -- cycle;

        \draw[thin]
            (-0.72,0)
            arc[start angle=180,end angle=0,radius=0.72];

        \draw[thin,dashed]
            (0.72,0)
            arc[start angle=0,end angle=-180,radius=0.72];

        \draw[thin,densely dotted]
            (-0.72,0) -- (0.72,0);

        \node[branch point] at (0,0) {};
        \node[above=2pt] at (0,0) {$a_i$};

        \node[above=1pt] at (0.47,0) {$+$};
        \node[above=1pt] at (-0.47,0) {$-$};

        \node[font=\scriptsize] at (0,0.43)
            {};
        \node[font=\scriptsize] at (0,-0.40)
            {};

        \node at (0,-1.05) {Full $s_i$-disk};

    \end{scope}

    \draw[map arrow]
        (2.65,-0.30) -- (0.95,-0.30)
        node[midway,below=5pt]
        {$x=u_i+\dfrac{s_i^2}{2}$};

\end{scope}

\end{tikzpicture}

\caption{
Construction of the surface over the $x$-plane. 
(a) Slit neighbourhoods of the $u_i$ are joined by narrow corridors arranged along a tree.
(b) A chosen square-root branch identifies a slit neighbourhood with a half-disc. Adjoining the other half gives a full local coordinate disk centred at $a_i$, on which $x=u_i+s_i^2/2$.
}
\label{fig:covering-surface}
\end{figure}

The planar domain is illustrated
in Figure~\ref{fig:covering-surface}. We have the following lemma.

\begin{lemma}
\label{lem:covering-surface}
The disks and corridors can be chosen so that the resulting
surface $\Sigma$ is simply connected, $x$ is bounded, and its
critical points are precisely $a_1,\ldots,a_N$, where
\begin{equation}
    s_i(a_i)=0,
    \qquad x(a_i)=u_i.
\end{equation}
All these critical points are simple, and $\Sigma$ is
biholomorphic to the upper half-plane $\mathbb{H}$.
\end{lemma}

We fix this surface and its local coordinates $s_i$ near each critical point $a_i$.
In the subsequent construction, we may shrink the neighbourhoods of $a_i$ and connecting corridors. 

\subsection{Normalized periods and scalar primitives}
\label{subsec:normalized-periods}

The second structure equation
\begin{equation}
    (x\operatorname{Id}-A)\partial_xJ
    =\left(\theta-\frac12\operatorname{Id}\right)J
    \label{eq:construction-period-equation}
\end{equation}
has regular singularities at the canonical values; see
\cite{dubrovinGeometry2DTopological1996,dubrovinPainleveTranscendentsTwoDimensional1999,milanovPeriodMapQuantum2019}.
The following normalization selects the period associated with each primitive idempotent $\pi_i$.

\begin{proposition}
\label{prop:normalized-periods}
For each $i$, there exists a unique odd solution $J_i(s_i):=J_i(x(s_i))$ defined on a punctured disk around $s_i=0$ such that
\begin{equation}
    J_i(s_i)=\frac{f_i}{s_i}+O(s_i).
\end{equation}
Its convergent expansion is
\begin{equation}
    J_i(s_i)=\sum_{m\ge0}
    \frac{(-1)^m}{(2m-1)!!}\,
    f_j(R_m)^j{}_i\,s_i^{2m-1}.
    \label{eq:construction-period-expansion}
\end{equation}
The scalar function
\begin{equation}
    Y_i(s_i)=\int_0^{s_i}\eta(e,J_i(\sigma))\,\sigma\,d\sigma
    \label{eq:construction-scalar-primitive}
\end{equation}
is holomorphic and odd, and satisfies
\begin{equation}
    \partial_xY_i=\eta(e,J_i),
    \qquad Y_i(0)=0.
\end{equation}
Moreover,
\begin{equation}
    Y_i(s_i)=\sum_{m\ge0}
    [z^m]\bigl(R(z)^{-1}h\bigr)_i
    \frac{s_i^{2m+1}}{(2m+1)!!}.
    \label{eq:construction-unit-expansion}
\end{equation}
In particular, $Y_i(s_i)=h_i s_i+O(s_i^3)$.
\end{proposition}

\begin{proof}
In the normalized canonical frame, the equation becomes
\begin{equation}
    \partial_x(\Phi^{-1}J)
    =(x\operatorname{Id}-U)^{-1}
    \left(\Theta-\frac12\operatorname{Id}\right)\Phi^{-1}J.
\end{equation}
The residue of the right-hand side at $x=u_i$ has only its $i$-th row nonzero.
Since $\Theta_{ii}=0$, its eigenvalues are $-1/2$ and $0$,
and its $-1/2$-eigenspace is the $i$-th coordinate vector.
The regular singularity gives a unique solution
of the form $(x-u_i)^{-1/2}$ times a holomorphic vector with
leading coefficient $f_i/\sqrt{2}$ in the flat frame; see \cite{dubrovinPainleveTranscendentsTwoDimensional1999}.
This is the required odd solution in $s_i$.

To identify its coefficients, write
\begin{equation*}
    J_i(s_i)=\sum_{m\ge0}c_m^i s_i^{2m-1}.
\end{equation*}
After the coordinate transformation $x=u_i+s_i^2/2$, the equation \eqref{eq:construction-period-equation} becomes
\begin{equation}
    \left(u_i\operatorname{Id}-A+\frac{s_i^2}{2}\operatorname{Id}\right)
    \partial_{s_i}J_i
    =\bigl(\theta-\frac{1}{2}\operatorname{Id}\bigr)s_iJ_i.
    \label{eq:construction-period-equation-s}
\end{equation}
Substitution into this equation gives
\begin{equation}
    (2m-1)(u_i\operatorname{Id}-A)c_m^i
    =\bigl(\theta-(m-1)\operatorname{Id}\bigr)c_{m-1}^i,
    \qquad m\ge1.
\end{equation}
By \eqref{eq:frobenius-R-recursion}, the formal coefficients
\begin{equation*}
    c_m^i=\frac{(-1)^m}{(2m-1)!!}f_a(R_m)^a{}_i
\end{equation*}
satisfy this relation and have the prescribed leading term. Recall that here $f_a=(\Phi^{\alpha}{}_{a})_{\alpha=1}^N$ is viewed as a column vector.
Uniqueness of the regular singularity expansion identifies them
with the coefficients of the convergent solution.
This proves \eqref{eq:construction-period-expansion}.

Since $\sigma J_i(\sigma)$ is holomorphic and even,
\eqref{eq:construction-scalar-primitive} is holomorphic and odd.
The identity $dx=s_i\,ds_i$ gives its $x$-derivative.
Finally,
\begin{equation}
    \eta(e,J_i(s_i))
    =\sum_{m\ge0}
    \frac{(-1)^m h_a(R_m)^a{}_i}{(2m-1)!!}\,s_i^{2m-1}.
\end{equation}
Using $R(z)^{-1}=R(-z)^T$ and integrating term by term proves
\eqref{eq:construction-unit-expansion}.
\end{proof}

\subsection{The canonical local bidifferential}
\label{subsec:canonical-local-kernel}

Choose the branch neighbourhoods of $a_i$ small enough that their
$x$-images are disjoint. For points in the $i$-th and $j$-th
neighbourhoods, write
\begin{equation*}
    x_1=u_i+\frac{s^2}{2},
    \qquad x_2=u_j+\frac{t^2}{2}.
\end{equation*}
Define the local period bidifferential
\begin{equation}
\begin{aligned}
    B^{\mathrm{local}}_{ij}(s,t)
    ={}&\frac{1}{(x_1-x_2)^2}
    \Bigg\{\frac{\delta_{ij}}2
    +\eta\!\left(J_i(s),
    \left[(x_1-x_2)\theta
    +\frac{x_1+x_2}{2}\operatorname{Id}-A\right]J_j(t)\right)
    \Bigg\}\,dx_1\,dx_2.
\end{aligned}
\label{eq:construction-canonical-kernel}
\end{equation}
We verify its normalization and coefficients directly.

\begin{proposition}
\label{prop:canonical-local-kernel}
The bidifferentials $B^{\mathrm{local}}_{ij}$ are symmetric.
For $i\ne j$, they are holomorphic near $(s,t)=(0,0)$.
For $i=j$, their difference from $ds\,dt/(s-t)^2$ is holomorphic.
If we expand the bidifferentials as
\begin{equation}
    B^{\mathrm{local}}_{ij}(s,t)
    =\delta_{ij}\frac{ds\,dt}{(s-t)^2}
    +\sum_{k,l\ge0}B_{kl}^{ij}s^kt^l\,ds\,dt,
\end{equation}
then, for all $k\ge0$ and $i,j=1,\ldots,N$,
\begin{equation}
    (2k-1)!!B_{0,2k}^{ij}=(-1)^k(R_{k+1})^i{}_j, \qquad
    B_{0,2k+1}^{ij}=0,
\end{equation}
and
\begin{equation}
    (m+1)B_{m+2,l}^{ij}+(l+1)B_{m,l+2}^{ij}
    +\sum_{a=1}^{N}B_{0,m}^{ai}B_{0,l}^{aj}=0.
    \label{eq:construction-full-local-compatibility}
\end{equation}
In particular,
\begin{equation}
    \sum_{k,l\ge0}(2k-1)!!(2l-1)!!B_{2k,2l}^{ij}z^kw^l
    =\left[
    \frac{\operatorname{Id}-R(z)^{-1}R(w)^{-T}}{z+w}
    \right]_{ij}.
    \label{eq:construction-local-edge-identity}
\end{equation}
\end{proposition}

\begin{proof}
The relations $A^*=A$ and $\theta^*=-\theta$ show symmetry.
For $i\ne j$, the denominator is nonzero on sufficiently small
neighbourhoods. Multiplication by $dx_1dx_2=st\,ds\,dt$
removes the apparent singularities of $J_i(s)$ and $J_j(t)$.

For the diagonal case, fix $i=j$.
First, any two solutions $J,K$ of
\eqref{eq:construction-period-equation} satisfy
\begin{equation}
\begin{aligned}
    \partial_x\eta\bigl(J,(x\operatorname{Id}-A)K\bigr)
    ={}&
    \eta\bigl((\theta-\tfrac12\operatorname{Id})J,K\bigr)
    +\eta(J,K)\\
    &+\eta\bigl(J,(\theta-\tfrac12\operatorname{Id})K\bigr)
    =0.
\end{aligned}
\end{equation}
In particular,
$\eta(J_i,(x\operatorname{Id}-A)J_i)$ is constant.
Since
\begin{equation*}
    J_i(s)=\frac{f_i}{s}+O(s),
    \qquad
    Af_i=u_i f_i,
    \qquad
    \eta(f_i,f_i)=1,
\end{equation*}
it gives
\begin{equation*}
    \eta\bigl(J_i(s),(x\operatorname{Id}-A)J_i(s)\bigr)
    =\frac12+O(s^2).
\end{equation*}
Taking the limit as $s\to0$ therefore yields
\begin{equation}
    \eta\bigl(J_i(s),(x\operatorname{Id}-A)J_i(s)\bigr)
    =\frac12.
    \label{eq:construction-period-pairing}
\end{equation}

Within this proof, put
\begin{equation*}
    M(x_1,x_2)
    =(x_1-x_2)\theta
    +\frac{x_1+x_2}{2}\operatorname{Id}-A,
\end{equation*}
and, initially for $st\ne0$, set
\begin{equation*}
    N(s,t)=\eta\bigl(J_i(s),M(x_1,x_2)J_i(t)\bigr).
\end{equation*}
On $t=s$, equation \eqref{eq:construction-period-pairing}
gives $N(s,s)=1/2$.
On $t=-s$, oddness of $J_i$ gives 
\begin{equation}
    N(s,s)=\frac{1}{2},\qquad N(s,-s)=-\frac12.
    \label{eq:construction-diagonal-pairing-values}
\end{equation}

We also need the first derivatives.
Differentiating with respect to $t$ while keeping $s$ fixed,
and using $\partial_t=t\partial_{x_2}$,
we obtain
\begin{equation}
\begin{aligned}
    \partial_tN(s,t)
    =t\,\eta\Bigl(
    J_i(s),\,
    \bigl(-\theta+\tfrac12\operatorname{Id}\bigr)J_i(t)
    +M(x_1,x_2)\partial_{x_2}J_i(t)
    \Bigr).
\end{aligned}
\end{equation}
At either $t=s$ or $t=-s$, we have $x_2=x_1$ and
$M(x_1,x_2)=x_2\operatorname{Id}-A$.
The second structure equation consequently gives
\begin{equation}
    \partial_tN(s,s)=0,
    \qquad
    \partial_tN(s,-s)=0.
    \label{eq:construction-diagonal-pairing-derivatives}
\end{equation}

Define the holomorphic function
\begin{equation}
    K(s,t)
    =\frac{st}{2}
    +\eta\bigl(sJ_i(s),M(x_1,x_2)tJ_i(t)\bigr).
\end{equation}
Its holomorphicity follows from that of $sJ_i(s)$ and
$tJ_i(t)$.
For $st\ne0$, it equals $st(1/2+N(s,t))$.
Thus \eqref{eq:construction-diagonal-pairing-values} and
\eqref{eq:construction-diagonal-pairing-derivatives} imply
\begin{equation}
\begin{aligned}
    K(s,s)&=s^2,
    &\qquad \partial_tK(s,s)&=s,\\
    K(s,-s)&=0,
    &\qquad \partial_tK(s,-s)&=0.
\end{aligned}
\end{equation}

Since $dx_1dx_2=st\,ds\,dt$ and
$x_1-x_2=(s^2-t^2)/2$, we have
\begin{equation}
    B^{\mathrm{local}}_{ii}(s,t)
    -\frac{ds\,dt}{(s-t)^2}
    =
    \frac{4K(s,t)-(s+t)^2}
    {(s-t)^2(s+t)^2}\,ds\,dt.
    \label{eq:construction-Cauchy-difference}
\end{equation}
Let
\begin{equation*}
    F(s,t)=4K(s,t)-(s+t)^2.
\end{equation*}
The preceding identities give
\begin{equation}
\begin{aligned}
    F(s,s)&=0,
    &\qquad \partial_tF(s,s)&=0,\\
    F(s,-s)&=0,
    &\qquad \partial_tF(s,-s)&=0.
\end{aligned}
\end{equation}

We obtain a holomorphic factorization
\begin{equation}
    F(s,t)=(t-s)^2(t+s)^2H(s,t),
\end{equation}
with $H$ holomorphic on a neighbourhood of $(0,0)$.
Substitution into \eqref{eq:construction-Cauchy-difference}
proves
\begin{equation}
    B^{\mathrm{local}}_{ii}(s,t)
    =\frac{ds\,dt}{(s-t)^2}+H(s,t)\,ds\,dt.
\end{equation}
Hence the diagonal principal part is exactly the standard kernel. 

To determine the coefficients, we differentiate directly and use the second structure equation to obtain
\begin{equation}
    \left(\mathcal L_{\partial_{x_1}}^{(1)}
    +\mathcal L_{\partial_{x_2}}^{(2)}\right)B^{\mathrm{local}}_{ij}
    =-\eta(\partial_{x_1}J_i,\partial_{x_2}J_j)\,dx_1dx_2,
    \label{eq:construction-phase-derivative}
\end{equation}
where $\mathcal L_{\partial_{x_i}}^{(i)}$ denotes the Lie derivative with respect to $\partial_{x_i}$ on the i-th factor.
The Lie derivatives include the differentials.

Write
$B^{\mathrm{local}}_{ij}=b_{ij}(s,t)\,ds\,dt$.
Since
    $\partial_{x_1}=s^{-1}\partial_s$,
    and
    $\partial_{x_2}=t^{-1}\partial_t$,
the Lie derivatives give 
\begin{equation}
\begin{aligned}
    \mathcal L_{\partial_{x_1}}^{(1)}
    \bigl(b_{ij}(s,t)\,ds\,dt\bigr)
    &=
    \partial_s\left(\frac{b_{ij}(s,t)}s\right)ds\,dt,\\
    \mathcal L_{\partial_{x_2}}^{(2)}
    \bigl(b_{ij}(s,t)\,ds\,dt\bigr)
    &=
    \partial_t\left(\frac{b_{ij}(s,t)}t\right)ds\,dt.
\end{aligned}
\end{equation}
Consequently, \eqref{eq:construction-phase-derivative} becomes
\begin{equation}
    \partial_s\left(\frac{b_{ij}(s,t)}s\right)
    +
    \partial_t\left(\frac{b_{ij}(s,t)}t\right)
    =
    -\eta\left(\frac{dJ_i}{ds},\frac{dJ_j}{dt}\right).
    \label{eq:construction-phase-derivative-local}
\end{equation}
Here
\begin{equation*}
    b_{ij}(s,t)=\delta_{ij}\frac1{(s-t)^2}
    +\sum_{m,l\ge0}B_{ml}^{ij}s^mt^l.
\end{equation*}

Write
\begin{equation*}
    \frac{dJ_i}{ds}=-\frac{f_i}{s^2}
    +\sum_{m\ge0}C_m^i s^m.
\end{equation*}
In square-root coordinates, the principal part satisfies
\begin{equation*}
    \partial_s\frac1{s(s-t)^2}
    +\partial_t\frac1{t(s-t)^2}
    =-\frac1{s^2t^2}.
\end{equation*}
Comparing the coefficients of $s^{-2}t^l$ and $s^mt^l$ in
\eqref{eq:construction-phase-derivative-local} gives
\begin{equation}
    -B_{0,l}^{ij}=\eta(f_i,C_l^j),
    \qquad
    (m+1)B_{m+2,l}^{ij}+(l+1)B_{m,l+2}^{ij}
    =-\eta(C_m^i,C_l^j).
\end{equation}
By the expansion of $J_i$ in \eqref{eq:construction-period-expansion}, we have
\begin{equation}
    C_{2m}^i=\frac{(-1)^{m+1}}{(2m-1)!!}f_a(R_{m+1})^a{}_i,\qquad C_{2m+1}^i=0, \qquad m\ge0.
\end{equation}
Hence
\begin{equation*}
    (2k-1)!!B_{0,2k}^{ij}=(-1)^k(R_{k+1})^i{}_j, \qquad
    B_{0,2k+1}^{ij}=0, \qquad k\ge0.
\end{equation*}
\begin{equation*}
    (m+1)B_{m+2,l}^{ij}+(l+1)B_{m,l+2}^{ij}
    +\sum_{a=1}^{N}B_{0,m}^{ai}B_{0,l}^{aj}=0.
\end{equation*}
This proves \eqref{eq:construction-local-edge-identity}.
\end{proof}

The bidifferential $B^{\mathrm{local}}_{ij}$ is only defined for the product of the local disks. In Subsection~\ref{subsec:global-spectral-data}, we construct a global bidifferential on $\Sigma$ whose even--even regular coefficients agree with those of these local kernels.

\subsection{Global construction of \texorpdfstring{$y$ and $B$}{y and B}}
\label{subsec:global-spectral-data}

We now shrink the surface of Lemma~\ref{lem:covering-surface}. The scalar function $y$ is constructed first; the bidifferential $B$
is then obtained by correcting a reference bidifferential.

\begin{proposition}
\label{prop:global-scalar}
The surface $\Sigma$ can be chosen to carry a holomorphic function $y$ with
\begin{equation}
    \frac{y(s_i)-y(-s_i)}2=Y_i(s_i).
    \label{eq:construction-global-y}
\end{equation}
In particular, $dy(a_i)=h_i\,ds_i\ne0$.
\end{proposition}

\begin{proof}
Choose pairwise disjoint segments
\begin{equation}
    C_i=\left\{u_i+\frac{\rho_i^2 t^2}{2}:0\le t\le1\right\}
\end{equation}
inside the disks of convergence of the periods $J_i$, oriented away from $u_i$. Use the tree construction of Section~\ref{subsec:covering-surface} with these slits, retaining
full disks of radii smaller than $|\rho_i|$.
The outer endpoints remain outside the retained surface.

On the planar slit part, define
\begin{equation}
    y(x)=\frac1{2\pi\mathrm i}\sum_{i=1}^{N}
    \int_{C_i}
    \frac{2Y_i\bigl(\sqrt{2(\tau-u_i)}\bigr)}{\tau-x}\,d\tau.
    \label{eq:construction-y-Cauchy}
\end{equation}
Since $Y_i(s)=O(s)$, the density in
\eqref{eq:construction-y-Cauchy} is
$O((\tau-u_i)^{1/2})$, and the integral converges at $x=u_i$.

\begin{figure}[htbp]
\centering
\begin{tikzpicture}[
    x=1cm,
    y=1cm,
    font=\small,
    >=Stealth,
    line cap=round,
    line join=round
]

\begin{scope}

    \node[font=\small\bfseries] at (1.65,1.85)
        {The oriented cut in the $x$-plane};

    \fill[black!6] (0,0) circle[radius=1.15];
    \draw[densely dashed,black!50]
        (0,0) circle[radius=1.15];

    \draw[semithick] (0,0) -- (3.65,0);
    \draw[->,semithick] (1.55,0) -- (2.45,0);

    \fill (0,0) circle[radius=1.5pt];
    \fill (0.72,0) circle[radius=1.3pt];

    \draw[semithick,fill=white]
        (3.65,0) circle[radius=2pt];

    \node[below left=3pt] at (0,0) {$u_i$};
    \node[below right=2pt] at (0.72,0) {$x_0$};
    \node[below=5pt] at (2.05,0) {$C_i$};

    \node[below=5pt,align=center] at (3.65,0)
        {$u_i+\rho_i^2/2$};

    \draw[->] (0.72,0.85) -- (0.72,0.12);
    \draw[->] (0.72,-0.85) -- (0.72,-0.12);

    \node[above=2pt] at (0.72,0.85)
        {$y_+(x_0)$};
    \node[below=2pt] at (0.72,-0.85)
        {$y_-(x_0)$};

    \node[above] at (2.35,0.55)
        {$+$};
    \node[below] at (2.35,-0.55)
        {$-$};

\end{scope}

\draw[->,semithick] (6.00,0.55) -- (4.35,0.55)
    node[midway,above=5pt,align=center]
    {$x=u_i+\dfrac{s_i^2}{2}$};

\begin{scope}[shift={(7.75,0)}]

    \node[font=\small\bfseries] at (0,1.85)
        {The square-root chart};

    \fill[black!6] (0,0) circle[radius=1.25];
    \draw[thin] (0,0) circle[radius=1.25];
    \draw[densely dotted,black!50]
        (-1.25,0) -- (1.25,0);

    \fill (0,0) circle[radius=1.5pt];
    \node[below=3pt] at (0,0) {$a_i$};

    \fill (-0.78,0) circle[radius=1.4pt];
    \fill (0.78,0) circle[radius=1.4pt];

    \node[below=4pt] at (-0.78,0) {$-s_0$};
    \node[below=4pt] at (0.78,0) {$s_0$};

    \draw[<->,thin]
        (-0.72,0.45) to[bend left=35] (0.72,0.45);

    \node at (0,0.98) {$s_i\mapsto-s_i$};

\end{scope}

\end{tikzpicture}
\caption{
The boundary values of the Cauchy integral along an oriented
segment $C_i$ and their lifts to the square-root chart.
}
\label{fig:scalar-Cauchy-jump}
\end{figure}
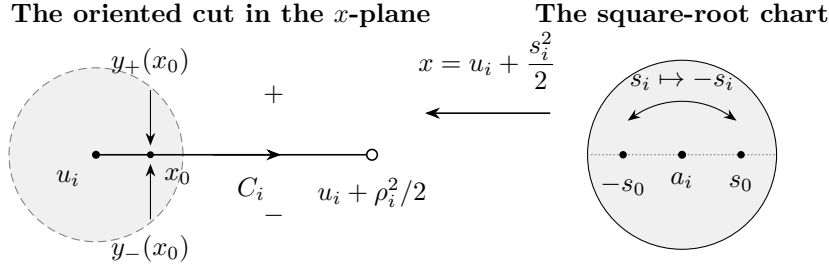

Let $x_0$ be an interior point of $C_i$ lying in the retained disk, and put
\begin{equation*}
    s_0=\sqrt{2(x_0-u_i)}
\end{equation*}
with the branch just specified.
Denote by $y_+(x_0)$ and $y_-(x_0)$ the boundary values from the $+$ and $-$ sides of the oriented segment, respectively; see Figure~\ref{fig:scalar-Cauchy-jump}.
The Sokhotski--Plemelj formula gives
\begin{equation}
\begin{aligned}
    y_\pm(x_0)
    ={}&
    \frac1{2\pi\mathrm i}
    \sum_{j=1}^{N}
    \operatorname{pv}\!\int_{C_j}
    \frac{2Y_j\bigl(\sqrt{2(\tau-u_j)}\bigr)}
         {\tau-x_0}\,d\tau
    \ \pm Y_i(s_0).
\end{aligned}
\end{equation}
Only the integral over $C_i$ requires a principal value.
The other integrals are holomorphic near $x_0$, so their
boundary values agree. The final term is one half of the density, with opposite signs on the two banks.
Subtracting therefore yields
\begin{equation}
    y_+(x_0)-y_-(x_0)
    =2Y_i(s_0)
    =2Y_i\bigl(\sqrt{2(x_0-u_i)}\bigr).
    \label{eq:construction-y-jump}
\end{equation}

We choose the square-root attachment so that the left bank
lifts to $s_i=s_0$ and the right bank to $s_i=-s_0$.
After the holomorphic extension at the branch point established
below, \eqref{eq:construction-y-jump} becomes
\begin{equation}
    y(s_0)-y(-s_0)=2Y_i(s_0).
\end{equation}
Both sides are holomorphic functions of the square-root
coordinate. The identity theorem consequently extends this
equality from the lifted segment to the full retained disk:
\begin{equation}
    \frac{y(s_i)-y(-s_i)}2=Y_i(s_i).
\end{equation}

Extension at the branch endpoint can be checked directly.
Write $Y_i(s_i)=s_i\,g(s_i^2)$, with $g$ holomorphic function.
The substitution $\tau=u_i+\xi^2/2$ reduces its local contribution to
\begin{equation*}
\begin{aligned}
\frac{2}{\pi\mathrm i}
\int_0^{\rho_i}
\frac{\xi^2g(\xi^2)}{\xi^2-s_i^2}\,\d\xi
=
\frac{2}{\pi\mathrm i}
\int_0^{\rho_i}
\frac{\xi^2g(\xi^2)-s_i^2g(s_i^2)}
{\xi^2-s_i^2}\,\d\xi
+
\frac{2s_i^2g(s_i^2)}{\pi\mathrm i}
\int_0^{\rho_i}\frac{\d\xi}{\xi^2-s_i^2}.
\end{aligned}
\end{equation*}
The integrand in the first term on the right-hand side is a holomorphic divided difference and hence is holomorphic in $s_i$. The remaining term is proportional to
\begin{equation*}
    s_i a(s_i^2)\left[
    \log\left(\frac{\rho_i-s_i}{\rho_i+s_i}\right)
    +\text{constant}\right],
\end{equation*}
which is holomorphic for $|s_i|<|\rho_i|$.
All other segments contribute holomorphic functions of $x$
near $u_i$, hence only even powers of $s_i$.
The possible logarithmic singularities at the outer endpoints
are excluded by the choice of surface. Thus $y$ extends
holomorphically, with \eqref{eq:construction-global-y}.
Its linear coefficient is $h_i$ by
Proposition~\ref{prop:normalized-periods}.
\end{proof}

The difference between $y$ and $Y_i$ on a branch disk is even
in $s_i$. This freedom is harmless for topological recursion. 

\begin{proposition}
\label{prop:global-bidifferential}
After a further choice of sufficiently small retained
neighbourhoods, there exists a symmetric bidifferential $B$
on $\Sigma\times\Sigma$ whose only pole is the normalized
diagonal double pole and whose even--even regular coefficients
agree with those of $B^{\mathrm{local}}_{ij}$ near every pair
of critical points.
\end{proposition}

\begin{proof}
Carry out the preceding construction on a slightly larger
surface, and choose a uniformizing coordinate $\zeta: \Sigma\to \mathbb H$ there.
Set
\begin{equation}
    B_0(p,q)
    =
    \frac{\d\zeta(p)\,\d\zeta(q)}
         {(\zeta(p)-\zeta(q))^2}.
\end{equation}
$B_0$ is symmetric and has no poles other than the normalized double pole along the diagonal.
The final surface will be obtained by restriction. 

Near the $i$-th and $j$-th critical points, write
\begin{equation*}
    x(p)=u_i+\frac{s^2}{2},
    \qquad
    x(q)=u_j+\frac{t^2}{2}.
\end{equation*}
The difference $B^{\mathrm{local}}_{ij}-B_0$ is holomorphic in the domain of $B^{\mathrm{local}}_{ij}$:
for $i=j$, its diagonal principal parts cancel.
Define holomorphic functions $b_{ij}$ by
\begin{equation}
    \frac14(1-\sigma_i^*)(1-\sigma_j^*)
    \bigl(B^{\mathrm{local}}_{ij}-B_0\bigr)
    =
    b_{ij}(s^2,t^2)\d s \d t.
    \label{eq:construction-B-discrepancy}
\end{equation}
Here $\sigma_i$ is the local involution, $\sigma_i^*(f(s)\d s)=-f(-s)\d s$ and $\sigma_j^*(f(t)\d t)=-f(-t)\d t$.

Choose the oriented segments
\begin{equation*}
    C_i=
    \left\{
        u_i+\frac{\rho_i^2 r^2}{2}:0\leq r\leq1
    \right\}
\end{equation*}
sufficiently short that all the functions $b_{ij}$ are
holomorphic on neighbourhoods of the corresponding products
of closed integration segments in the square-root coordinates.
Use the attachment convention of the preceding construction:
the slit neighbourhood of $u_i$ lifts to
$\operatorname{Im}(s/\rho_i)>0$, with the left bank
corresponding to positive multiples of $\rho_i$.

On the planar slit part, define
\begin{equation}
\begin{aligned}
    \mathcal H(p,q)
    ={}&
    \frac{dx(p)\,dx(q)}{(\pi\mathrm i)^2}
    \sum_{i,j=1}^{N}
    \int_0^{\rho_i}\int_0^{\rho_j}
    \frac{b_{ij}(\xi^2,\chi^2)\,d\chi\,d\xi}
    {\bigl(u_i+\xi^2/2-x(p)\bigr)
     \bigl(u_j+\chi^2/2-x(q)\bigr)}.
\end{aligned}
\label{eq:construction-B-Cauchy}
\end{equation}
The integrals are taken along straight segments.
This expression is holomorphic on the product of the planar
slit parts, and it is symmetric under interchange of $p$ and $q$.

We verify both its extension and its prescribed local
coefficients using the following one-variable calculation.
Let $a(\xi^2)$ be holomorphic on a neighbourhood of
$\{\xi^2:\xi\in[0,\rho]\}$.
On $\operatorname{Im}(s/\rho)>0$, one has
\begin{equation}
\begin{aligned}
    \frac{2s}{\pi\mathrm i}
    \int_0^\rho
    \frac{a(\xi^2)}{\xi^2-s^2}\,d\xi
    =
    a(s^2)
    +\frac{a(s^2)}{\pi\mathrm i}
       \log\left(\frac{\rho-s}{\rho+s}\right)
    +
    \frac{2s}{\pi\mathrm i}
    \int_0^\rho
    \frac{a(\xi^2)-a(s^2)}
         {\xi^2-s^2}\,d\xi .
\end{aligned}
\label{eq:construction-B-local-integral}
\end{equation}
The logarithm on the right is the holomorphic branch
vanishing at $s=0$.

The divided difference in the last term of
\eqref{eq:construction-B-local-integral}
has removable singularities at $\xi^2=s^2$.
Thus the right-hand side extends holomorphically to a
sufficiently small full disc about $s=0$.
Its first term is even in $s$, whereas both remaining
terms are odd.
Consequently, after multiplication by $ds$, its
anti-invariant part is exactly $a(s^2)\,ds$.

Now consider the first variable of
\eqref{eq:construction-B-Cauchy}
near the $i$-th critical point.
For the summands integrated over $[0,\rho_i]$, the identity
$dx(p)=s\,ds$ converts the corresponding one-variable
factor into
\begin{equation*}
    \frac{dx(p)}{\pi\mathrm i}
    \int_0^{\rho_i}
    \frac{a(\xi^2)}
         {u_i+\xi^2/2-x(p)}\,d\xi
    =
    \frac{2s}{\pi\mathrm i}
    \int_0^{\rho_i}
    \frac{a(\xi^2)}{\xi^2-s^2}\,d\xi\,ds.
\end{equation*}
It therefore extends holomorphically across the branch
point, with anti-invariant part $a(s^2)\,ds$.
The summands associated with the other segments are
holomorphic functions of $x(p)$ times $dx(p)$.
Their pullbacks have the form
$s\,F(s^2)\,ds$, which is invariant under $s\mapsto-s$;
hence they contribute nothing to the projection.

Apply this argument successively in the two variables.
The parameter-dependent version of
\eqref{eq:construction-B-local-integral}
gives a jointly holomorphic extension of $\mathcal H$
to every retained product of branch disks and yields
\begin{equation}
    \frac14(1-\sigma_i^*)(1-\sigma_j^*)\mathcal H
    =
    b_{ij}(s^2,t^2)\,ds\,dt.
    \label{eq:construction-B-correction}
\end{equation}
The outer integration endpoints are excluded from the retained surface.
Thus $\mathcal H$ defines a holomorphic symmetric
bidifferential on $\Sigma\times\Sigma$.

Finally, set
\begin{equation}
    B=B_0+\mathcal H.
    \label{eq:construction-global-B}
\end{equation}
Since $\mathcal H$ is holomorphic and symmetric, $B$
has exactly the required diagonal principal part and no other poles.
By
\eqref{eq:construction-B-discrepancy}
and
\eqref{eq:construction-B-correction},
$B$ agrees with $B^{\mathrm{local}}_{ij}$ in its even--even regular coefficients on every product of branch discs.
\end{proof}

The smaller retained surface is essential: finite Cauchy
integrals can have logarithmic singularities at their outer endpoints. 
Moreover,the full 
condition \eqref{eq:construction-full-local-compatibility}
is not asserted for $B$.

\subsection{Spectral curve realization of ancestor correlators}
\label{subsec:spectral-ancestor-reconstruction}

We now assemble the construction and apply Theorem \ref{thm:DOSS}  in Section~\ref{subsec:EO-DOSS}
to get a global version of the DOSS correspondence.

\begin{theorem}
\label{thm:connected-spectral-realization}
Let $M$ be a conformal Frobenius manifold
with flat unit, and let $v\in M$ be a tame semisimple point.
There exists a connected spectral curve
\begin{equation}
    (\Sigma,x,y,B),\qquad \Sigma\cong\mathbb H,
\end{equation}
such that $x$ has exactly $N$ simple critical points $a_i$
with $x(a_i)=u_i(v)$, $dy(a_i)\ne0$, and $B$ has only its
normalized diagonal double pole. In the local coordinates
$x=u_i+s_i^2/2$, the expansions of $y$ and $B$ satisfy
\begin{equation}
    \frac{y(s_i)-y(-s_i)}2
    =\sum_{k\ge0}[z^k](R(z)^{-1}h)_i
    \frac{s_i^{2k+1}}{(2k+1)!!}
\end{equation}
and
\begin{equation}
    \sum_{k,l\ge0}(2k-1)!!(2l-1)!!B_{2k,2l}^{ij}z^kw^l
    =\left[
    \frac{\operatorname{Id}-R(z)^{-1}R(w)^{-T}}{z+w}
    \right]_{ij}.
\end{equation}
For $n\ge1$ and $2g-2+n>0$, the recursion differentials satisfy
\begin{equation}
\begin{aligned}
    \omega_{g,n}(p_1,\ldots,p_n)
    ={}&\sum_{i_1,\ldots,i_n=1}^{N}
    \sum_{k_1,\ldots,k_n\ge0}
    \left\langle\prod_{r=1}^{n}\bar\tau_{k_r}(f_{i_r})
    \right\rangle^R_{g,n;v}
    \prod_{r=1}^{n}W_{k_r}^{i_r}(p_r),
\end{aligned}
\label{eq:construction-ancestor-correspondence}
\end{equation}
where $\Xi$ and $W=R\Xi$ are extracted from the bidifferential $B$ as in \eqref{eq:DOSS-coefficient-forms} and
\eqref{eq:DOSS-ancestor-forms}.
\end{theorem}

\begin{proof}
Lemma~\ref{lem:covering-surface} and
Proposition~\ref{prop:global-scalar} give $x,y$ and the critical
points, while Proposition~\ref{prop:normalized-periods} identifies
the odd expansion of $y$. Propositions~\ref{prop:canonical-local-kernel}
and \ref{prop:global-bidifferential} give the pole structure
of $B$ and its prescribed even--even coefficients.
The ancestor expansion then follows from
Theorem~\ref{thm:DOSS}, with the ancestor one-forms
defined by $W=R\Xi$.

\end{proof}

The theorem concerns the ancestor theory at the chosen point $v$; both the spectral curve data and the ancestor correlators may vary with $v$. 
The forms $W_k^i$ are defined by $W=R\Xi$, and are
not assumed to be successive $x$-derivatives of $W_0^i$.


\subsection{Simultaneously normalized periods}
\label{subsec:Dubrovin-superpotential}


We now study the case when the individually normalized local periods
$J_i$ can be realized as the odd parts of a single solution $J$
of the second structure equation. Such a common period underlies
Dubrovin's superpotential construction
\cite{dubrovinGeometry2DTopological1996,
dubrovinPainleveTranscendentsTwoDimensional1999,
dunin-barkowskiDubrovinsSuperpotentialGlobal2019}.


Fix a regular base point in the $x$-plane and a distinguished
system of continuation paths to the canonical values.
Continue the normalized local solutions $J_i$ to this base
point, so that they belong to the same solution space.

The relevant pairing is Dubrovin's flat-pencil pairing \cite{dubrovinGeometry2DTopological1996,dubrovinAlmostDualityFrobenius2004}.
Its expression in our flat frame is
\begin{equation}
    \langle J,K\rangle_x
    =
    \eta\bigl(J,(A-x\operatorname{Id})K\bigr).
    \label{eq:superpotential-pencil-pairing}
\end{equation}
It is nondegenerate at every regular value of $x$, and
its value on solutions of the second structure equation
is independent of $x$.
Our normalization gives
\begin{equation}
    \langle J_i,J_i\rangle_x=-\frac12.
    \label{eq:superpotential-local-norm}
\end{equation}

Near the $i$-th critical point, any solution of the second structure equation admits a
decomposition
\begin{equation}
    J(s_i)=c_iJ_i(s_i)+E_i(s_i^2),
\end{equation}
where $E_i$ is holomorphic and $c_i$ is constant.
The two summands belong to opposite eigenspaces of the local monodromy and are orthogonal for the flat-pencil pairing. 
Consequently,
\begin{equation}
    c_i=-2\langle J,J_i\rangle_x.
\end{equation}
Thus a simultaneously normalized period $J$ is characterized by
\begin{equation}
    \frac{J(s_i)-J(-s_i)}2=J_i(s_i)
    \quad\Longleftrightarrow\quad
    \langle J,J_i\rangle_x=-\frac12,
    \qquad i=1,\ldots,N.
    \label{eq:superpotential-common-period}
\end{equation}
Equivalently, the required common period $J$ satisfies
\[
    (x\operatorname{Id}-A)\partial_xJ
    =\left(\theta-\frac12\operatorname{Id}\right)J,
    \qquad
    J(s_i)=\frac{f_i}{s_i}+O(1),
    \quad i=1,\ldots,N.
\]

Define the Gram matrix by
\begin{equation}
    G_{ij}=-2\langle J_i,J_j\rangle_x.
    \label{eq:superpotential-Gram}
\end{equation}
The factor $-2$ accounts for our normalization
and gives $G_{ii}=1$.
The matrix $G$ is invertible if and only if the simultaneous
normalization determines a unique solution $J$.
In this case,
\begin{equation}
    J=\sum_{i,j=1}^{N}(G^{-1})_{ij}J_j.
    \label{eq:superpotential-Gram-solution}
\end{equation}

If $G$ is singular, a simultaneously normalized solution
may still exist, but its normalization leaves an ambiguity
in the orthogonal complement of the span of the $J_i$.
Such a solution can be used directly in Dubrovin's superpotential construction, as explained in \cite[Remark~3.2]{dunin-barkowskiDubrovinsSuperpotentialGlobal2019}.
This distinction will be essential for odd-dimensional
projective spaces.

A useful sufficient condition for invertibility of $G$ is
\begin{equation}
    \operatorname{Spec}(\theta)
    \cap
    \left\{
        \frac12,\frac32,\frac52,\ldots
    \right\}
    =\varnothing.
    \label{eq:superpotential-nonresonance}
\end{equation}
To see this, a solution orthogonal to all $J_i$ has no
singular part at any canonical value and is fixed by all
local monodromies. It therefore extends to an entire
solution on the $x$-plane. Regular singularity at infinity implies
polynomial growth, so this solution is polynomial.
If it is nonzero and has degree $m\ge0$, comparison of
leading coefficients in the second structure equation
requires $m+\tfrac12$ to be an eigenvalue of $\theta$.
Condition~\eqref{eq:superpotential-nonresonance} excludes
this possibility. Hence the orthogonal complement
vanishes and $G$ is invertible.


The projective space examples illustrate the possible ambiguity in the common period. 
In Subsection~\ref{subsec:Pn-spectral-curve}, we construct a simultaneously normalized period for $\mathbb P^n$ at $Q=1$, $t=0$, for $n\ge2$.
Here Frobenius dimension $N=n+1$, the conformal charge is $d=n$, and
\begin{equation}
    \theta(H^k)
    =
    \left(\frac n2-k\right)H^k.
\end{equation}
We prove that
\begin{equation}
    \operatorname{rank}G
    =
    \begin{cases}
        n+1,& n \text{ even},\\
        n,& n \text{ odd}.
    \end{cases}
    \label{eq:superpotential-Pn-Gram-rank}
\end{equation}
For even $n$, the normalized common period is unique.
For odd $n\ge3$, it exists despite the singular Gram matrix,
and all solutions with the same normalized singular parts
differ by
\begin{equation}
    cH^{(n-1)/2},
    \qquad c\in\mathbb C.
\end{equation}
The explicit construction selects the solution satisfying
\begin{equation}
    J(x)=O\!\left(x^{-(n+1)/2}\right)
\end{equation}
on the chosen branch at infinity.

\subsection{Spectral curves from a common period}
\label{subsec:common-period-spectral-curve}

Fix a solution $J$ of the second structure equation satisfying
the simultaneous normalization
\eqref{eq:superpotential-common-period}.
Choose a simply connected retained surface $\Sigma$
as in Subsection~\ref{subsec:covering-surface}.
After shrinking the branch neighbourhoods and connecting corridors if necessary, $J$ is holomorphic away from the selected critical points and has the prescribed local expansions there.

We construct the scalar function $y^J$, the bidifferential
$B^J$, and the associated ancestor forms from this chosen
period. 

For $d\ne1$, define
\begin{equation}
    y^J
    =
    \frac{2}{1-d}
    \eta\bigl(e,(x\operatorname{Id}-A)J\bigr).
    \label{eq:superpotential-period-primitive}
\end{equation}
This is the expression for Dubrovin's superpotential associated with $J$ in \cite{dubrovinPainleveTranscendentsTwoDimensional1999,dunin-barkowskiDubrovinsSuperpotentialGlobal2019}.
Indeed, the second structure equation gives
\begin{equation}
\begin{aligned}
    \partial_x
    \eta\bigl(e,(x\operatorname{Id}-A)J\bigr)
    =
    \eta\!\left(
        e,\left(\theta+\frac12\operatorname{Id}\right)J
    \right)
    =
    \frac{1-d}{2}\eta(e,J),
\end{aligned}
\end{equation}
where we used $\theta e=\tfrac d2e$.
Therefore
\begin{equation}
    dy^J=\eta(e,J)\,dx.
    \label{eq:superpotential-period-differential}
\end{equation}
The apparent pole of $J$ at $a_i$ is removed by
$x\operatorname{Id}-A$, since $Af_i=u_if_i$.
Thus $y^J$ is holomorphic near every selected critical point.

By \eqref{eq:superpotential-common-period},
the odd part of $y^J$ is
\begin{equation}
\begin{aligned}
    \frac{y^J(s_i)-y^J(-s_i)}2
    =
    Y_i(s_i)
    =
    \sum_{k\ge0}
    [z^k]\bigl(R(z)^{-1}h\bigr)_i
    \frac{s_i^{2k+1}}{(2k+1)!!}.
\end{aligned}
\label{eq:superpotential-odd-expansion}
\end{equation}
It therefore supplies exactly the coefficients
used in the preceding construction.
When $d=1$, formula~\eqref{eq:superpotential-period-primitive}
is unavailable, but the differential
$\eta(e,J)\,dx$ is still holomorphic at the selected branch
points and admits a primitive on a simply connected
retained surface.

This period $J$ also determines a bidifferential. 
For $q_1,q_2\in\Sigma$, write $x_\nu=x(q_\nu)$.
Define 
\begin{equation}
\begin{aligned}
    B^J(q_1,q_2)
    :={}&
    \frac{dx_1\,dx_2}{(x_1-x_2)^2}
    \Bigg\{
    1+\langle J,J\rangle_x
    \\
    &\qquad+
    \eta\!\left(
        J(q_1),
        \left[
            (x_1-x_2)\theta
            +\frac{x_1+x_2}{2}\operatorname{Id}-A
        \right]J(q_2)
    \right)
    \Bigg\}.
\end{aligned}
\label{eq:superpotential-phase-kernel}
\end{equation}
Here $\langle J,J\rangle_x$ denotes the constant
flat-pencil self-pairing of the chosen solution.

\begin{lemma}
\label{lem:common-period-kernel}
The bidifferential
$B^J$ is symmetric and has no poles other than the
normalized diagonal double pole. Its even--even regular
coefficients near every pair of critical points agree
with those of $B^{\mathrm{local}}_{ij}$.
\end{lemma}

\begin{proof}
Symmetry follows from $A^*=A$ and $\theta^*=-\theta$,
as in Proposition~\ref{prop:canonical-local-kernel}.

Put $c=\langle J,J\rangle_x$ and
\[
    M(x_1,x_2)
    =(x_1-x_2)\theta
    +\frac{x_1+x_2}{2}\operatorname{Id}-A.
\]
Near $a_i$, write $x_1=u_i+s^2/2$ and
$x_2=u_i+t^2/2$, and set
\[
    F(s,t)=1+c+\eta\bigl(J(s),M(x_1,x_2)J(t)\bigr).
\]
The simultaneous normalization gives
$J(-s)=J(s)-2J_i(s)$ and
$\langle J,J_i\rangle_x=-1/2$. Consequently,
\[
    F(s,s)=1,\qquad F(s,-s)=0.
\]
The second structure equation also gives
\[
    \partial_tF(s,s)=\partial_tF(s,-s)=0,
    \qquad s\neq0,
\]
by the same calculation as in
\eqref{eq:construction-diagonal-pairing-derivatives}.
Since $stF(s,t)$ is holomorphic, the argument used in the proof of
Proposition~\ref{prop:canonical-local-kernel} yields
\[
    B^J(s,t)-\frac{ds\,dt}{(s-t)^2}
    =\text{holomorphic bidifferential}\quad \text{for } t\to s.
\]

Finally, write
$J(s)=J_i(s)+E_i(s^2)$ and
$J(t)=J_j(t)+E_j(t^2)$, with $E_i,E_j$ holomorphic.
Subtracting $B^{\mathrm{local}}_{ij}$ cancels the term
containing $J_i(s)$ and $J_j(t)$. After multiplication
by $dx_1dx_2=st\,ds\,dt$, every remaining term is odd
in at least one of $s,t$, whereas $M(x_1,x_2)$ and
$(x_1-x_2)^2$ are even in both variables.
The difference is holomorphic by the preceding
argument and Proposition~\ref{prop:canonical-local-kernel},
so its even--even Taylor coefficients vanish.
\end{proof}


\begin{lemma}
For the bidifferential $B^J$, the coefficient forms
$\Xi_k^i$ defined in \eqref{eq:DOSS-coefficient-forms} satisfy the differential recursion
\eqref{eq:B-external-differential-recursion}.
\end{lemma}
We verify this directly from the period construction,
and then identify the ancestor forms
$W_k^i$ defined by \eqref{eq:DOSS-ancestor-forms}.

\begin{proof}
First, differentiating
\eqref{eq:superpotential-phase-kernel} and using the
second structure equation gives
\begin{equation}
    \left(
        \mathcal L_{\partial_{x_1}}^{(1)}
        +
        \mathcal L_{\partial_{x_2}}^{(2)}
    \right)B^J(q_1,q_2)
    =
    -\eta\bigl(dJ(q_1),dJ(q_2)\bigr).
    \label{eq:superpotential-kernel-translation}
\end{equation}
To verify this identity, we apply the period equation in both variables and obtain
\begin{equation}
\begin{aligned}
    &(\partial_{x_1}+\partial_{x_2})
    \eta\bigl(J(q_1),M(x_1,x_2)J(q_2)\bigr)
    \\
    &\qquad=
    -(x_1-x_2)^2
    \eta\bigl(\partial_{x_1}J(q_1),
              \partial_{x_2}J(q_2)\bigr),
\end{aligned}
\end{equation}
where $M(x_1,x_2)=(x_1-x_2)\theta+\frac{x_1+x_2}{2}\operatorname{Id}-A$.
Since $\partial_{x_1}+\partial_{x_2}$ annihilates
$x_1-x_2$ and the constant correction
$1+\langle J,J\rangle_x$, this proves
\eqref{eq:superpotential-kernel-translation}.

Fix a point $q$ away from $a_i$ and expand
the first variable in the coordinate $s_i$.
Recall that
\begin{equation}
    \Xi_k^i(q)
    =
    (2k-1)!!\,[s_i^{2k}]
    \frac{B^J(s_i,q)}{ds_i},
    \qquad k\ge0.
\end{equation}
Because $\partial_x=s_i^{-1}\partial_{s_i}$,
\begin{equation}
    \mathcal L_{\partial_x}
    \bigl(s_i^m\,ds_i\bigr)
    =
    (m-1)s_i^{m-2}\,ds_i.
    \label{eq:superpotential-Lie-monomial}
\end{equation}

The odd part of $J$ is the normalized period $J_i$.
Differentiating its expansion gives
\begin{equation}
    \partial_{s_i}J_i(s_i)
    =
    -\frac{f_i}{s_i^2}
    +
    \sum_{k\ge0}
    \frac{(-1)^{k+1}}{(2k-1)!!}
    f_a(R_{k+1})^a{}_i\,s_i^{2k}.
    \label{eq:superpotential-period-derivative}
\end{equation}
The derivative of the even part of $J$ has only odd
powers of $s_i$, so it does not enter the following
coefficient comparisons.

Comparing the coefficient of $s_i^{-2}ds_i$ in
\eqref{eq:superpotential-kernel-translation} yields
\begin{equation}
    W_0^i(q)=\Xi_0^i(q)
    =
    \left.
    \frac{B^J(s_i,q)}{ds_i}
    \right|_{s_i=0}
    =
    -\eta(f_i,dJ(q)).
    \label{eq:superpotential-leading-form}
\end{equation}
Next, for $k\ge0$, compare the coefficient of
$s_i^{2k}ds_i$ and multiply by $(2k-1)!!$.
Substituting \eqref{eq:superpotential-Lie-monomial} and \eqref{eq:superpotential-period-derivative} into \eqref{eq:superpotential-kernel-translation} gives
\begin{equation}
    \begin{aligned}
        -\mathcal L_{\partial_x}\Xi_k^i
        &=
        \Xi_{k+1}^i
        +
        (-1)^{k+1}
        \eta\bigl(f_a(R_{k+1})^a{}_i,dJ\bigr)
        \\
        &=
        \Xi_{k+1}^i+(-1)^k\sum_{a=1}^{N}(R_{k+1})^a{}_i\,\Xi_0^a.
    \end{aligned}
\label{eq:superpotential-external-recursion}
\end{equation}

For comparison with
\eqref{eq:B-external-differential-recursion},
write $B_{mn}^{ij}$ for the regular coefficients
of $B^J$.
This is precisely the differential recursion
\eqref{eq:B-external-differential-recursion}.
\end{proof}

It remains to pass from $\Xi$ to the ancestor forms.
The ancestor forms admit the explicit expression
\begin{equation}
\begin{aligned}
    W_k^i(q)
    &=
    \bigl(-\mathcal L_{\partial_x}\bigr)^kW_0^i(q)
    \\
    &=
    (-1)^{k+1}
    d_q\,\eta\bigl(f_i,\partial_x^kJ(q)\bigr).
\end{aligned}
\label{eq:superpotential-ancestor-period-forms}
\end{equation}

\begin{theorem}
\label{thm:common-period-spectral-realization}
Let $v$ be a tame semisimple point of a conformal Frobenius manifold with flat unit.
Suppose that a solution $J$ of the second structure
equation satisfies
\begin{equation}
    \langle J,J_i\rangle_x=-\frac12,
    \qquad i=1,\ldots,N.
    \label{eq:common-period-theorem-normalization}
\end{equation}
Then there exists a connected, simply connected retained surface $\Sigma$ over the $x$-plane, with exactly one
simple critical point $a_i$ over each $u_i$, on which
the differential
\begin{equation}
    dy^J=\eta(e,J)\,dx
    \label{eq:common-period-theorem-y}
\end{equation}
has a holomorphic primitive.

Consequently, the Eynard--Orantin topological recursion on
$(\Sigma,x,y^J,B^J)$ realizes the prescribed ancestor
coefficients for $2g-2+n>0$, with the
one-forms $W_k^i$ extracted from $B^J$ as in \eqref{eq:superpotential-ancestor-period-forms}.
\end{theorem}

\begin{remark}
One checks that
\begin{equation}
    \d \left(\frac{dy^J}{dx}\right)
    +
    \sum_{i=1}^{N}
    \operatorname*{Res}_{q'=a_i}
    \left(
        \frac{dy^J}{dx}(q')B^J(\,\cdot\,,q')
    \right)
    =0.
    \label{eq:superpotential-DNOPSS-compatibility}
\end{equation}
It is a version of the compatibility condition proposed in \cite{dunin-barkowskiDubrovinsSuperpotentialGlobal2019}.
\end{remark}


Theorem~\ref{thm:common-period-spectral-realization}
reduces this explicit spectral-curve construction to the choice of a simultaneously normalized common period.
Section~\ref{sec:projective-space} carries out this step for $\mathbb P^n$, $n\ge2$, and applies the formulas above to obtain the corresponding spectral  curve data.


\section{Spectral curves of projective spaces}
\label{sec:projective-space}

We apply the construction of Section~\ref{sec:spectral-construction}
to the quantum cohomology of projective space $\mathbb P^n$ for $n\geq2$.
We first fix the Frobenius data at a small quantum point
and construct the normalized period and its associated
spectral data.
Subsection~\ref{subsec:Pn-uniformization}
realizes these data on an explicit domain in the upper
half-plane and states the spectral curve theorem.
Subsection~\ref{subsec:Pn-quotient-geometry}
then describes the stabilizer of $x$ and its algebraic
quotients.
Section~\ref{sec:Pn-modular-spectral-data}
develops the analytic Schwarz uniformization and the
resulting automorphic and modular expressions.

The spectral curve of $\mathbb P^1$ has been studied extensively; see \cite{norburyGromovWittenInvariants2014,dunin-barkowskiIdentificationGiventalFormula2014,borotLoopEquationsGromovWitten2019}.

\subsection{The Frobenius structure of projective space}
\label{subsec:Pn-Frobenius}

Let $H\in H^2(\mathbb P^n,\mathbb C)$ denote the hyperplane class,
and assume throughout that $n\geq2$. 
Put
\begin{equation}
    m=n+1,\qquad \zeta=e^{2\pi\mathrm i/m}.
    \label{eq:Pn-m-zeta}
\end{equation}
In this section, both flat and canonical indices range from $0$ to $n$.

The genus-zero Gromov--Witten potential of $\mathbb P^n$ defines
a conformal Frobenius manifold near its small quantum locus;
see \cite{dubrovinGeometry2DTopological1996,maninSemisimpleFrobeniusSupermanifolds1997a}.
In the standard flat coordinates $t^0,\ldots,t^n$, the flat frame,
metric and unit are
\begin{equation}
    \partial_\alpha=H^\alpha,\qquad
    \eta_{\alpha\beta}=\delta_{\alpha+\beta,n},\qquad
    e=\partial_0.
    \label{eq:Pn-flat-data}
\end{equation}
The Euler field and conformal charge are
\begin{equation}
    E=m\partial_1+
    \sum_{\alpha=0}^{n}(1-\alpha)t^\alpha\partial_\alpha,
    \qquad d=n.
    \label{eq:Pn-Euler}
\end{equation}

On the small quantum locus, where only the divisor coordinate $t^1$ is retained and the other coordinates are set to zero,
the quantum product satisfies
\begin{equation}
    H^{\circ m}=Qe^{t^1}\,1.
    \label{eq:Pn-small-quantum-product}
\end{equation}
Here $Q$ denotes the Novikov parameter.
We specialize to $Q=1$ and set all flat coordinates $t^\alpha$ to zero.
The tangent Frobenius algebra at this point is
\begin{equation}
    \mathbb C[H]/(H^m-1).
    \label{eq:Pn-quantum-algebra}
\end{equation}
All Frobenius data and period equations below are evaluated at this
fixed point. In particular, the endomorphisms
in \eqref{eq:frobenius-A-theta} become
\begin{equation}
    A=m(H\circ-),\qquad
    \theta H^\alpha=\left(\frac n2-\alpha\right)H^\alpha.
    \label{eq:Pn-A-theta}
\end{equation}
The eigenvalues of $A$ (canonical values at the chosen point) are
\begin{equation}
    u_i=m\zeta^i,\qquad 0\leq i<m.
    \label{eq:Pn-canonical-values}
\end{equation}
They are pairwise distinct, so the chosen point is tame and semisimple.

The corresponding idempotents are
\begin{equation}
    \pi_i=\frac1m\sum_{\alpha=0}^{n}\zeta^{-i\alpha}H^\alpha.
    \label{eq:Pn-idempotents}
\end{equation}
They satisfy
\begin{equation}
    \pi_i\circ\pi_j=\delta_{ij}\pi_i,\qquad
    A\pi_i=u_i\pi_i,\qquad
    \eta(\pi_i,\pi_j)=\frac{\delta_{ij}}{m\zeta^{ni}},
    \label{eq:Pn-idempotent-metric}
\end{equation}
where no summation is intended in these fixed-index identities.
We fix the square roots
\begin{equation}
    h_i=\frac1{\sqrt m}e^{-\pi\mathrm i ni/m},\qquad
    f_i=\frac{\pi_i}{h_i}.
    \label{eq:Pn-normalized-frame}
\end{equation}
Then
\begin{equation}
    \eta(f_i,f_j)=\delta_{ij},\qquad
    e=\sum_{i=0}^{n}h_i f_i.
    \label{eq:Pn-unit-components}
\end{equation}
Regarded as columns in the flat frame, the normalized canonical vectors
are given by
\begin{equation}
    f_i=\Phi^\alpha{}_i\partial_\alpha,\qquad
    \Phi^\alpha{}_i=
    \frac1{\sqrt m}
    \exp\!\left(
        \frac{2\pi\mathrm i\,i}{m}
        \left(\frac n2-\alpha\right)
    \right).
    \label{eq:Pn-Phi}
\end{equation}
The inverse transformation is
\begin{equation}
    (\Phi^{-1})^i{}_\alpha=
    \frac1{\sqrt m}
    \exp\!\left(
        \frac{2\pi\mathrm i\,i}{m}
        \left(\alpha-\frac n2\right)
    \right).
    \label{eq:Pn-Phi-inverse}
\end{equation}

With
\begin{equation*}
    U=\operatorname{diag}(u_0,\ldots,u_n),\qquad
    \Theta=\Phi^{-1}\theta\Phi,
\end{equation*}
the homogeneous $R$-matrix satisfies the following relations, in the conventions of \eqref{eq:frobenius-R-recursion}:
\begin{equation}
    R_0=I,\qquad
    [U,R_{k+1}]=(\Theta-kI)R_k,\qquad
    R(z)R(-z)^T=I.
    \label{eq:Pn-R-recursion}
\end{equation}
We shall construct its spectral realization directly from the
second structure equation.

\subsection{The solution of the second structure  equation}
\label{subsec:Pn-spectral-curve}

When $n=1$, the solution of the second structure equation can be 
found in e.g. \cite{dunin-barkowskiIdentificationGiventalFormula2014}. 
We present the computation for general $n$ in this subsection.

Put
\begin{equation}
    a=\frac{n-1}{2}=\frac{m-2}{2}.
    \label{eq:Pn-a}
\end{equation}
Write the solution in the flat frame,
\begin{equation*}
    J(x)=\sum_{\alpha=0}^{n}j_\alpha(x)H^\alpha.
\end{equation*}
The second structure equation at the chosen point (where all $t=0$ and $Q=1$) 
is equivalent to the cyclic system
\begin{equation}
    \begin{aligned}
        xj_0'-mj_n'&=a j_0,\\
        xj_\alpha'-mj_{\alpha-1}'
        &=(a-\alpha)j_\alpha,
        \qquad 1\leq\alpha\leq n.
    \end{aligned}
    \label{eq:Pn-period-components}
\end{equation}
Primes denote differentiation with respect to $x$.

Since $n\ne1$, the scalar primitive
in \eqref{eq:superpotential-period-primitive} is
\begin{equation}
    y=-\frac1a\eta\bigl(e,(xI-A)J\bigr)
    =\frac{m j_{n-1}-xj_n}{a}.
    \label{eq:Pn-scalar-primitive}
\end{equation}
Differentiating and using the component with index $n$ gives
\begin{equation}
    j_n=y',\qquad
    j_{n-1}=\frac{xy'+ay}{m}.
    \label{eq:Pn-top-components}
\end{equation}
We claim more generally that
\begin{equation}
    \partial_x^{r-1}j_{n-r}
    =\frac1{m^r}(x\partial_x+a)^r y,
    \qquad 1\leq r\leq n.
    \label{eq:Pn-component-recursion}
\end{equation}
The case $r=1$ is \eqref{eq:Pn-top-components}. For the induction step,
the component equation with index $n-r$ reads
\begin{equation*}
    m\partial_xj_{n-r-1}
    =(x\partial_x+a+1-r)j_{n-r},
    \qquad 1\leq r<n,
\end{equation*}
where we used $n=2a+1$. The operator identity
\begin{equation}
    \partial_x^{r-1}(x\partial_x+a+1-r)
    =(x\partial_x+a)\partial_x^{r-1}
    \label{eq:Pn-operator-shift}
\end{equation}
therefore gives
\begin{equation*}
    \begin{aligned}
        m\partial_x^rj_{n-r-1}
        &=(x\partial_x+a)\partial_x^{r-1}j_{n-r}\\
        &=\frac1{m^r}(x\partial_x+a)^{r+1}y.
    \end{aligned}
\end{equation*}
This proves \eqref{eq:Pn-component-recursion}.

The cyclic equation is $(x\partial_x-a)j_0=my''$.
Applying $\partial_x^{n-1}$ and 
substituting \eqref{eq:Pn-component-recursion} with $r=n$ yields
\begin{equation}
    \left[(x\partial_x+a)^m-m^m\partial_x^m\right]y=0.
    \label{eq:Pn-scalar-equation}
\end{equation}
The leading coefficient is $x^m-m^m$, so the finite singularities
are precisely the canonical values $u_i$.

We select the solution without logarithmic terms at infinity.
Substitution of $y=Cx^{-a}\sum_{d\geq0}b_dx^{-md}$, with $b_0=1$,
into \eqref{eq:Pn-scalar-equation} gives
\begin{equation}
    d^m b_d=(a+m(d-1))_m b_{d-1},\qquad
    b_d=\frac{(a)_{md}}{(d!)^m}.
    \label{eq:Pn-series-recurrence}
\end{equation}
Here $(a)_r=\Gamma(a+r)/\Gamma(a)$ is the rising factorial.
The normalization to be verified below selects
\begin{equation}
    y(x)=-(2\pi)^a\Gamma(a)x^{-a}
    \sum_{d\geq0}\frac{(a)_{md}}{(d!)^m}x^{-md}.
    \label{eq:Pn-normalized-period-series}
\end{equation}
Since
\[
    \frac{b_{d+1}}{b_d}
    =\frac{(a+md)_m}{(d+1)^m}\longrightarrow m^m, \qquad d\to \infty,
\]
the power series in $x^{-m}$ appearing in
\eqref{eq:Pn-normalized-period-series} converges for $|x|>m$
and may be differentiated termwise.
Equivalently,
\begin{equation}
    \begin{aligned}
        y(x)={}&-(2\pi)^a\Gamma(a)x^{-a}\\
        &\times{}_mF_{m-1}\!\left(
        \begin{matrix}
            \dfrac am,\dfrac{a+1}{m},\ldots,\dfrac{a+m-1}{m}\\
            1,\ldots,1
        \end{matrix}
        ;\frac{m^m}{x^m}\right).
    \end{aligned}
    \label{eq:Pn-hypergeometric-period}
\end{equation}
We   choose the branch for which $x^{-a}>0$ on $x>m$.


Starting from the scalar solution $y$ in
\eqref{eq:Pn-normalized-period-series}, set $j_n=y'$ and
consider the recursive formulas
\begin{equation}
    j_{\alpha-1}(x)
    =-\frac1m\int_x^\infty
    \left(tj_\alpha'(t)+(\alpha-a)j_\alpha(t)\right)\,\d t,
    \qquad 1\leq\alpha\leq n.
    \label{eq:Pn-vector-reconstruction}
\end{equation}
We now show that these integrals converge and define a solution of the second structure equation with the prescribed scalar primitive.

\begin{lemma}
\label{lemma:Pn-differential-reconstruction}
Let $n\geq2$, and let $y$ be the solution
\eqref{eq:Pn-normalized-period-series}.
The integrals \eqref{eq:Pn-vector-reconstruction} converge
and define the solution
$J=\sum_{\alpha=0}^{n}j_\alpha H^\alpha$ of
\eqref{eq:Pn-period-components} whose
scalar primitive is $y$ and which tends to zero as
$x\to+\infty$.
\end{lemma}

\begin{proof}
The convergent expansion
\eqref{eq:Pn-normalized-period-series} gives, when $x\to+\infty$,
\[
    y=Cx^{-a}+O(x^{-a-m}),
    \qquad C=-(2\pi)^a\Gamma(a).
\]
For $\alpha=n$, the integrand in
\eqref{eq:Pn-vector-reconstruction} is
\[
    xy''+(n-a)y'
    =\partial_x\bigl((x\partial_x+a)y\bigr),
\]
where $n=2a+1$.
Since $(x\partial_x+a)y=O(x^{-a-m})$, the first integral converges
and gives
\[
    j_{n-1}=\frac1m(x\partial_x+a)y=O(x^{-a-m}).
\]
Inductively, termwise integration of the convergent
series shows that all the remaining integrals converge
and that
\[
    j_{n-r}=O(x^{r-1-a-m}),
    \qquad 1\leq r\leq n,
\]
with the corresponding estimates for their derivatives.

By construction, these functions satisfy
\[
    xj_\alpha'-mj_{\alpha-1}'
    =(a-\alpha)j_\alpha,
    \qquad 1\leq\alpha\leq n.
\]
These equations, together with
$j_{n-1}=m^{-1}(x\partial_x+a)y$, imply
\eqref{eq:Pn-component-recursion} by the same induction
as above. It remains to verify the first equation in
\eqref{eq:Pn-period-components}.

Define its difference by
\[
    D:=(x\partial_x-a)j_0-my''.
\]
Using $n-1-a=a$ and
\eqref{eq:Pn-component-recursion}, we obtain
\[
\begin{aligned}
    \partial_x^{n-1}D
    &=(x\partial_x+a)\partial_x^{n-1}j_0
      -m\partial_x^{n+1}y\\
    &=\frac1{m^n}
      \left[(x\partial_x+a)^m-m^m\partial_x^m\right]y
      =0.
\end{aligned}
\]
Thus $D$ is a polynomial of degree at most
$n-2$. On the other hand, the preceding estimates give
\[
    D=O(x^{-a-2})\longrightarrow0
    \qquad (x\to+\infty).
\]
Hence $D=0$, proving that $J$ satisfies the
full system \eqref{eq:Pn-period-components}.
Furthermore,
\[
    \frac{mj_{n-1}-xj_n}{a}=y,
\]
so its scalar primitive is the prescribed function.

\end{proof}

In particular, we obtain
\begin{equation}
    J(x)=O(x^{-m/2}).
    \label{eq:Pn-cusp-decay}
\end{equation}

Indeed, one has the proposition:
\begin{proposition}
\label{prop:Pn-differential-reconstruction}
For the equation \eqref{eq:Pn-scalar-equation}, every component
$j_k$, $0\leq k\leq n$, is a linear differential
polynomial in $y$ with polynomial coefficients in $x$,
involving derivatives of order at most $n$.
\end{proposition}

\begin{proof}
Since $j_n=y'$, it suffices to consider $j_{n-r}$,
where $1\leq r\leq n=m-1$.
Choose polynomials $p(z),q(z)$ such that
\begin{equation}
    (z+a)^r
    =
    p(z)\prod_{h=1}^{r-1}(z+h)
    +q(z)(z+a)^m.
    \label{eq:Pn-polynomial-bezout}
\end{equation}
Such polynomials exist because the greatest common
divisor of
$\prod_{h=1}^{r-1}(z+h)$ and $(z+a)^m$
is either $1$ or $z+a$, and hence divides $(z+a)^r$.
For $r\geq2$, reducing $q$ modulo
$\prod_{h=1}^{r-1}(z+h)$, we may arrange
\begin{equation*}
    \deg q\leq r-2,
    \qquad
    \deg p\leq m-1.
\end{equation*}
For $r=1$, take $p(z)=z+a$ and $q(z)=0$.

We claim that
\begin{equation}
\begin{aligned}
    \tilde{j}_{n-r}
    :=
    \frac{x^{r-1}}{m^r}\,
    p(x\partial_x)y
    +
    m^{m-r}
    q(x\partial_x+1-r)
    \partial_x^{m-r+1}y
\end{aligned}
\label{eq:Pn-finite-differential-reconstruction}
\end{equation}
satisfies the equation \eqref{eq:Pn-component-recursion}.
The right-hand side has polynomial coefficients
in $x$, and its differential order is at most
$n$.

To verify the claim, use the operator identities
\begin{equation*}
    \partial_x^{r-1}x^{r-1}
    =
    \prod_{h=1}^{r-1}(x\partial_x+h),
    \qquad
    \partial_x^{r-1}q(x\partial_x+1-r)
    =
    q(x\partial_x)\partial_x^{r-1}.
\end{equation*}
Applying $\partial_x^{r-1}$ to the right-hand side
of \eqref{eq:Pn-finite-differential-reconstruction},
and then using
\eqref{eq:Pn-scalar-equation}, gives
\begin{equation*}
\begin{aligned}
    &\frac1{m^r}
    p(x\partial_x)
    \prod_{h=1}^{r-1}(x\partial_x+h)y
    +
    m^{m-r}q(x\partial_x)\partial_x^m y
    \\
    &\qquad=
    \frac1{m^r}
    \left[
        p(x\partial_x)
        \prod_{h=1}^{r-1}(x\partial_x+h)
        +
        q(x\partial_x)(x\partial_x+a)^m
    \right]y
    \\
    &\qquad=
    \frac1{m^r}(x\partial_x+a)^r y.
\end{aligned}
\end{equation*}
By \eqref{eq:Pn-component-recursion}, this equals
$\partial_x^{r-1}j_{n-r}$.
Thus, for $r\geq2$, the difference between $\tilde{j}_{n-r}$ and $j_{n-r}$
is a polynomial of degree at most $r-2$. This proves the proposition.
\end{proof}

For the solution $y$ in \eqref{eq:Pn-normalized-period-series}, the decay condition at infinity implies that $j_{n-r}$ in \eqref{eq:Pn-vector-reconstruction} equals the expression $\tilde j_{n-r}$ in \eqref{eq:Pn-finite-differential-reconstruction}.
Indeed, it remains to eliminate the polynomial ambiguity. Equation~\eqref{eq:Pn-polynomial-bezout} implies $p(-a)=0$.
Consequently, applying $p(x\partial_x)$ to \eqref{eq:Pn-normalized-period-series} annihilates its leading $x^{-a}$ term. Both terms on the right-hand side of \eqref{eq:Pn-finite-differential-reconstruction}
are therefore
\begin{equation*}
    O\left(x^{r-1-a-m}\right)
    =
    O\left(x^{-a-2}\right),
    \qquad x\to\infty,
\end{equation*}
on the chosen branch.
Since $j_{n-r}$ also decays at infinity,
the polynomial difference vanishes identically.
This proves $j_{n-r}=\tilde j_{n-r}$, with $\tilde j_{n-r}$ defined by \eqref{eq:Pn-finite-differential-reconstruction}.


For example, when $n=2$, the reconstruction is
\begin{equation}
    \begin{aligned}
        j_2&=y',\\
        j_1&=\frac{xy'+\frac12y}{3},\\
        j_0&=\frac{2(x^3-27)}9y''
             +\frac{4x^2}{9}y'+\frac{x}{18}y.
    \end{aligned}
    \label{eq:P2-vector-from-y}
\end{equation}

When $n=3$, it is
\begin{equation}
    \begin{aligned}
    j_3&=y',\\
    j_2&=\frac{xy'+y}{4},\\
    j_1&=\frac{x^2y'+xy}{16},\\
    j_0&=\frac{(x^4-256)y''+3x^3y'+x^2y}{64}.
\end{aligned}
\label{eq:P3-vector-from-y}
\end{equation}

We next verify the normalization at the canonical values. Near
$u_0=m$, choose
\begin{equation}
    x=m+\frac{s_0^2}{2},
    \label{eq:Pn-s0-coordinate}
\end{equation}
with $s_0>0$ for $x>m$. 

\begin{lemma}
\label{lem:Pn-normalization-at-u0}
For $n\geq2$, the solution $y$ defined by
\eqref{eq:Pn-normalized-period-series} extends holomorphically
to $s_0=0$ in the coordinate \eqref{eq:Pn-s0-coordinate}, and satisfies
\begin{equation}
    \left.\frac{\d y}{\d s_0}\right|_{s_0=0}
    =\frac1{\sqrt m}=h_0.
    \label{eq:Pn-slope-at-u0}
\end{equation}
\end{lemma}

\begin{proof}
The local form of the second structure equation gives
$J(s_0)=O(s_0^{-1})$. Since
$dy=\eta(e,J)\,dx$ and $dx=s_0\,ds_0$, the differential
$dy$ is holomorphic at $s_0=0$, and hence so is $y$.
It remains to determine its linear coefficient.

Write
\begin{equation}
    y(x)=Cx^{-a}F(z),\qquad
    C=-(2\pi)^a\Gamma(a),\qquad
    z=\frac{m^m}{x^m},
    \label{eq:Pn-slope-notation}
\end{equation}
where
\begin{equation}
    F(z)=\sum_{d\geq0}c_dz^d, \qquad c_d=\frac{(a)_{md}}{m^{md}(d!)^m}.
    \label{eq:Pn-hypergeometric-coefficients}
\end{equation}
Stirling's formula gives, as $d\to \infty$,
\begin{equation}
    c_d=Kd^{-3/2}+O(d^{-5/2}),\qquad
    K=\frac{(2\pi)^{(1-m)/2}m^{a-1/2}}{\Gamma(a)}.
    \label{eq:Pn-Stirling-coefficient}
\end{equation}
Compare $F$ with the polylogarithm
\begin{equation*}
    \operatorname{Li}_{3/2}(z)
    =\sum_{d\geq1}\frac{z^d}{d^{3/2}}.
\end{equation*}
Since $d(c_d-Kd^{-3/2})=O(d^{-3/2})$, the derivative of
$F(z)-K\operatorname{Li}_{3/2}(z)$ extends continuously to $z=1$
from the closed unit disk. Consequently, as $z\to1^-$,
\begin{equation}
    \begin{aligned}
        F(z)-F(1)
        ={}&K\bigl(\operatorname{Li}_{3/2}(z)
                  -\operatorname{Li}_{3/2}(1)\bigr)\\
           &+O(1-z).
    \end{aligned}
    \label{eq:Pn-polylog-comparison}
\end{equation}
The standard polylogarithm expansion (see \href{https://dlmf.nist.gov/25.12.E12}{DLMF, Eq.~25.12.12}) is
\begin{equation}
    \operatorname{Li}_{3/2}(z)
    =\operatorname{Li}_{3/2}(1)
     +\Gamma\!\left(-\frac12\right)(-\log z)^{1/2}
     +O(1-z).
    \label{eq:Pn-polylog-expansion}
\end{equation}
Using $\Gamma(-\tfrac12)=-2\sqrt\pi$ and
$-\log z=(1-z)+O((1-z)^2)$, we obtain
\begin{equation}
    F(z)=F(1)-2K\sqrt\pi\sqrt{1-z}+O(1-z).
    \label{eq:Pn-square-root-expansion}
\end{equation}

In the coordinate \eqref{eq:Pn-s0-coordinate},
\begin{equation}
    z=\left(1+\frac{s_0^2}{2m}\right)^{-m}
     =1-\frac{s_0^2}{2}+O(s_0^4),
    \label{eq:Pn-z-s0}
\end{equation}
so, along the chosen positive real branch,
\begin{equation}
    \sqrt{1-z}=\frac{s_0}{\sqrt2}+O(s_0^3),\qquad
    x^{-a}=m^{-a}+O(s_0^2).
    \label{eq:Pn-square-root-coordinate-conversion}
\end{equation}
It follows that
\begin{equation}
    y(s_0)=Cm^{-a}F(1)
           -\sqrt{2\pi}\,Cm^{-a}K\,s_0+O(s_0^2).
    \label{eq:Pn-local-y-expansion}
\end{equation}
The linear coefficient is
\begin{equation}
    \begin{aligned}
        -\sqrt{2\pi}\,Cm^{-a}K
        &=\sqrt{2\pi}(2\pi)^a\Gamma(a)m^{-a}
          \frac{(2\pi)^{(1-m)/2}m^{a-1/2}}{\Gamma(a)}\\
        &=(2\pi)^{a+(2-m)/2}m^{-1/2}\\
        &=m^{-1/2}.
    \end{aligned}
    \label{eq:Pn-normalization-cancellation}
\end{equation}

Substituting \eqref{eq:Pn-normalization-cancellation}
into \eqref{eq:Pn-local-y-expansion} gives
\[
    y(s_0)=y(0)+\frac{s_0}{\sqrt m}+O(s_0^2),
    \qquad s_0\to0^+.
\]
Since $y$ is holomorphic in $s_0$, this determines its
derivative at the origin and proves
\eqref{eq:Pn-slope-at-u0}.
\end{proof}


At the other canonical values, we similarly obtain that
\begin{equation}
    \left.\frac{\d y}{\d s_i}\right|_{s_i=0}
    =\frac1{\sqrt m}e^{-\pi\mathrm i ni/m}=h_i.
    \label{eq:Pn-all-canonical-slopes}
\end{equation}
Since $\d y=\eta(e,J)\,\d x$ and $h_i\ne0$, these slopes imply that
the singular part of $J$ at every selected critical point is $f_i/s_i$.
Uniqueness of the normalized local anti-invariant solution then gives
\begin{equation}
    \frac{J(s_i)-J(-s_i)}2=J_i(s_i).
    \label{eq:Pn-common-period}
\end{equation}
Thus $J$ is a simultaneously normalized common period in the sense of
\eqref{eq:superpotential-common-period}.

We now describe its uniqueness and the associated Gram matrix.
With the pairing
\eqref{eq:superpotential-pencil-pairing}, their Gram matrix is
\begin{equation}
    G_{ij}=-2\langle J_i,J_j\rangle_x
          =-2\eta\bigl(J_i,(A-xI)J_j\bigr).
    \label{eq:Pn-Gram}
\end{equation}

If $J$ and $\widetilde J$ have the same normalized singular parts,
their difference $D=\widetilde J-J$ satisfies
\begin{equation*}
    \langle D,J_i\rangle_x=0\qquad\text{for every }i.
\end{equation*}
Its local monodromy is trivial, so $D$ is single-valued and extends holomorphically
across all finite singularities. Regular singularity at infinity
implies polynomial growth; hence $D$ is a polynomial.

Write
\begin{equation*}
    D(x)=\sum_{r=0}^{d}v_rx^r,\qquad v_d\ne0.
\end{equation*}
Comparing coefficients in the second structure equation gives
\begin{equation}
    \left(rI-\theta+\frac12I\right)v_r
    =(r+1)Av_{r+1},\qquad v_{d+1}=0.
    \label{eq:Pn-polynomial-period-recursion}
\end{equation}
Since $(\theta-\tfrac12I)H^\alpha=(a-\alpha)H^\alpha$, comparison of the leading coefficients gives $d=a-\alpha$ for some $\alpha$.
If $n$ is even, $a$ is a half-integer, so no nonzero polynomial
solution exists.

If $n$ is odd, $a$ is an integer and $v_d=cH^{a-d}$.
For $d\geq1$, the next coefficient equation is
\begin{equation}
    \left((d-1)I-\theta+\frac12I\right)v_{d-1}
    =dm\,cH^{a-d+1}.
    \label{eq:Pn-polynomial-period-obstruction}
\end{equation}
The coefficient of $H^{a-d+1}$ on the left vanishes, whereas that
on the right is nonzero. Thus $d=0$ and $D=cH^a$.
Conversely, $H^a$ is a constant solution because
$(\theta-\tfrac12I)H^a=0$. We have proved that all solutions with the
prescribed singular parts are
\begin{equation}
    \widetilde J=
    \begin{cases}
        J,&n\text{ even},\\
        J+cH^a,\quad c\in\mathbb C,&n\text{ odd}.
    \end{cases}
    \label{eq:Pn-common-period-ambiguity}
\end{equation}

The same argument identifies the orthogonal complement of the local
period span:
\begin{equation}
    \bigl(\operatorname{span}\{J_0,\ldots,J_n\}\bigr)^\perp
    =
    \begin{cases}
        \{0\},&n\text{ even},\\
        \mathbb C H^a,&n\text{ odd}.
    \end{cases}
    \label{eq:Pn-period-span-complement}
\end{equation}
In the odd case, the constant direction is nonisotropic:
\begin{equation}
    \langle H^a,H^a\rangle_x
    =\eta\bigl(H^a,(A-xI)H^a\bigr)=m.
    \label{eq:Pn-constant-period-norm}
\end{equation}
Consequently, the rank of the Gram matrix $G$ is
\begin{equation}
    \operatorname{rank}G=
    \begin{cases}
        n+1,&n\text{ even},\\
        n,&n\text{ odd}.
    \end{cases}
    \label{eq:Pn-Gram-rank}
\end{equation}
The decay condition \eqref{eq:Pn-cusp-decay} removes the constant
ambiguity and selects the period constructed above. Its flat-pencil self-pairing vanishes: $\langle J,J\rangle_x=0$.
Indeed, the pairing is constant, while \eqref{eq:Pn-cusp-decay}
forces its limit at infinity to be zero.

We may now apply Theorem~\ref{thm:common-period-spectral-realization}.
On the retained surface $\Sigma$, with one selected simple critical
point over each $u_i$, write $x_\nu=x(q_\nu)$. The bidifferential
\eqref{eq:superpotential-phase-kernel} takes the form
\begin{equation}
    \begin{aligned}
        B^J(q_1,q_2)={}&\frac{\d x_1\,\d x_2}{(x_1-x_2)^2}
        \Bigg\{1+\eta\!\left(
        J(q_1),
        \left[(x_1-x_2)\theta\right.
        \right.\\[-2pt]
        &\hspace{43mm}\left.\left.
        +\frac{x_1+x_2}{2}I-A\right]J(q_2)
        \right)\Bigg\}.
    \end{aligned}
    \label{eq:Pn-phase-bidifferential}
\end{equation}

The bidifferential admits a direct expression in the flat
components $j_\alpha$. Indeed, using
$\eta(H^\alpha,H^\beta)=\delta_{\alpha+\beta,n}$ and
$\theta H^\beta=(n/2-\beta)H^\beta$, we obtain
\begin{equation*}
\begin{aligned}
    \eta\bigl(J(q_1),J(q_2)\bigr)
    &=
    \sum_{\alpha=0}^{n}
    j_\alpha(q_1)j_{n-\alpha}(q_2),\\
    \eta\bigl(J(q_1),\theta J(q_2)\bigr)
    &=
    \sum_{\alpha=0}^{n}
    \left(\alpha-\frac n2\right)
    j_\alpha(q_1)j_{n-\alpha}(q_2).
\end{aligned}
\end{equation*}
Moreover, $A=m(H\circ-)$ and $H^m=1$ give
\begin{equation*}
    \eta\bigl(J(q_1),AJ(q_2)\bigr)
    =
    m\left[
        \sum_{\alpha=0}^{n-1}
        j_\alpha(q_1)j_{n-1-\alpha}(q_2)
        +j_n(q_1)j_n(q_2)
    \right].
\end{equation*}
Consequently, with $a=(n-1)/2$, formula
\eqref{eq:Pn-phase-bidifferential} becomes
\begin{equation}
\begin{aligned}
    B^J(q_1,q_2)
    ={}&
    \frac{\d x_1\,\d x_2}{(x_1-x_2)^2}
    \Biggl\{
        1+
        \sum_{\alpha=0}^{n}
        \Bigl[
            (\alpha-a)x_1+(a+1-\alpha)x_2
        \Bigr]
        j_\alpha(q_1)j_{n-\alpha}(q_2)
        \\
        &\hspace{17mm}
        -m\sum_{\alpha=0}^{n-1}
        j_\alpha(q_1)j_{n-1-\alpha}(q_2)
        -m j_n(q_1)j_n(q_2)
    \Biggr\}.
\end{aligned}
\label{eq:Pn-phase-bidifferential-components}
\end{equation}

The associated ancestor forms are
\begin{equation}
    W_k^i(q)
    =
    (-1)^{k+1}\d_q\,
    \eta\bigl(f_i,\partial_x^kJ(q)\bigr),
    \qquad 0\leq i\leq n,\quad k\geq0.
    \label{eq:Pn-ancestor-forms}
\end{equation}
In the next subsection we give an explicit domain for this
spectral curve and state the resulting reconstruction theorem.
For $\mathbb P^2$ and $\mathbb P^3$, the components $j_\alpha$
and the bidifferential
\eqref{eq:Pn-phase-bidifferential-components}
will be expressed in modular forms in
Sections~\ref{subsec:P2-modular-spectral-data}
and~\ref{subsec:P3-modular-spectral-data}.

\subsection{Explicit spectral curves for projective spaces}
\label{subsec:Pn-uniformization}

We now realize the period data constructed in
Subsection~\ref{subsec:Pn-spectral-curve}
on an explicit domain in the upper half-plane.
We first introduce the uniformizing coordinate $\tau\in \mathbb{H}$,
then specify the domain and its local involutions,
and finally state the spectral curve theorem.
The uniformization is inspired by Milanov's work on
$\mathbb P^2$ \cite{milanovPeriodMapQuantum2019}.

Equip the $x$-plane with
an orbifold point of order two at each $u_i=m\zeta^i$, and regard
infinity as a cusp. The orbifold Euler characteristic is
\begin{equation*}
  \chi_x=2-1-\sum_{r=0}^{n}(1-\frac{1}{2})=  1-\frac m2<0,
\end{equation*}
so its universal orbifold cover is the complex upper half-plane $\mathbb H$.
The rotation $x\mapsto\zeta x$ preserves this structure. The quotient plane with
coordinate $x^m/m^m$ has an orbifold point of order $m$ at $0$, one
of order two at $1$, and a cusp at infinity. Its uniformizing group
is therefore the Hecke triangle group of type $(2,m,\infty)$; see below.

Set
\begin{equation}
    \ell_m=2\cos\frac\pi m,\qquad
    \rho_m=e^{\pi\mathrm i/m},
    \label{eq:Pn-Hecke-parameters}
\end{equation}
and define the transformations on the upper half-plane $\mathbb{H}$:  
\begin{equation}
    S(\tau)=-\frac1\tau,\qquad
    P(\tau)=\tau+\ell_m, \qquad \tau\in \mathbb{H}.
    \label{eq:Pn-Hecke-generators}
\end{equation}
The orientation-preserving Hecke triangle group of type $(2,m,\infty)$ is
\cite[Section~2.2]{doranAutomorphicFormsTriangle2013}
\begin{equation}
    \Delta_m=\langle S,P\mid S^2=(SP)^m=1\rangle.
    \label{eq:Pn-Hecke-group}
\end{equation}
A fundamental domain is
\begin{equation}
    \mathcal F_m=
    \left\{\tau\in\mathbb H:
        |\operatorname{Re}\tau|\leq\frac{\ell_m}{2},\quad
        |\tau|\geq1\right\}.
    \label{eq:Pn-Hecke-fundamental-domain}
\end{equation}
It consists of two reflected triangles, each with angles
$\pi/2,\pi/m,0$. In particular, $2m$ reflection triangles meet at
$\rho_m$.
Figure~\ref{fig:Hecke-fundamental-domain} illustrates
these configurations for $m=4$.

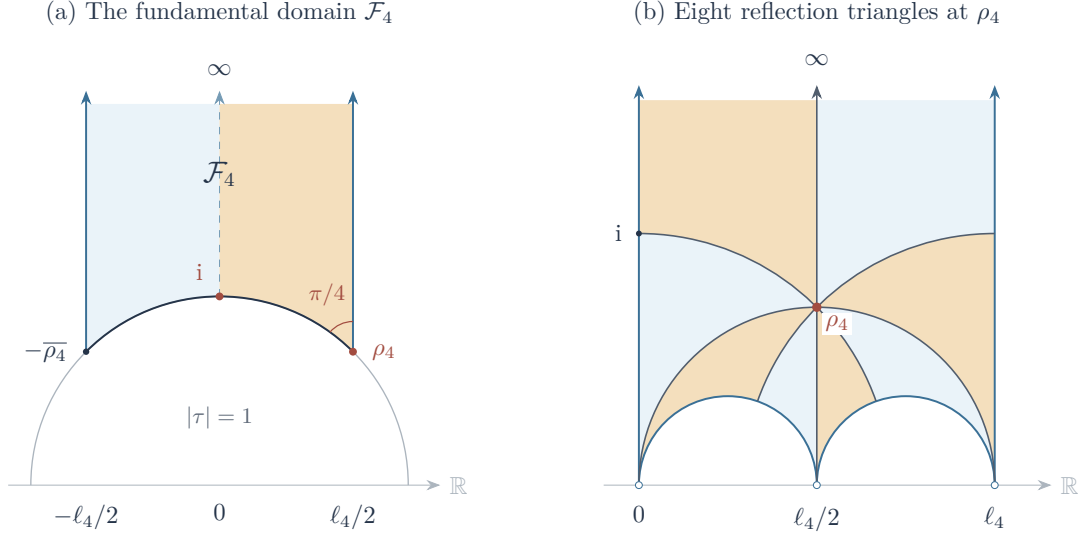
\begin{figure}[htbp]
\centering
\resizebox{\textwidth}{!}{%
\begin{tikzpicture}[
    font=\small,
    text=figInk,
    >={Stealth[length=2mm]}
]
\path[use as bounding box]
    (-.1,.55) rectangle (16.3,8.3);

\node at (3.2,7.95)
    {(a) The fundamental domain $\mathcal F_4$};
\node at (11.75,7.95)
    {(b) Eight reflection triangles at $\rho_4$};

\begin{scope}[
    shift={(3.2,1.2)},
    x=2.7cm,
    y=2.7cm
]
    \path[fill=figPale]
        (0,1)
        arc[start angle=90,end angle=135,radius=1]
        -- ({-sqrt(2)/2},2.02)
        -- (0,2.02) -- cycle;

    \path[fill=figWarm]
        (0,1)
        arc[start angle=90,end angle=45,radius=1]
        -- ({sqrt(2)/2},2.02)
        -- (0,2.02) -- cycle;

    \draw[->,figGuide]
        (-1.12,0) -- (1.17,0)
        node[right] {$\mathbb R$};

    \draw[figGuide,line width=.55pt]
        (1,0)
        arc[start angle=0,end angle=180,radius=1];

    \draw[figInk,line width=.8pt]
        ({sqrt(2)/2},{sqrt(2)/2})
        arc[start angle=45,end angle=135,radius=1];

    \draw[figBlue,line width=.8pt,->]
        ({sqrt(2)/2},{sqrt(2)/2})
        -- ({sqrt(2)/2},2.09);

    \draw[figBlue,line width=.8pt,->]
        ({-sqrt(2)/2},{sqrt(2)/2})
        -- ({-sqrt(2)/2},2.09);

    \draw[figBlue!70,dashed,->]
        (0,1) -- (0,2.09);

    \node[font=\large] at (0,1.65) {$\mathcal F_4$};
    \node at (0,2.20) {$\infty$};

    \fill[figRed] (0,1) circle[radius=1.6pt];
    \node[above left=4pt,text=figRed]
        at (0,1) {$\mathrm i$};

    \fill[figRed]
        ({sqrt(2)/2},{sqrt(2)/2})
        circle[radius=1.6pt];
    \node[right=4pt,text=figRed]
        at ({sqrt(2)/2},{sqrt(2)/2}) {$\rho_4$};

    \fill[figInk]
        ({-sqrt(2)/2},{sqrt(2)/2})
        circle[radius=1.3pt];
    \node[left=4pt]
        at ({-sqrt(2)/2},{sqrt(2)/2})
        {$-\overline{\rho_4}$};

    \draw[figRed,line width=.5pt]
        ({sqrt(2)/2},{sqrt(2)/2+.16})
        arc[start angle=90,end angle=139.59,radius=.16];
    \node[text=figRed,font=\footnotesize]
        at (.57,1.01) {$\pi/4$};

    \node[below=4pt] at ({-sqrt(2)/2},0) {$-\ell_4/2$};
    \node[below=4pt] at (0,0) {$0$};
    \node[below=4pt] at ({sqrt(2)/2},0) {$\ell_4/2$};
    \node[font=\footnotesize,text=figInk!75]
        at (0,.36) {$|\tau|=1$};
\end{scope}

\begin{scope}[
    shift={(9.2,1.2)},
    x=3.6cm,
    y=3.6cm
]
    \pgfmathsetmacro{\hc}{sqrt(2)/2}
    \pgfmathsetmacro{\hs}{sqrt(2)/4}
    \pgfmathsetmacro{\ha}{asin(1/3)}

    \foreach \sgn/\colA/\colB in {1/figWarm/figPale,-1/figPale/figWarm}{
        \begin{scope}[
            shift={(\hc,0)},
            xscale=\sgn,
            shift={(-\hc,0)}
        ]
            \path[fill=\colA]
                (\hc,\hc)
                arc[start angle=45,end angle=90,radius=1]
                -- (0,1.53) -- (\hc,1.53) -- cycle;

            \path[fill=\colB]
                (\hc,\hc)
                arc[start angle=45,end angle=90,radius=1]
                -- (0,0)
                arc[start angle=180,end angle=90,radius=\hc]
                -- cycle;

            \path[fill=\colA]
                (\hc,\hc)
                arc[
                    start angle=135,
                    end angle={180-\ha},
                    radius=1
                ]
                arc[
                    start angle={90-\ha},
                    end angle=180,
                    radius=\hs
                ]
                arc[
                    start angle=180,
                    end angle=90,
                    radius=\hc
                ]
                -- cycle;

            \path[fill=\colB]
                (\hc,\hc)
                arc[
                    start angle=135,
                    end angle={180-\ha},
                    radius=1
                ]
                arc[
                    start angle={90-\ha},
                    end angle=0,
                    radius=\hs
                ]
                -- (\hc,\hc) -- cycle;
        \end{scope}
    }

    \draw[->,figGuide]
        (-.14,0) -- (1.64,0)
        node[right] {$\mathbb R$};

    \draw[figBlue,line width=.8pt,->]
        (0,0) -- (0,1.59);
    \draw[figBlue,line width=.8pt,->]
        ({sqrt(2)},0) -- ({sqrt(2)},1.59);

    \draw[figBlue,line width=.8pt]
        (0,0)
        arc[start angle=180,end angle=0,radius=\hs]
        arc[start angle=180,end angle=0,radius=\hs];

    \draw[figInk!80,line width=.65pt,->]
        (\hc,0) -- (\hc,1.59);

    \draw[figInk!80,line width=.65pt]
        (0,1)
        arc[start angle=90,end angle=\ha,radius=1];

    \draw[figInk!80,line width=.65pt]
        ({sqrt(2)},1)
        arc[start angle=90,end angle={180-\ha},radius=1];

    \draw[figInk!80,line width=.65pt]
        (0,0)
        arc[start angle=180,end angle=0,radius=\hc];

    \fill[figRed] (\hc,\hc) circle[radius=1.9pt];
    \node[
        fill=white,
        inner sep=1.2pt,
        text=figRed,
        below right=3pt
    ] at (\hc,\hc) {$\rho_4$};

    \node at (\hc,1.69) {$\infty$};

    \fill[figInk] (0,1) circle[radius=1.2pt];
    \node[left=3pt] at (0,1) {$\mathrm i$};

    \foreach \xx in {0,\hc,1.414213562}{
        \draw[fill=white,draw=figBlue]
            (\xx,0) circle[radius=1.4pt];
    }

    \node[below=5pt] at (0,0) {$0$};
    \node[below=5pt] at (\hc,0) {$\ell_4/2$};
    \node[below=5pt] at ({sqrt(2)},0) {$\ell_4$};
\end{scope}
\end{tikzpicture}%
}
\caption{
    The Hecke fundamental domain for $m=4$, with
    $\ell_4=\sqrt2$ and $\rho_4=e^{\pi\mathrm i/4}$.
    (a) $\mathcal F_4$ consists of two reflected triangles
    with angles $\pi/2,\pi/4,0$.
    (b) Eight reflection triangles meet at $\rho_4$,
    each with angle $\pi/4$ at this vertex.
}
\label{fig:Hecke-fundamental-domain}
\end{figure}

The compactified quotient of $\mathbb H$ by $\Delta_m$
has genus zero and is therefore biholomorphic to
$\mathbb P^1$.
A coordinate on this quotient, pulled back to $\mathbb H$,
gives a Hauptmodul, unique up to a M\"obius transformation.
We choose the unique normalization
\cite[Section~2.3]
{doranAutomorphicFormsTriangle2013}
\begin{equation}
    \mathcal J_m(\mathrm i)=1,\qquad
    \mathcal J_m(\rho_m)=0,\qquad
    \mathcal J_m(\mathrm i\infty)=\infty.
    \label{eq:Pn-Hauptmodul-normalization}
\end{equation}
When $n=m-1=2$, $\mathcal J_3=j/1728$ is the classical modular $j$-invariant, normalized to take the value $1$ at $\mathrm i$. 

The spectral curve coordinate $x(\tau)$ is determined by
\begin{equation}
    x(\tau)^m=m^m\mathcal J_m(\tau),\qquad x(\mathrm i)=m.
    \label{eq:Pn-Hecke-x}
\end{equation}
All zeros of $\mathcal J_m$ on $\mathbb H$ have order $m$, so the
chosen root $x(\tau)$ defines a single-valued holomorphic function on
$\mathbb H$.

The transformation laws of $x$ are
\begin{equation}
    x(S\tau)=x(\tau),\qquad
    x(P\tau)=\zeta^{-1}x(\tau).
    \label{eq:Pn-x-transformation}
\end{equation}

An analytic description of this uniformization through
the auxiliary Gauss equation
\eqref{eq:Pn-Schwarz-equation}
is given in Subsection~\ref{subsec:Pn-modular-spectral-data}.

We now specify an open domain in the upper half-plane on which
the period data define the spectral curve. Choose a real number $b$ such that
\begin{equation}
    \max\left\{
        \sin\frac{\pi}{m},
        \cos\frac{\pi}{m}
    \right\}<b<1,
    \label{eq:Pn-spectral-domain-height}
\end{equation}
and set
\begin{equation}
    \Sigma=
    \left\{
        \tau\in\mathbb H:
        -\left(m-\frac12\right)\ell_m
        <\operatorname{Re}\tau<\frac{\ell_m}{2},
        \quad
        \operatorname{Im}\tau>b
    \right\}.
    \label{eq:Pn-explicit-spectral-domain}
\end{equation}

Consider the upper parts of the translated fundamental domains
\[
    \mathcal F_m,\quad
    P^{-1}\mathcal F_m,\quad\ldots,\quad
    P^{-(m-1)}\mathcal F_m.
\]
The domain \eqref{eq:Pn-explicit-spectral-domain}
contains these upper regions and extends across the
circular sides
$|\tau+i\ell_m|=1$, $0\leq i<m$.
In particular, the critical points of $x(\tau)$:
\begin{equation}
    a_i=-i\ell_m+\mathrm i,
    \qquad 0\leq i<m,
    \label{eq:Pn-selected-critical-lifts}
\end{equation}
lie in the interior of $\Sigma$.
The lower bound in \eqref{eq:Pn-spectral-domain-height}
is a sufficient condition for the open spectral curve used below.

The domain $\Sigma$ and the images of its selected critical points
are illustrated in Figure~\ref{fig:Pn-selected-domain}
for $\mathbb P^2$.

\begin{figure}[htbp]
\centering
\resizebox{\textwidth}{!}{%
\begin{tikzpicture}[
    font=\small,
    text=figInk,
    >={Stealth[length=2mm]}
]
\path[use as bounding box]
    (-.3,1.0) rectangle (14.5,6.4);

\begin{scope}[
    shift={(6.4,1.2)},
    x=2.4cm,
    y=2.4cm
]
    \shade[top color=white,bottom color=figPale]
        (-2.5,.9) rectangle (.5,2.03);

    \begin{scope}
        \clip (-2.5,.9) rectangle (.5,2.03);
        \foreach \cc in {-2,-1,0}{
            \path[fill=figWarm]
                ({\cc-.5},{sqrt(3)/2})
                arc[start angle=120,end angle=60,radius=1]
                -- cycle;
        }
    \end{scope}

    \foreach \cc in {-1.5,-.5}{
        \draw[figGuide,densely dotted]
            (\cc,{sqrt(3)/2}) -- (\cc,1.86);
    }

    \foreach \cc in {-2,-1,0}{
        \draw[figInk!65,line width=.6pt]
            ({\cc-.5},{sqrt(3)/2})
            arc[start angle=120,end angle=60,radius=1];
    }

    \draw[figBlue,dashed,line width=.7pt]
        (-2.5,.9) -- (.5,.9);
    \draw[figBlue,dashed,line width=.7pt,->]
        (-2.5,.9) -- (-2.5,2.04);
    \draw[figBlue,dashed,line width=.7pt,->]
        (.5,.9) -- (.5,2.04);

    \node[text=figBlue,font=\large]
        at (-1,1.98) {$\Sigma$};

    \node at (-2,1.59) {$P^{-2}\mathcal F_3$};
    \node at (-1,1.59) {$P^{-1}\mathcal F_3$};
    \node at (0,1.59) {$\mathcal F_3$};

    \foreach \cc/\idx in {-2/2,-1/1,0/0}{
        \fill[figRed] (\cc,1) circle[radius=1.8pt];
        \node[above=5pt,text=figRed]
            at (\cc,1) {$a_{\idx}$};
    }

    \node[anchor=east,text=figBlue]
        at (-2.56,.9) {$b$};
    \node[anchor=north] at (-2.5,.79) {$-\frac52$};
    \node[anchor=north] at (.5,.79) {$\frac12$};

    \node[anchor=north,font=\footnotesize]
        at (-1,.76)
        {$a_i=-i+\mathrm i,\qquad b=0.90$};
\end{scope}

\draw[->,figInk,line width=.8pt]
    (8.15,4.45) -- (9.55,4.45)
    node[midway,above=5pt] {$x$};

\begin{scope}[shift={(11.85,4.45)}]
    \node[font=\large] at (0,1.75) {$\mathbb C_x$};
    \draw[figGuide,dashed] (0,0) circle[radius=1.16];

    \fill[figGuide] (0,0) circle[radius=1pt];
    \node[below=2pt,text=figGuide] at (0,0) {$0$};

    \foreach \ang in {0,120,240}{
        \fill[figRed] (\ang:1.16) circle[radius=1.8pt];
    }

    \node[anchor=west,text=figRed]
        at (1.25,0) {$u_0=3$};
    \node[anchor=south,text=figRed]
        at (120:1.30) {$u_1=3\zeta$};
    \node[anchor=north,text=figRed]
        at (240:1.30) {$u_2=3\zeta^2$};

    \node[font=\footnotesize] at (0,-1.83)
        {$\zeta=e^{2\pi\mathrm i/3}$};
\end{scope}

\node at (7.1,1.43) {
    $\displaystyle
    \Sigma=\bigl\{
        \tau\in\mathbb H:
        -\tfrac52<\operatorname{Re}\tau<\tfrac12,\quad
        \operatorname{Im}\tau>b
    \bigr\}$
};
\end{tikzpicture}%
}
\caption{
    The spectral domain $\Sigma$ for $\mathbb P^2$,
    with $b=0.90$.
    The highlighted caps extend across the circular sides
    of the translated fundamental domains, placing the three
    critical points $a_i=-i+\mathrm i$ in the interior of $\Sigma$.
    Their images are $u_i=3\zeta^i$, $0\leq i<3$.
}
\label{fig:Pn-selected-domain}
\end{figure}

The local involution of $x(\tau)$ at $a_i$ is induced by
\begin{equation}
    \sigma_i=P^{-i}SP^i,
    \qquad
    \sigma_i(\tau)
    =
    -i\ell_m-\frac1{\tau+i\ell_m}.
    \label{eq:Pn-explicit-local-involutions}
\end{equation}
It satisfies $\sigma_i(a_i)=a_i$,
$\sigma_i^2=\operatorname{id}$, and
$x\circ\sigma_i=x$.

The following lemma follows from the geometry of the Hecke triangle group. 
\begin{lemma}
\label{lem:Pn-explicit-spectral-domain}
The critical points of $x|_\Sigma$ are precisely the $m$ points $a_i$ in
\eqref{eq:Pn-selected-critical-lifts}.
They are simple and satisfy
\[
    x(a_i)=m\zeta^i=u_i.
\]
If two distinct points $\tau_1,\tau_2\in\Sigma$ satisfy
$x(\tau_1)=x(\tau_2)$, then they satisfy
$\tau_2=\sigma_i(\tau_1)$ for some $0\leq i<m$.
\end{lemma}

The following theorem gives a spectral curve for $\mathbb P^n$ for every $n\ge2$.
\begin{theorem}
\label{thm:Pn-explicit-spectral-curve}
Let $m=n+1\geq 3$, and let $\Sigma$ be the domain
\eqref{eq:Pn-explicit-spectral-domain}.
Let $x$ be the holomorphic function determined by the following relation with the normalized Hauptmodul \eqref{eq:Pn-Hauptmodul-normalization}:
\[
    x(\tau)^m=m^m\mathcal J_m(\tau),
    \qquad x(\mathrm i)=m.
\]
Let $y$ be the analytic continuation to $\Sigma$ of the function \eqref{eq:Pn-normalized-period-series} from the chosen branch near infinity
\begin{equation*}
    y(x)=-(2\pi)^{\frac{n-1}{2}}\Gamma(\frac{n-1}{2})x^{-\frac{n-1}{2}}
    \sum_{d\geq0}\frac{(\frac{n-1}{2})_{md}}{(d!)^m}x^{-md},
\end{equation*}
and let $J=\sum_{\alpha=0}^{n}j_\alpha H^\alpha$ be its period vector, reconstructed by
\eqref{eq:Pn-top-components} and
\eqref{eq:Pn-component-recursion},
with the integration constants fixed by decay at infinity.
Define $B^J$ by
\eqref{eq:Pn-phase-bidifferential}, equivalently by
\eqref{eq:Pn-phase-bidifferential-components}.

Then $(\Sigma,x,y,B^J)$ is a spectral curve with precisely
the $m$ simple critical points
\[
    a_i=-i\ell_m+\mathrm i,
    \qquad x(a_i)=u_i=m\zeta^i,
    \qquad 0\leq i<m.
\]

The Eynard--Orantin topological recursion on these data reproduces
the ancestor coefficients of $\mathbb P^n$ at the chosen
small quantum point, with external ancestor forms
$W_k^i$ given by \eqref{eq:Pn-ancestor-forms}.
\end{theorem}

This completes the explicit spectral curve construction.
We next describe the stabilizer of $x$ and an associated
algebraic compact curves.

\subsection{Stabilizer groups and algebraic compact curves}
\label{subsec:Pn-quotient-geometry}

We now study the geometry of the function $x$
introduced in Subsection~\ref{subsec:Pn-uniformization}.
The transformation laws
\eqref{eq:Pn-x-transformation}
determine its stabilizer in the Hecke triangle group.
An index two subgroup then gives a smooth double cover
of the $x$-plane.

Define the character
\begin{equation}
    \chi_m:\Delta_m\longrightarrow\mathbb C^\times,
    \qquad \chi_m(S)=1,\quad \chi_m(P)=\zeta^{-1}.
    \label{eq:Pn-x-character}
\end{equation}
Its kernel is the stabilizer of $x$:
\begin{equation}
    \Gamma_m=\ker\chi_m
    =\{\gamma\in\operatorname{PSL}_2(\mathbb R):
        x(\gamma\tau)=x(\tau)\}.
    \label{eq:Pn-full-x-stabilizer}
\end{equation}
It has index $m$ in $\Delta_m$.
A fundamental domain for $\Gamma_m$ is obtained by taking
\[
    \bigcup_{i=0}^{m-1}P^{-i}\mathcal F_m
\]
with the induced boundary identifications.

A complete set of generators of $\Gamma_m$ is
\begin{equation}
    \gamma_k=P^kSP^{-k},\qquad
    \gamma_k(\tau)=k\ell_m-\frac1{\tau-k\ell_m},
    \qquad 0\leq k<m.
    \label{eq:Pn-x-stabilizer-generators}
\end{equation}
They give the presentation
\begin{equation}
    \Gamma_m=
    \left\langle\gamma_0,\ldots,\gamma_{m-1}
    \ \middle|\ \gamma_0^2=\cdots=\gamma_{m-1}^2=1\right\rangle.
    \label{eq:Pn-x-stabilizer-presentation}
\end{equation}
I.e., $\Gamma_m$ is the universal Coxeter group of rank $m$.
The fixed point of $\gamma_k$ is $\tau=k\ell_m+\mathrm i$, whose
$x(\tau)$ value is the canonical value $m\zeta^{-k}$.

The underlying Riemann surface of $\Gamma_m\backslash\mathbb H$ is $\mathbb C_x$, and its cusp
compactification is $\mathbb P^1_x$. 
There is also a natural smooth double cover. Let
\begin{equation}
    K_m=\ker\left(
        \Gamma_m\longrightarrow\{\pm1\},\quad
        \gamma_k\longmapsto-1
    \right).
    \label{eq:Pn-torsion-free-subgroup}
\end{equation}
One has the lemma below.
\begin{lemma}
\label{lem:Pn-torsion-free-subgroup}
For $m\geq3$, $K_m$ is the normal subgroup of $\Gamma_m$, with the index two, and is torsion-free. Consequently, $K_m$ acts freely and properly discontinuously on $\mathbb H$, and $K_m\backslash\mathbb H$ is a smooth Riemann surface.
\end{lemma}

\begin{proof}
Let $\varepsilon_m:\Gamma_m\to\{\pm1\}$ denote the
homomorphism in \eqref{eq:Pn-torsion-free-subgroup}.
The presentation \eqref{eq:Pn-x-stabilizer-presentation} shows that it is well defined, since $(-1)^2=1$, and surjective. 
Thus $K_m=\ker\varepsilon_m$ is normal and has index two.

The same presentation identifies $\Gamma_m$ with the
free product of the $m$ order two groups
$\langle\gamma_k\rangle$. By the normal form for free
products, every nonidentity element has a unique reduced
expression $\gamma_{i_1}\cdots\gamma_{i_r}$, in which
consecutive indices are distinct. If $r\geq2$ and
$i_1=i_r$, conjugation by $\gamma_{i_1}$ removes the first
and last letters. Repeating this procedure shows that
every nonidentity element is conjugate to a reduced word
$v=\gamma_{j_1}\cdots\gamma_{j_s}$ such that either $s=1$,
or $s\geq2$ and $j_1\neq j_s$.

In the latter case, no cancellation occurs between
successive copies of $v$, so $v^q$ has reduced length
$qs>0$ for every $q\geq1$. Hence $v$ has infinite order.
It follows that every nontrivial finite-order element
of $\Gamma_m$ is conjugate to some $\gamma_k$. However,
for every $h\in\Gamma_m$,
\[
    \varepsilon_m(h\gamma_k h^{-1})
    =\varepsilon_m(h)(-1)\varepsilon_m(h)^{-1}
    =-1.
\]
Thus no nontrivial finite-order element belongs to
$K_m$, proving that $K_m$ is torsion-free.

Finally, $K_m$ is discrete because it is a subgroup of
the Hecke triangle group $\Delta_m$. Its action on
$\mathbb H$ is therefore properly discontinuous, with
finite point stabilizers. Since $K_m$ is torsion-free,
these stabilizers are trivial. Hence the action is free
and the quotient is a smooth Riemann surface.
\end{proof}

We now identify this quotient explicitly. At every preimage
of a canonical value $u_i$, the function $x-u_i$ has a zero
of order two.
Since
\[
    x^m-m^m=\prod_{i=0}^{m-1}(x-u_i)
\]
has simple roots as a polynomial in $x$, 
it induces a biholomorphism
\begin{equation}
    K_m\backslash\mathbb H
    \xrightarrow{\;\sim\;}
    \{(x,w)\in\mathbb C^2:w^2=x^m-m^m\}.
    \label{eq:Pn-hyperelliptic-quotient}
\end{equation}
This hyperelliptic curve is smooth because the polynomial $x^m-m^m$ has distinct roots.

The inclusions $K_m\subset\Gamma_m\subset\Delta_m$ induce
the following diagram of underlying Riemann surfaces:
\begin{equation}
\label{eq:Pn-quotient-diagram}
\xymatrix@C=1.5em@R=2.5em{
    \mathbb H \ar[r]
    &
    K_m\backslash\mathbb H
        \ar[r]^{2:1}
        \ar[d]_{\cong}
    &
    \Gamma_m\backslash\mathbb H
        \ar[r]^{m:1}
        \ar[d]_{\cong}
    &
    \Delta_m\backslash\mathbb H
        \ar[d]^{\cong}
    \\
    &
    \{(x,w):w^2=x^m-m^m\}
        \ar[r]
    &
    \mathbb C_x
        \ar[r]
    &
    \mathbb C_{\mathcal J_m}.
}
\end{equation}
The lower arrows are given by $(x,w)\mapsto x$ and
$x\mapsto\mathcal J_m=x^m/m^m$, respectively.

The projection to $\mathbb P^1_x$ extends to a degree two map from the smooth compactification. 
It is simply ramified over the $m$ canonical values. 
At infinity, the polynomial $x^m-m^m$ has a pole of order $m$, so its square root has nontrivial local monodromy precisely when $m$ is odd. 
Consequently, there is one ramified point above infinity when $m$ is odd, and two unramified points when $m$ is even. 
The Riemann--Hurwitz formula gives
that the smooth compactification has genus
\begin{equation}
    g=\left\lfloor\frac{m-1}{2}\right\rfloor
     =\left\lfloor\frac n2\right\rfloor.
    \label{eq:Pn-hyperelliptic-genus}
\end{equation}

This curve is an algebraic model for the branched
$x$-cover. The auxiliary coordinate $w$ is distinct from
the period $y$; descent of $y$ and $B^J$ requires their
own transformation laws.

For $\mathbb{P}^2$, the smooth double cover is the elliptic curve
\begin{equation}
    w^2=x^3-27.
    \label{eq:P2-elliptic-quotient}
\end{equation}
On its compactification one may equivalently use the Weierstrass function
\begin{equation}
    x=\wp(z),\qquad w=\frac12\wp'(z),\qquad
    g_2=0,\quad g_3=108.
    \label{eq:P2-Weierstrass-parametrization}
\end{equation}
It also appears in \cite{almeida1DLandauGinzburgSuperpotential2025}.
The scalar function $y$ and the bidifferential $B^J$ do not descend
to this elliptic curve quotient.

We now return to the spectral curve data on $\Sigma$ and derive
their automorphic and modular expressions.

\section{Modular expressions for the spectral curve data}
\label{sec:Pn-modular-spectral-data}

We retain the small quantum point and the period $J$ fixed in Section~\ref{sec:projective-space}. The
uniformization of $x$ by the Hecke triangle group $\Delta_m$, where
$m=n+1$, first gives an expression in automorphic forms valid for
every $n\geq2$. Its arithmetic specializations then give classical
modular formulas for $\mathbb P^2$, $\mathbb P^3$ and $\mathbb P^5$.
For $\mathbb P^2$ and $\mathbb P^3$ we also determine $y$, every
component of $J$, and the bidifferential $B^J$ explicitly.

Throughout this section, $\eta_{\mathrm D}$ denotes the Dedekind eta
function, and $E_2,E_4,E_6$ are the normalized Eisenstein series, with
constant term $1$ at the cusp. The series $E_2$ is quasimodular. All roots are initially chosen at the cusp and continued with the normalization $x(\mathrm i)=m$.

\subsection{Automorphic forms and arithmetic Hecke groups}
\label{subsec:Pn-modular-spectral-data}

Put $q=e^{2\pi\mathrm i\tau/\ell_m}$ and
$D=q\frac{\d}{\d q}=\frac{\ell_m}{2\pi\mathrm i}\frac{\d}{\d\tau}$.
The derivative construction from a triangle-group Hauptmodul $\mathcal{J}_m$ gives
the following automorphic forms:

\begin{equation}
    \mathcal E_{4,m}
    =\frac{(D\mathcal J_m)^2}{\mathcal J_m(\mathcal J_m-1)},
    \qquad
    \mathcal D_m
    =\frac{(D\mathcal J_m)^{2m}}
    {\mathcal J_m^{2m-2}(\mathcal J_m-1)^m}.
    \label{eq:Pn-Hecke-modular-forms}
\end{equation}

This is the specialization of the construction in
\cite[Theorem~2]{doranAutomorphicFormsTriangle2013}.
Since $D\mathcal J_m$ has weight two, the weights of
$\mathcal E_{4,m}$ and $\mathcal D_m$ are $4$ and $4m$, respectively.
At a zero of $\mathcal J_m$, the orders of $\mathcal J_m$ and
$D\mathcal J_m$ are $m$ and $m-1$. Thus $\mathcal E_{4,m}$ has a
zero of order $m-2$, whereas $\mathcal D_m$ is nonzero there.
Above $\mathcal J_m=1$, the corresponding orders are two and one,
so both forms are holomorphic and nonzero. At the cusp,
$\mathcal J_m$ has a simple pole with respect to the coordinate $q$; hence $\mathcal E_{4,m}$ has
constant term $1$, and $\mathcal D_m$ has a zero of order $m-2$.
In particular, $\mathcal D_m$ has no zeros on $\mathbb H$.

The identity
$\mathcal E_{4,m}^{\,m}/\mathcal D_m=\mathcal J_m^{m-2}$ gives

\begin{equation}
    x(\tau)
    =m\frac{\mathcal E_{4,m}(\tau)^{1/(m-2)}}
    {\mathcal D_m(\tau)^{1/[m(m-2)]}}.
    \label{eq:Pn-modular-x-general}
\end{equation}

The root of $\mathcal E_{4,m}$ is holomorphic because its zeros have order $m-2$, and the root of $\mathcal D_m$ exists because
$\mathcal D_m$ never vanishes there. 

We next relate the automorphic forms to the Schwarz
uniformization.
Put
\[
    z=\frac{m^m}{x^m}=\frac1{\mathcal J_m}.
\]
Consider the auxiliary Gauss equation
\begin{equation}
    \left[
        (z\partial_z)^2
        -z\left(z\partial_z+\frac14-\frac1{2m}\right)
          \left(z\partial_z+\frac14+\frac1{2m}\right)
    \right]v=0.
    \label{eq:Pn-Schwarz-equation}
\end{equation}
Its exponent differences at $0,1,\infty$ are
$0,\tfrac12,\tfrac1m$, respectively.
A suitably normalized ratio of its solutions gives
the Schwarz parameter $\tau$  \cite[Section~7.1]{doranAutomorphicFormsTriangle2013}.
We verify this below.

Let $F(z)$ denote the Gauss series solution in the auxiliary equation:
\begin{equation}
    F(z)=\sum_{d\geq0}\frac{(\frac14-\frac1{2m})_d(\frac14+\frac1{2m})_d}{(1)_d^2}z^d
    =\,_2F_1\left(\frac14-\frac1{2m},\frac14+\frac1{2m};1;z\right).
    \label{eq:Pn-Gauss-series}
\end{equation}
Choose another solution around $z=0$ which has the form $G(z)=F(z)\log z+H(z)$ with $H(z)$ holomorphic and $H(0)=0$. The normal form of the Gauss
equation gives
\begin{equation}
    \frac{\d}{\d z}(FG'-F'G)
    =-\frac{1-3z/2}{z(1-z)}(FG'-F'G).
    \label{eq:Pn-Gauss-Wronskian-equation}
\end{equation}
Thus, with this normalization, 
\begin{equation}
    FG'-F'G=\frac1{z\sqrt{1-z}}.
    \label{eq:Pn-Gauss-Wronskian}
\end{equation}

For the Schwarz parameter $\tau$ corresponding to the Hauptmodul
normalized by \eqref{eq:Pn-Hauptmodul-normalization}, one has the lemma below. The construction of the Schwarz parameter is standard \cite[Chapter~V, Section~7]{nehariConformalMapping1975}.

\begin{lemma}
\label{lem:Pn-Schwarz-parameter-normalization}
Let $F$ and $G=F\log z+H$ be as above. The Schwarz parameter corresponding to the Hauptmodul
normalized by \eqref{eq:Pn-Hauptmodul-normalization} is given by
\begin{equation}
    \tau(z)=\frac{\ell_m}{2\pi\mathrm i}\frac{G(z)}{F(z)}+\text{constant}.
    \label{eq:Schwarz parameter}
\end{equation}
\end{lemma}

\begin{proof}
For $(m_1,m_2)=(2,m)$, the hypergeometric
uniformization described in
\cite[Section~2.2]{movasatiShokriIntegrality2014}
identifies
\[
    \widehat q(z)
    :=z\exp\!\left(\frac{H(z)}{F(z)}\right)
\]
with the cusp parameter
$e^{2\pi\mathrm i\tau/\ell_m}$ up to a nonzero
constant factor. Thus, near $z=0$,
\[
    e^{2\pi\mathrm i\tau(z)/\ell_m}
    =\kappa_m\widehat q(z),
    \qquad \kappa_m\neq0.
\]
Since $G=F\log z+H$, taking logarithms with compatible
branches gives
\[
    \tau(z)
    =\frac{\ell_m}{2\pi\mathrm i}\frac{G(z)}{F(z)}
     +B_m.
\]
Choose $\log z$ real on $(0,1)$ and the inverse branch
satisfying $\tau(1)=\mathrm i$. Then
\[
    B_m
    =\mathrm i-\frac{\ell_m}{2\pi\mathrm i}
      \frac{G(1)}{F(1)}.
\]
This proves \eqref{eq:Schwarz parameter}.
\end{proof}

As a result, the Wronskian identity gives
\begin{equation}
Dz=\frac{\ell_m}{2\pi\mathrm{i}}\frac{\d z}{\d \tau}= z\sqrt{1-z}\,F^2.
\end{equation}


The solution $F(z)$ gives a second expression for the weight four form:
\begin{equation}
    \mathcal E_{4,m}(\tau)
    =\left[
    {}_2F_1\!\left(
        \frac14-\frac1{2m},\frac14+\frac1{2m};
        1;\frac1{\mathcal J_m(\tau)}
    \right)\right]^4.
    \label{eq:Pn-Hecke-hypergeometric-form}
\end{equation}
Indeed, by the definition \eqref{eq:Pn-Hecke-modular-forms}, $\mathcal E_{4,m}=(Dz)^2/[z^2(1-z)]=F^4$, 
For this Schwarz map construction, see
\cite[Section~2.3 and Appendix~A]{doranAutomorphicFormsTriangle2013}.

For integers $m\ge3$, the group $\Delta_m$ is arithmetic
exactly for $m=3,4,6$; see
\cite[Section~4, Table~1]{doranAutomorphicFormsTriangle2013}.
The case $m=3$ gives $\Delta_3=\operatorname{PSL}_2(\mathbb Z)$.
For $m=4,6$, put $N=2,3$, respectively, and use $t=\tau/\sqrt N$.
Since $\ell_m=\sqrt N$, the generators in
\eqref{eq:Pn-Hecke-generators} become
\begin{equation}
    t\longmapsto t+1,\qquad
    t\longmapsto-\frac1{Nt}.
    \label{eq:Pn-arithmetic-Hecke-generators}
\end{equation}
These transformations identify the respective groups with
$\Gamma_0(2)^+$ and $\Gamma_0(3)^+$; see \cite[Section~4, Table~1]{doranAutomorphicFormsTriangle2013}. 
Here

\begin{equation}
    \begin{aligned}
    \Gamma_0(N)
    &=\left\{
    \begin{pmatrix}a&b\\c&d\end{pmatrix}
    \in\operatorname{SL}_2(\mathbb Z):c\equiv0\pmod N
    \right\}/\{\pm I\},\\
    \Gamma_0(N)^+
    &=\left\langle\Gamma_0(N),\ t\mapsto-\frac1{Nt}\right\rangle < \operatorname{PSL}_2(\mathbb R) .
    \end{aligned}
    \label{eq:Pn-congruence-Fricke-groups}
\end{equation}

Thus $\Gamma_0(N)$ is the congruence subgroup whose lower-left
matrix entry is divisible by $N$, and the superscript $+$ denotes
adjoining the Fricke involution $t\to -1/Nt$.

\subsection{The projective plane}
\label{subsec:P2-modular-spectral-data}

For $\mathbb P^2$, we have $m=3$ and $\ell_3=1$. The classical
modular invariant satisfies $j(\mathrm i)=1728$,
$j(\rho_3)=0$, and $j(\mathrm i\infty)=\infty$. Therefore
\eqref{eq:Pn-Hauptmodul-normalization} identifies
$\mathcal J_3=j/1728$, and the identity
$j=E_4^3/\eta_{\mathrm D}^{24}$ gives

\begin{equation}
    \mathcal J_3(\tau)
    =\frac{j(\tau)}{1728}
    =\frac{E_4(\tau)^3}{1728\eta_{\mathrm D}(\tau)^{24}},
    \qquad
    x(\tau)=\frac{E_4(\tau)}{4\eta_{\mathrm D}(\tau)^8}.
    \label{eq:P2-classical-modular-x}
\end{equation}

In the normalization of \eqref{eq:Pn-Hecke-modular-forms}, the
corresponding forms are
\begin{equation}
    \mathcal E_{4,3}(\tau)=E_4(\tau),\qquad
    \mathcal D_3(\tau)=1728\eta_{\mathrm D}(\tau)^{24}.
    \label{eq:P2-Hecke-classical-forms}
\end{equation}

The normalized scalar period \eqref{eq:Pn-normalized-period-series}
specializes to

\begin{equation}
    y(x)=-\sqrt2\,\pi x^{-1/2}
    {}_3F_2\!\left(\frac16,\frac12,\frac56;
    1,1;\frac{27}{x^3}\right).
    \label{eq:P2-hypergeometric-modular-period}
\end{equation}

Clausen's identity and the classical Gauss representation of $E_4$
are
\begin{equation}
    \begin{aligned}
    {}_3F_2\!\left(\frac16,\frac12,\frac56;1,1;z\right)
    &={}_2F_1\!\left(\frac1{12},\frac5{12};1;z\right)^2,\\
    {}_2F_1\!\left(\frac1{12},\frac5{12};1;
    \frac{1728}{j(\tau)}\right)
    &=E_4(\tau)^{1/4};
    \end{aligned}
    \label{eq:P2-Clausen-modular-identities}
\end{equation}
see \cite[Eqs.~(7.29)]{maierNonlinearDifferentialEquations2011}. 
Since $27/x^3=1728/j$, substitution into the normalized period gives
\begin{equation}
    y(\tau)=-2\sqrt2\,\pi\eta_{\mathrm D}(\tau)^4.
    \label{eq:P2-classical-modular-y}
\end{equation}
For $q=e^{2\pi\mathrm i\tau}$, one has
$x\sim\frac14q^{-1/3}$ and $y\sim-2\sqrt2\,\pi q^{1/6}$.
Thus $y\sim-\sqrt2\,\pi x^{-1/2}$. 
This expression is consistent with the modular description of the
quantum cohomology period map in
\cite{milanovPeriodMapQuantum2019}.

The period components can now be reconstructed by differentiation.
For $D=(2\pi\mathrm i)^{-1}\partial_\tau$, the Ramanujan identities
and the logarithmic derivative of the eta function are
\begin{equation}
    \begin{aligned}
    D\log\eta_{\mathrm D}&=\frac{E_2}{24},&
    DE_2&=\frac{E_2^2-E_4}{12},\\
    DE_4&=\frac{E_2E_4-E_6}{3},&
    DE_6&=\frac{E_2E_6-E_4^2}{2};
    \end{aligned}
    \label{eq:P2-Ramanujan-identities}
\end{equation}
see \cite[Section~2]{maierNonlinearDifferentialEquations2011}.
In particular,
\begin{equation}
    \frac{\d x(\tau)}{\d\tau}
    =-\frac{\pi\mathrm i E_6(\tau)}
    {6\eta_{\mathrm D}(\tau)^8}.
    \label{eq:P2-modular-dx}
\end{equation}
Using $j_2=\d y/\d x$ and the reconstruction
\eqref{eq:P2-vector-from-y} gives every component explicitly:

\begin{equation}
    \begin{aligned}
    j_2(\tau)
    &=\frac{4\sqrt2\,\pi\eta_{\mathrm D}(\tau)^{12}E_2(\tau)}
    {E_6(\tau)},\\
    j_1(\tau)
    &=-\frac{\sqrt2\,\pi\eta_{\mathrm D}(\tau)^4}{3}
    \left(1-\frac{E_2(\tau)E_4(\tau)}{E_6(\tau)}\right),\\
    j_0(\tau)
    &=-\frac{\sqrt2\,\pi}{72\eta_{\mathrm D}(\tau)^4}
    \left(E_2(\tau)^2+E_4(\tau)
    -\frac{2E_2(\tau)E_4(\tau)^2}{E_6(\tau)}\right).
    \end{aligned}
    \label{eq:P2-modular-period-components}
\end{equation}

Substitute these components and \eqref{eq:P2-modular-dx} into
\eqref{eq:Pn-phase-bidifferential}. 
The result is

\begin{align}
    &B^J(\tau_1,\tau_2)\notag\\
    &=-\frac{4\pi^2}
    {9\bigl(E_4(\tau_1)\eta_{\mathrm D}(\tau_2)^8
    -E_4(\tau_2)\eta_{\mathrm D}(\tau_1)^8\bigr)^2}
    \notag\\
    &\times\Biggl\{
    \eta_{\mathrm D}(\tau_1)^8\eta_{\mathrm D}(\tau_2)^8
    E_6(\tau_1)E_6(\tau_2)
    \notag+\frac{\pi^2}{72}
    \bigl(E_4(\tau_1)\eta_{\mathrm D}(\tau_2)^8
    -E_4(\tau_2)\eta_{\mathrm D}(\tau_1)^8\bigr)
    \notag\\
    &\qquad\times\Biggl[
    \frac{E_6(\tau_1)E_2(\tau_2)\eta_{\mathrm D}(\tau_2)^4}
    {\eta_{\mathrm D}(\tau_1)^4}
    \notag\times
    \Bigl[\eta_{\mathrm D}(\tau_2)^8
    \bigl(E_2(\tau_1)^2+E_4(\tau_1)\bigr)
    -2E_4(\tau_2)\eta_{\mathrm D}(\tau_1)^8\Bigr]
    \notag\\
    &\qquad\quad+
    \frac{E_6(\tau_2)E_2(\tau_1)\eta_{\mathrm D}(\tau_1)^4}
    {\eta_{\mathrm D}(\tau_2)^4}
    \notag\times
    \Bigl[2E_4(\tau_1)\eta_{\mathrm D}(\tau_2)^8
    -\eta_{\mathrm D}(\tau_1)^8
    \bigl(E_2(\tau_2)^2+E_4(\tau_2)\bigr)\Bigr]
    \Biggr]
    \notag\\
    &\quad-\frac{\pi^2}{36}
    \eta_{\mathrm D}(\tau_1)^4\eta_{\mathrm D}(\tau_2)^4
    E_6(\tau_1)E_6(\tau_2)
    \notag\times
    \Bigl[\eta_{\mathrm D}(\tau_2)^8E_2(\tau_1)^2
    +\eta_{\mathrm D}(\tau_1)^8E_2(\tau_2)^2\Bigr]
    \notag\\
    &\quad-\frac{\pi^2E_2(\tau_1)E_2(\tau_2)}
    {36\eta_{\mathrm D}(\tau_1)^4\eta_{\mathrm D}(\tau_2)^4}
    \notag\\
    &\qquad\times\Bigl[
    E_4(\tau_1)^3\eta_{\mathrm D}(\tau_2)^{24}
    +E_4(\tau_2)^3\eta_{\mathrm D}(\tau_1)^{24}
    \notag-2E_4(\tau_1)^2E_4(\tau_2)
    \eta_{\mathrm D}(\tau_1)^8\eta_{\mathrm D}(\tau_2)^{16}
    \notag\\
    &\qquad\quad-2E_4(\tau_1)E_4(\tau_2)^2
    \eta_{\mathrm D}(\tau_1)^{16}\eta_{\mathrm D}(\tau_2)^8
    +3456\eta_{\mathrm D}(\tau_1)^{24}
    \eta_{\mathrm D}(\tau_2)^{24}
    \Bigr]\Biggr\}\,\d\tau_1\d\tau_2.
    \label{eq:P2-full-modular-bidifferential}
\end{align}

Let $\Sigma$ be the domain
\eqref{eq:Pn-explicit-spectral-domain} for $m=3$,
that is,
\[
    \Sigma=
    \left\{
        \tau\in\mathbb H:
        -\frac52<\operatorname{Re}\tau<\frac12,\quad
        \operatorname{Im}\tau>b
    \right\},
    \qquad \frac{\sqrt3}{2}<b<1.
\]

Then, one has the proposition:
\begin{proposition}[Modular spectral curve for $\mathbb P^2$]
\label{prop:P2-modular-spectral-data}
Let $\Sigma$ be the domain defined above. 
The functions
\[ 
    x(\tau)=\frac{E_4(\tau)}{4\eta_{\mathrm D}(\tau)^8},\qquad
     y(\tau)=-2\sqrt2\,\pi\eta_{\mathrm D}(\tau)^4,
\]
together with the bidifferential $B^J$ defined in \eqref{eq:P2-full-modular-bidifferential}, give the spectral curve
$(\Sigma,x(\tau),y(\tau),B^J)$ of $\mathbb P^2$.

The topological recursion with these data reproduces the
ancestor coefficients of $\mathbb P^2$ at the chosen quantum point
for $r\geq1$ and $2g-2+r>0$, with external forms
\eqref{eq:Pn-ancestor-forms}.
\end{proposition}

\subsection{Projective three-space}
\label{subsec:P3-modular-spectral-data}

For $\mathbb P^3$, one has $m=4$ and $\ell_4=\sqrt2$.
We use $t=\tau/\sqrt2$, so the group acting on $t$ is
$\Gamma_0(2)^+$ and $q=e^{2\pi\mathrm i t}$. The quotient
$\bigl(\eta_{\mathrm D}(t)/\eta_{\mathrm D}(2t)\bigr)^{24}$ is
a Hauptmodul for $\Gamma_0(2)$; see
\cite[Section~4.3]{doranAutomorphicFormsTriangle2013}.
The transformation
$\eta_{\mathrm D}(-1/t)=(-\mathrm i t)^{1/2}\eta_{\mathrm D}(t)$
gives
\begin{equation}
    \left(\frac{\eta_{\mathrm D}(-1/(2t))}
    {\eta_{\mathrm D}(-1/t)}\right)^{24}
    =4096\left(\frac{\eta_{\mathrm D}(2t)}
    {\eta_{\mathrm D}(t)}\right)^{24}.
    \label{eq:P3-eta-Fricke-transformation}
\end{equation}
The corresponding Fricke invariant Hauptmodul is

\begin{equation}
    \mathcal J_4(\sqrt2t)
    =\frac1{256}
    \left[
    \left(\frac{\eta_{\mathrm D}(t)}
    {\eta_{\mathrm D}(2t)}\right)^{12}
    +64\left(\frac{\eta_{\mathrm D}(2t)}
    {\eta_{\mathrm D}(t)}\right)^{12}
    \right]^2.
    \label{eq:P3-eta-Hauptmodul}
\end{equation}

The quotient $\bigl(\eta_{\mathrm D}(t)/\eta_{\mathrm D}(2t)\bigr)^{24}$ takes the values $64$ and $-64$
at $t=\mathrm i/\sqrt2$ and $t=(1+\mathrm i)/2$, respectively.
Thus the expression on the right has values $1$ and $0$ there,
and has one simple pole at the cusp of the quotient.
These properties give the normalization
\eqref{eq:Pn-Hauptmodul-normalization}.

The classical level-two identity is

\begin{equation}
    \bigl(2E_2(2t)-E_2(t)\bigr)^2
    =\frac{\eta_{\mathrm D}(t)^{16}}
    {\eta_{\mathrm D}(2t)^8}
    +64\frac{\eta_{\mathrm D}(2t)^{16}}
    {\eta_{\mathrm D}(t)^8}
    \label{eq:P3-level-two-identity}
\end{equation}

It follows from the eta expressions and Eisenstein series identity
in \cite[Eqs.~(3.19)--(3.20)]{alimSpecialPolynomialRings2014}.
Its numerator $2E_2(2t)-E_2(t)$ is modular of weight two for
$\Gamma_0(2)$. Combining the identity with
$x^4=256\mathcal J_4$ and choosing the branch positive on the
imaginary axis gives

\begin{equation}
    x(\sqrt2t)
    =\frac{2E_2(2t)-E_2(t)}
    {\eta_{\mathrm D}(t)^2\eta_{\mathrm D}(2t)^2}.
    \label{eq:P3-classical-modular-x}
\end{equation}

One may also differentiate \eqref{eq:P3-eta-Hauptmodul}, using
\begin{equation}
    \frac1{2\pi\mathrm i}\frac{\d}{\d t}
    \log\left(\frac{\eta_{\mathrm D}(t)}
    {\eta_{\mathrm D}(2t)}\right)^{24}
    =E_2(t)-2E_2(2t),
    \label{eq:P3-eta-quotient-derivative}
\end{equation}
to identify the forms in \eqref{eq:Pn-Hecke-modular-forms}:
\begin{equation}
    \begin{aligned}
    \mathcal E_{4,4}(\sqrt2t)
    &=\bigl(2E_2(2t)-E_2(t)\bigr)^2,\\
    \mathcal D_4(\sqrt2t)
    &=2^{16}\eta_{\mathrm D}(t)^{16}\eta_{\mathrm D}(2t)^{16}.
    \end{aligned}
    \label{eq:P3-Hecke-classical-forms}
\end{equation}
This also recovers \eqref{eq:P3-classical-modular-x} directly from
\eqref{eq:Pn-modular-x-general}.

The normalized scalar period is initially expressed as a ${}_4F_3$ hypergeometric function. A common
upper and lower parameter $1$ cancels, giving

\begin{equation}
    \begin{aligned}
    y(x)
    &=-\frac{2\pi}{x}
    {}_4F_3\!\left(\frac14,\frac12,\frac34,1;
    1,1,1;\frac{256}{x^4}\right)\\
    &=-\frac{2\pi}{x}
    {}_3F_2\!\left(\frac14,\frac12,\frac34;
    1,1;\frac{256}{x^4}\right).
    \end{aligned}
    \label{eq:P3-hypergeometric-modular-period}
\end{equation}

Clausen's identity in this case is
\begin{equation}
    {}_3F_2\!\left(\frac14,\frac12,\frac34;1,1;z\right)
    ={}_2F_1\!\left(\frac18,\frac38;1;z\right)^2;
    \label{eq:P3-Clausen-modular-identity}
\end{equation}
see \cite[Eqs.~(7.28a)--(7.28b)]
{maierNonlinearDifferentialEquations2011}.
Equations~\eqref{eq:Pn-Hecke-hypergeometric-form} and
\eqref{eq:P3-Hecke-classical-forms}, with constant term $1$ at the
cusp, imply
\begin{equation}
    {}_2F_1\!\left(\frac18,\frac38;1;
    \frac1{\mathcal J_4(\sqrt2t)}\right)^2
    =2E_2(2t)-E_2(t).
    \label{eq:P3-Gauss-modular-identity}
\end{equation}
Substitution into \eqref{eq:P3-hypergeometric-modular-period} gives
\begin{equation}
    y(\sqrt2t)=-2\pi\eta_{\mathrm D}(t)^2
    \eta_{\mathrm D}(2t)^2.
    \label{eq:P3-classical-modular-y}
\end{equation}
Indeed, $x\sim q^{-1/4}$ and $y\sim-2\pi q^{1/4}$ as
$q=e^{2\pi\mathrm i t}\to0$, so $y\sim-2\pi/x$, as required.

Differentiation of \eqref{eq:P3-classical-modular-x} gives
\begin{equation}
    \frac{\d x(\sqrt2t)}{\d t}
    =-\frac{\pi\mathrm i\bigl(\eta_{\mathrm D}(t)^{24}
    -64\eta_{\mathrm D}(2t)^{24}\bigr)}
    {2\eta_{\mathrm D}(t)^{10}\eta_{\mathrm D}(2t)^{10}}.
    \label{eq:P3-modular-dx}
\end{equation}
Using $j_3=\d y/\d x$ and the component relations \eqref{eq:P3-vector-from-y}, we obtain
\begin{equation}
    \begin{aligned}
    j_3(\sqrt2t)
    &=\frac{2\pi\eta_{\mathrm D}(t)^{12}\eta_{\mathrm D}(2t)^{12}
    \bigl(E_2(t)+2E_2(2t)\bigr)}
    {3\bigl(\eta_{\mathrm D}(t)^{24}
    -64\eta_{\mathrm D}(2t)^{24}\bigr)},\\
    j_2(\sqrt2t)
    &=\frac{\pi\eta_{\mathrm D}(t)^{10}\eta_{\mathrm D}(2t)^{10}
    \bigl(4E_2(2t)^2-E_2(t)^2\bigr)}
    {6\bigl(\eta_{\mathrm D}(t)^{24}
    -64\eta_{\mathrm D}(2t)^{24}\bigr)}
    -\frac\pi2\eta_{\mathrm D}(t)^2\eta_{\mathrm D}(2t)^2,\\
    j_1(\sqrt2t)
    &=\frac{\pi\eta_{\mathrm D}(t)^8\eta_{\mathrm D}(2t)^8
    \bigl(2E_2(2t)-E_2(t)\bigr)^2
    \bigl(E_2(t)+2E_2(2t)\bigr)}
    {24\bigl(\eta_{\mathrm D}(t)^{24}
    -64\eta_{\mathrm D}(2t)^{24}\bigr)}\\
    &\hspace{8mm}-\frac\pi8\bigl(2E_2(2t)-E_2(t)\bigr),\\
    j_0(\sqrt2t)
    &=\frac{\pi\eta_{\mathrm D}(t)^6\eta_{\mathrm D}(2t)^6
    \bigl(2E_2(2t)-E_2(t)\bigr)^3
    \bigl(E_2(t)+2E_2(2t)\bigr)}
    {96\bigl(\eta_{\mathrm D}(t)^{24}
    -64\eta_{\mathrm D}(2t)^{24}\bigr)}\\
    &\hspace{8mm}-\frac{\pi\bigl(5E_2(t)^2-16E_2(t)E_2(2t)
    +20E_2(2t)^2\bigr)}
    {288\eta_{\mathrm D}(t)^2\eta_{\mathrm D}(2t)^2}.
    \end{aligned}
    \label{eq:P3-modular-period-components}
\end{equation}

Substitution into \eqref{eq:Pn-phase-bidifferential} yields the
following bidifferential. We write it in the coordinates $t_1,t_2$:

\begin{align}
    &B^J(\sqrt2t_1,\sqrt2t_2)\notag\\
    &=-\pi^2\Big/
    \Bigl[4\Bigl(
    \bigl(2E_2(2t_1)-E_2(t_1)\bigr)
    \eta_{\mathrm D}(t_2)^2\eta_{\mathrm D}(2t_2)^2
    \notag\\
    &\hspace{19mm}-\bigl(2E_2(2t_2)-E_2(t_2)\bigr)
    \eta_{\mathrm D}(t_1)^2\eta_{\mathrm D}(2t_1)^2
    \Bigr)^2\Bigr]
    \notag\\
    &\times\Biggl\{
    \frac{\bigl(\eta_{\mathrm D}(t_1)^{24}
    -64\eta_{\mathrm D}(2t_1)^{24}\bigr)
    \bigl(\eta_{\mathrm D}(t_2)^{24}
    -64\eta_{\mathrm D}(2t_2)^{24}\bigr)}
    {\eta_{\mathrm D}(t_1)^6\eta_{\mathrm D}(2t_1)^6
    \eta_{\mathrm D}(t_2)^6\eta_{\mathrm D}(2t_2)^6}
    \notag\\
    &\quad+\frac{\pi^2}{432}
    \Bigl[\bigl(2E_2(2t_1)-E_2(t_1)\bigr)
    \eta_{\mathrm D}(t_2)^2\eta_{\mathrm D}(2t_2)^2
    \notag\\
    &\qquad\quad-\bigl(2E_2(2t_2)-E_2(t_2)\bigr)
    \eta_{\mathrm D}(t_1)^2\eta_{\mathrm D}(2t_1)^2\Bigr]
    \notag\\
    &\qquad\times\Biggl[
    \frac{\bigl(\eta_{\mathrm D}(t_1)^{24}
    -64\eta_{\mathrm D}(2t_1)^{24}\bigr)
    \bigl(E_2(t_2)+2E_2(2t_2)\bigr)}
    {\eta_{\mathrm D}(t_1)^{10}\eta_{\mathrm D}(2t_1)^{10}}
    \notag\\
    &\qquad\quad\times\Bigl[
    \eta_{\mathrm D}(t_2)^4\eta_{\mathrm D}(2t_2)^4
    \bigl(5E_2(t_1)^2-16E_2(t_1)E_2(2t_1)+20E_2(2t_1)^2\bigr)
    \notag\\
    &\qquad\qquad-9\bigl(2E_2(2t_2)-E_2(t_2)\bigr)^2
    \eta_{\mathrm D}(t_1)^4\eta_{\mathrm D}(2t_1)^4\Bigr]
    \notag\\
    &\qquad\quad+
    \frac{\bigl(\eta_{\mathrm D}(t_2)^{24}
    -64\eta_{\mathrm D}(2t_2)^{24}\bigr)
    \bigl(E_2(t_1)+2E_2(2t_1)\bigr)}
    {\eta_{\mathrm D}(t_2)^{10}\eta_{\mathrm D}(2t_2)^{10}}
    \notag\\
    &\qquad\quad\times\Bigl[
    9\bigl(2E_2(2t_1)-E_2(t_1)\bigr)^2
    \eta_{\mathrm D}(t_2)^4\eta_{\mathrm D}(2t_2)^4
    \notag\\
    &\qquad\qquad-\eta_{\mathrm D}(t_1)^4\eta_{\mathrm D}(2t_1)^4
    \bigl(5E_2(t_2)^2-16E_2(t_2)E_2(2t_2)+20E_2(2t_2)^2\bigr)
    \Bigr]\Biggr]
    \notag\\
    &\quad-\frac{\pi^2
    \bigl(\eta_{\mathrm D}(t_1)^{24}
    -64\eta_{\mathrm D}(2t_1)^{24}\bigr)
    \bigl(\eta_{\mathrm D}(t_2)^{24}
    -64\eta_{\mathrm D}(2t_2)^{24}\bigr)}
    {144\eta_{\mathrm D}(t_1)^8\eta_{\mathrm D}(2t_1)^8
    \eta_{\mathrm D}(t_2)^8\eta_{\mathrm D}(2t_2)^8}
    \notag\\
    &\qquad\times\Bigl[
    \eta_{\mathrm D}(t_2)^4\eta_{\mathrm D}(2t_2)^4
    \bigl(5E_2(t_1)^2-16E_2(t_1)E_2(2t_1)+20E_2(2t_1)^2\bigr)
    \notag\\
    &\qquad\quad+\eta_{\mathrm D}(t_1)^4\eta_{\mathrm D}(2t_1)^4
    \bigl(5E_2(t_2)^2-16E_2(t_2)E_2(2t_2)+20E_2(2t_2)^2\bigr)
    \notag\\
    &\qquad\quad-9\bigl(2E_2(2t_1)-E_2(t_1)\bigr)
    \bigl(2E_2(2t_2)-E_2(t_2)\bigr)
    \notag\\
    &\qquad\qquad\times
    \eta_{\mathrm D}(t_1)^2\eta_{\mathrm D}(2t_1)^2
    \eta_{\mathrm D}(t_2)^2\eta_{\mathrm D}(2t_2)^2\Bigr]
    \notag\\
    &\quad-\frac{\pi^2
    \bigl(E_2(t_1)+2E_2(2t_1)\bigr)
    \bigl(E_2(t_2)+2E_2(2t_2)\bigr)}
    {144\eta_{\mathrm D}(t_1)^2\eta_{\mathrm D}(2t_1)^2
    \eta_{\mathrm D}(t_2)^2\eta_{\mathrm D}(2t_2)^2}
    \notag\\
    &\qquad\times\Bigl[
    \bigl(2E_2(2t_1)-E_2(t_1)\bigr)^4
    \eta_{\mathrm D}(t_2)^8\eta_{\mathrm D}(2t_2)^8
    \notag\\
    &\qquad\quad+\bigl(2E_2(2t_2)-E_2(t_2)\bigr)^4
    \eta_{\mathrm D}(t_1)^8\eta_{\mathrm D}(2t_1)^8
    \notag\\
    &\qquad\quad-\bigl(2E_2(2t_1)-E_2(t_1)\bigr)^3
    \bigl(2E_2(2t_2)-E_2(t_2)\bigr)
    \notag\\
    &\qquad\qquad\times
    \eta_{\mathrm D}(t_1)^2\eta_{\mathrm D}(2t_1)^2
    \eta_{\mathrm D}(t_2)^6\eta_{\mathrm D}(2t_2)^6
    \notag\\
    &\qquad\quad-\bigl(2E_2(2t_1)-E_2(t_1)\bigr)^2
    \bigl(2E_2(2t_2)-E_2(t_2)\bigr)^2
    \notag\\
    &\qquad\qquad\times
    \eta_{\mathrm D}(t_1)^4\eta_{\mathrm D}(2t_1)^4
    \eta_{\mathrm D}(t_2)^4\eta_{\mathrm D}(2t_2)^4
    \notag\\
    &\qquad\quad-\bigl(2E_2(2t_1)-E_2(t_1)\bigr)
    \bigl(2E_2(2t_2)-E_2(t_2)\bigr)^3
    \notag\\
    &\qquad\qquad\times
    \eta_{\mathrm D}(t_1)^6\eta_{\mathrm D}(2t_1)^6
    \eta_{\mathrm D}(t_2)^2\eta_{\mathrm D}(2t_2)^2
    \notag\\
    &\qquad\quad+256
    \eta_{\mathrm D}(t_1)^8\eta_{\mathrm D}(2t_1)^8
    \eta_{\mathrm D}(t_2)^8\eta_{\mathrm D}(2t_2)^8
    \Bigr]\Biggr\}\,\d t_1\d t_2.
    \label{eq:P3-full-modular-bidifferential}
\end{align}

For $\mathbb P^3$, using $t=\tau/\sqrt2$, the same
construction gives the domain
\[
    \Sigma=
    \left\{
        t\in\mathbb H:
        -\frac72<\operatorname{Re}t<\frac12,\quad
        \operatorname{Im}t>\frac b{\sqrt2}
    \right\},
    \qquad \frac1{\sqrt2}<b<1.
\]

\begin{proposition}[Modular spectral curve for $\mathbb P^3$]
\label{prop:P3-modular-spectral-data}
Let $\Sigma$ be the domain defined above.
In the modular coordinate
$t=\tau/\sqrt2$, the functions
\[ 
    x(\sqrt2t)
    =\frac{2E_2(2t)-E_2(t)}
    {\eta_{\mathrm D}(t)^2\eta_{\mathrm D}(2t)^2},\qquad
      y(\sqrt2t)=-2\pi\eta_{\mathrm D}(t)^2
    \eta_{\mathrm D}(2t)^2,
\]
together with the bidifferential $B^J$ defined in \eqref{eq:P3-full-modular-bidifferential}, give the spectral curve
$(\Sigma,x(\tau),y(\tau),B^J)$ of $\mathbb P^3$.

The topological recursion with these data reproduces the
ancestor coefficients of $\mathbb P^3$ at the chosen quantum point
for $r\geq1$ and $2g-2+r>0$, with external forms
\eqref{eq:Pn-ancestor-forms}.
\end{proposition}

\subsection{Projective five-space}
\label{subsec:P5-modular-spectral-coordinate}

For $\mathbb P^5$, we have $m=6$ and $\ell_6=\sqrt3$.
The coordinate $t=\tau/\sqrt3$ identifies the Hecke group with
$\Gamma_0(3)^+$. The classical eta quotient
$\bigl(\eta_{\mathrm D}(t)/\eta_{\mathrm D}(3t)\bigr)^{12}$
is a Hauptmodul for $\Gamma_0(3)$; see
\cite[Section~4.3]{doranAutomorphicFormsTriangle2013}.
Its Fricke transformation is
\begin{equation}
    \left(\frac{\eta_{\mathrm D}(-1/(3t))}
    {\eta_{\mathrm D}(-1/t)}\right)^{12}
    =729\left(\frac{\eta_{\mathrm D}(3t)}
    {\eta_{\mathrm D}(t)}\right)^{12}.
    \label{eq:P5-eta-Fricke-transformation}
\end{equation}
The values of this Hauptmodul at $t=\mathrm i/\sqrt3$ and
$t=\rho_6/\sqrt3$ are $27$ and $-27$, respectively. The same
Fricke invariance and normalization argument as above therefore gives
\begin{equation}
    \mathcal J_6(\sqrt3t)
    =\frac1{108}\left[
    \left(\frac{\eta_{\mathrm D}(t)}
    {\eta_{\mathrm D}(3t)}\right)^6
    +27\left(\frac{\eta_{\mathrm D}(3t)}
    {\eta_{\mathrm D}(t)}\right)^6
    \right]^2.
    \label{eq:P5-eta-Hauptmodul}
\end{equation}

The sixth root needed for $x$ is naturally expressed using the
theta series of the $A_2$ lattice:

\begin{equation}
    \vartheta_{A_2}(t)
    :=\sum_{r,s\in\mathbb Z}e^{2\pi\mathrm i t(r^2+rs+s^2)}.
    \label{eq:P5-A2-theta}
\end{equation}

This is a modular form of weight one for $\Gamma_0(3)$ with
quadratic character $\left(\frac{-3}{d}\right)$ and constant term
$1$ at the cusp; see
\cite[Propositions~3.2--3.3 and Eq.~(A.1a)]
{maierNonlinearDifferentialEquations2011}.
The classical Eisenstein and eta product identities are

\begin{equation}
    \vartheta_{A_2}(t)^2
    =\frac{3E_2(3t)-E_2(t)}2,
    \qquad
    \vartheta_{A_2}(t)^3
    =\frac{\eta_{\mathrm D}(t)^9}{\eta_{\mathrm D}(3t)^3}
    +27\frac{\eta_{\mathrm D}(3t)^9}{\eta_{\mathrm D}(t)^3}.
    \label{eq:P5-theta-identities}
\end{equation}

See \cite[Eqs.~(3.19)--(3.20) and (A.26)]
{alimSpecialPolynomialRings2014}.
Dividing the cubic identity by
$\eta_{\mathrm D}(t)^3\eta_{\mathrm D}(3t)^3$ gives
\begin{equation}
    \frac{\vartheta_{A_2}(t)^3}
    {\eta_{\mathrm D}(t)^3\eta_{\mathrm D}(3t)^3}
    =\left(\frac{\eta_{\mathrm D}(t)}
    {\eta_{\mathrm D}(3t)}\right)^6
    +27\left(\frac{\eta_{\mathrm D}(3t)}
    {\eta_{\mathrm D}(t)}\right)^6.
    \label{eq:P5-theta-Hauptmodul-bridge}
\end{equation}
Squaring and using \eqref{eq:P5-eta-Hauptmodul} now yields

\begin{equation}
    \mathcal J_6(\sqrt3t)
    =\frac{\vartheta_{A_2}(t)^6}
    {108\eta_{\mathrm D}(t)^6\eta_{\mathrm D}(3t)^6},
    \qquad
    x(\sqrt3t)
    =\frac6{108^{1/6}}
    \frac{\vartheta_{A_2}(t)}
    {\eta_{\mathrm D}(t)\eta_{\mathrm D}(3t)}.
    \label{eq:P5-classical-modular-x}
\end{equation}

The chosen root has $x(\mathrm i)=6$. For
$q=e^{2\pi\mathrm i t}$, the cusp expansions
$\vartheta_{A_2}(t)=1+6q+O(q^2)$ and
$\eta_{\mathrm D}(t)\eta_{\mathrm D}(3t)
=q^{1/6}(1+O(q))$ give
\begin{equation}
    \mathcal J_6(\sqrt3t)\sim\frac1{108q},\qquad
    x(\sqrt3t)\sim\frac6{108^{1/6}}q^{-1/6}.
    \label{eq:P5-modular-cusp-normalization}
\end{equation}

The corresponding forms in the general construction are
\begin{equation}
    \begin{aligned}
    \mathcal E_{4,6}(\sqrt3t)
    &=\vartheta_{A_2}(t)^4
    =\frac{\bigl(3E_2(3t)-E_2(t)\bigr)^2}{4},\\
    \mathcal D_6(\sqrt3t)
    &=108^4\eta_{\mathrm D}(t)^{24}
    \eta_{\mathrm D}(3t)^{24}.
    \end{aligned}
    \label{eq:P5-Hecke-classical-forms}
\end{equation}
The transformation
$\vartheta_{A_2}(-1/(3t))=-\mathrm i\sqrt3\,t\,
\vartheta_{A_2}(t)$, as in
\cite[Eq.~(3.28)]{alimSpecialPolynomialRings2014}, shows that
$\vartheta_{A_2}^4$ has weight four for $\Gamma_0(3)^+$.
Its constant term is $1$, and the corresponding space has dimension
one by
\cite[Theorem~2]{doranAutomorphicFormsTriangle2013}.
The second identity follows from
$\mathcal D_6=\mathcal E_{4,6}^{\,6}/\mathcal J_6^4$.

For $\mathbb{P}^5$, we do not derive a modular expression for $y(\tau)$ here.

\section{Conclusions and outlook}
\label{sec:conclusions-outlook}

The main contribution of this paper is a constructive
realization of semisimple conformal ancestor theories by
Eynard--Orantin topological recursion on connected open
Riemann surfaces.
The DOSS correspondence
\cite{dunin-barkowskiIdentificationGiventalFormula2014}
identifies the local spectral curve data required by Givental reconstruction.
We address the additional analytic problem of realizing these data on a connected surface.
At a fixed tame semisimple point of a conformal Frobenius manifold with flat unit, our construction produces $(\Sigma,x,y,B)$ with the prescribed critical values and a bidifferential having only the normalized diagonal double pole.
The resulting recursion realizes the stable ancestor
correlators with at least one marked point.

Our construction uses Dubrovin's superpotential to realize the spectral curve data. 
The compact curve correspondence of DNOPS \cite{dunin-barkowskiDubrovinsSuperpotentialGlobal2019} uses a normalized Bergman kernel subject to compatibility requirements.
Here the bidifferential is constructed directly to realize the prescribed $R$-matrix data on an open surface.
The general existence construction uses individually
normalized local periods and does not require a
simultaneously normalized common period.
When such a common period $J$ is available, we obtain
explicit formulas for $y^J$, $B^J$, and the ancestor forms.

The projective-space examples make these constructions explicit.
For every $\mathbb P^n$ with $n\ge2$, at the small quantum point $Q=1$, $t=0$, we select a hypergeometric scalar period, reconstruct its vector period, and verify the required normalizations at all canonical values.
Together with the explicit bidifferential and ancestor forms, this gives a uniform family of connected spectral curve realizations for higher-dimensional projective spaces.

For $\mathbb P^2$ and $\mathbb P^3$, we further express the full spectral curve data in classical modular and
quasimodular forms.
In particular, the formulas include the bidifferential
$B^J$, providing the complete input for topological
recursion on the specified open domains.
For $\mathbb P^2$, these results build on the modular
period-map description of
\cite{milanovPeriodMapQuantum2019}.



A first direction is to extend this construction to generalized Frobenius manifolds \cite{liuGeneralizedFrobeniusManifolds2025}, in particular to structures without an Euler vector field.
In the present approach, the Euler field determines multiplication by $E$ and the grading operator, and supplies the homogeneity condition used to fix the $R$-matrix and the period equations. 
Without this field, these ingredients require a replacement. 

For a prescribed $R$-calibrated semisimple theory, the local correspondence with topological recursion remains a starting point. 
The problem is to find suitable period data and a connected spectral curve realization compatible with that calibration. 
If the generalized structure also has a non-flat unit, the role of the unit in the reconstruction and in the choice of $y$ must be reconsidered as well.

A closely related question concerns the dependence on the chosen point of the Frobenius manifold. 
Our construction is made at a fixed point. 
It would be desirable to organize this choice
into a family over the semisimple locus, initially over sufficiently small parameter domains.
Such a family should transport the periods according to the flat connection and reproduce the variation of the ancestor theory under the Givental formalism.

The two questions are strongly connected. In the conformal setting, homogeneity makes it possible to recover the relevant reconstruction data from information at a single point. Without an Euler field, the parameter dependence and the choice of calibration become part of the input needed to determine these data. 

The projective-space examples also raise a broader question about the
geometry of spectral curves in Gromov--Witten theory. Their modular
descriptions, particularly for $\mathbb P^2$ and $\mathbb P^3$, suggest
that compact Riemann surfaces equipped with meromorphic spectral
functions may be too restrictive to provide B-models for general
targets through ordinary Eynard--Orantin topological recursion. 
A compact algebraic description of the $x$-function need not accommodate the full data.
For example, in the $\mathbb P^2$ construction, neither
$y$ nor $B^J$ descends to the elliptic quotient described in
Subsection~\ref{subsec:Pn-quotient-geometry}.
Although this obstruction concerns the specific spectral curve data constructed here; it does not rule out a possible compact spectral curve realization of the same theory.

From this perspective, Dubrovin's superpotential is a natural candidate for the spectral curve function, with its period equations supplying the remaining data. The resulting relation between $x$ and $y$ is naturally expressed by a higher-order differential equation. For instance, the scalar period
for $\mathbb P^n$ satisfies
\begin{equation*}
    \left[
        \left(x\partial_x+\frac{n-1}{2}\right)^{n+1}
        -(n+1)^{n+1}\partial_x^{n+1}
    \right]y=0.
\end{equation*}
There is no general reason for the normalized solution to satisfy a polynomial relation in $x$ and $y$. 
These examples thus motivate a broader view of spectral curve B-models. 

Finally, the construction admits a more physical interpretation, which we will develop in subsequent work. 

\appendix

\section{The non-tame semisimple case}
\label{sec:non-tame}

Let $v$ be a semisimple point at which some canonical values
$u_i(v)$ coincide. 
We explain the modifications needed in Section~\ref{sec:spectral-construction}; throughout, the ramification points of the function $x$ remain distinct and simple.

\medskip
\noindent\textit{The $R$ matrix.}
At a coincident value, one cannot determine an off-diagonal
entry of $R_{k+1}$ by dividing
\eqref{eq:frobenius-R-recursion} by $u_i-u_j$.
Instead, use the full parameter equation
\eqref{eq:frobenius-R-parameter-equation}. Its coefficients give
\begin{equation}
    [\mathrm dU,R_m]
    =\mathrm dR_{m-1}
    +\Phi^{-1}(\nabla\Phi)R_{m-1},
    \qquad m\ge1.
    \label{eq:non-tame-R-parameter}
\end{equation}
On the dense tame locus, evaluation on $\pi_i=\partial_{u_i}$,
together with the diagonal part of
\eqref{eq:frobenius-R-recursion} at the next order, yields
\begin{equation}
\begin{aligned}
    (R_m)^i{}_j
    &=\left(
        \partial_{u_i}R_{m-1}
        +\bigl(\Phi^{-1}\nabla_{\pi_i}\Phi\bigr)R_{m-1}
      \right)^i{}_j,
      &&i\ne j,\\
    (R_m)^i{}_i
    &=\frac1m\sum_{a\ne i}\Theta^i{}_a(R_m)^a{}_i,
      &&m\ge1.
\end{aligned}
\label{eq:non-tame-R-extension}
\end{equation}
Starting from $R_0=\operatorname{Id}$, the first line extends all
off-diagonal entries holomorphically, and the second then extends
the diagonal entries. Induction therefore gives a unique
holomorphic extension of the $R$ matrix from the
tame locus. All its equations, including homogeneity and
$R(z)R(-z)^T=\operatorname{Id}$, extend coefficientwise by the
identity theorem. This concerns holomorphic dependence on the
Frobenius parameters, not convergence of $R$ in $z$.
In particular, the degree-zero commutator relation gives
\begin{equation}
    \Theta^i{}_j=(u_i-u_j)(R_1)^i{}_j,
    \qquad
    \Theta^i{}_j=0\quad\text{if }u_i=u_j.
    \label{eq:non-tame-grading-block}
\end{equation}

\medskip
\noindent\textit{A preliminary covering.}
We now evaluate the data at $v$.
The planar construction of Section~\ref{subsec:covering-surface}
cannot start from disjoint discs centred at coincident $u_i$.
Choose instead pairwise distinct auxiliary points
$w_1,\ldots,w_N\in\mathbb C$. Hermite interpolation gives a
polynomial $P$ of degree at most $2N-1$ such that
\begin{equation}
    P(w_i)=u_i,
    \qquad P'(w_i)=1,
    \qquad i=1,\ldots,N.
    \label{eq:non-tame-Hermite}
\end{equation}
Choose small discs about the $w_i$ and join them by narrow
corridors along a tree avoiding the finite zero set of $P'$.
Apply the square-root attachment of
Lemma~\ref{lem:covering-surface} in the $w$-plane, and set
\begin{equation}
    w=w_i+\frac{r_i^2}{2},
    \qquad x=P(w),
    \qquad
    s_i(r_i)=\sqrt{2\bigl(P(w_i+r_i^2/2)-u_i\bigr)}
             =r_i+O(r_i^3).
    \label{eq:non-tame-cover}
\end{equation}
The square root is chosen with $s_i'(0)=1$.
Since $P'$ does not vanish on the retained $w$-domain,
$dx=P'(w)\,dw$ has precisely the $N$ prescribed simple zeros.
Their values under $x$ may coincide, but their coordinate discs
are disjoint on the surface. 
The surface is connected, simply connected and biholomorphic to the upper half-plane. 

\medskip
\noindent\textit{Local periods.}
At a common value $u_i$, the residue of the second structure
equation in the normalized frame has no eigenvalues other than
$-1/2$ and $0$:
\begin{equation*}
    \partial_x\Phi^{-1} J
    =(x\operatorname{Id}-U)^{-1}\left(\Theta-\frac12\operatorname{Id}\right) \Phi^{-1} J.
\end{equation*}
Indeed, assume that $u_1=\cdots=u_r=u$ and $u_k\ne u$ for $k>r$. Then
\begin{equation*}
    (x\operatorname{Id}-U)^{-1}\left(\Theta-\frac12\operatorname{Id}\right)=\frac{1}{x-u}
    \begin{pmatrix}
    -\frac12I_r&C\\
    0&0
    \end{pmatrix} + D
\end{equation*}
for some matrices $C$ and $D$ that are holomorphic near $x=u$.
By \eqref{eq:non-tame-grading-block}, its $-1/2$-eigenspace is
exactly the span of the coordinate vectors labelled by that
value. Regular-singular theory therefore still gives a unique
odd solution $J_i=f_i/s_i+O(s_i)$ for each specified leading
vector $f_i$.   
Proposition~\ref{prop:normalized-periods} and its
coefficient comparison apply without change. 

One can also check that Proposition~\ref{prop:canonical-local-kernel}, including the
even--even coefficient identity
\eqref{eq:construction-local-edge-identity}, remains valid.

\medskip
\noindent\textit{Globalization and topological recursion.}
Carry out the Cauchy constructions of
Section~\ref{subsec:global-spectral-data} on disjoint slits in
the $w$-plane. Use the scalar germs $Y_i(s_i(r_i))$ and pull back
the complete local bidifferentials to the $r_i$-coordinates,
including their Jacobian factors, before forming the discrepancy
with $B_0$. 
Take smaller full branch discs and connecting corridors so that the outer integration endpoints lie outside the final surface. 
Because the changes $s_i(r_i)$ are odd,
the anti-invariant projections commute with these pullbacks.
The resulting $y$ has the prescribed odd parts in the actual
$s_i$-coordinates, and $B$ has the prescribed even--even regular
coefficients and no pole away from its normalized diagonal pole.

Finally, the recursion takes residues at the distinct points
$a_i$. 
The conclusion of Theorem~\ref{thm:connected-spectral-realization} therefore holds
also at $v$, for $2g-2+n>0$, with the external forms
defined from the actual global $B$ by $W=R\Xi$ as in the main text.

\end{document}